\documentclass[11pt]{article}
\usepackage[utf8]{inputenc}
\usepackage[T1]{fontenc}
\usepackage[letterpaper,top=1in,bottom=1in,left=1in,right=1in,marginparwidth=1.75cm]{geometry}
\usepackage[nottoc,numbib]{tocbibind}
\usepackage{graphicx}
\usepackage{authblk}
\usepackage{complexity}
\usepackage{bbm}	
\usepackage{amsmath}
\usepackage[ruled,vlined,linesnumbered,noresetcount]{algorithm2e}
\usepackage{amsfonts}
\usepackage{amsmath}
\usepackage[colorlinks=true, linkcolor=red, urlcolor=blue, citecolor=gray]{hyperref}
\usepackage{amssymb}

\usepackage{amsthm}
\usepackage{mathtools}
\usepackage{capt-of}
\theoremstyle{plain}
\usepackage{bbm}
\usepackage{bm}
\usepackage{xcolor}
\usepackage[shortlabels]{enumitem} 
\usepackage{comment}
\usepackage{tikz}   
\usepackage{xcolor}
\usepackage{cite}
\usetikzlibrary{decorations.pathreplacing}
\usepackage[title]{appendix}
\usepackage[framemethod=tikz]{mdframed}

\usepackage{titlesec}
\usepackage[misc]{ifsym}
\usepackage{amssymb}
\usepackage{amsmath}
\usepackage{algorithmic}
\usetikzlibrary{patterns}
\usepackage{commath}
\usepackage[capitalise]{cleveref}
\DeclareMathOperator*{\argmax}{arg\,max}
\DeclareMathOperator*{\argmin}{arg\,min}
\let\polylog\relax
\DeclareMathOperator*{\polylog}{\mathrm{polylog}}

\makeatletter
\newtheorem*{rep@theorem}{\rep@title}
\newcommand{\newreptheorem}[2]{%
\newenvironment{rep#1}[1]{%
    \def\rep@title{#2 \ref{##1}}%
    \begin{rep@theorem}}%
    {\end{rep@theorem}}}

\newcommand{\undersim}[1]{\mathrel{\mathpalette\@undersim{#1}}}
\newcommand{\@undersim}[2]{%
\vcenter{%
\ialign{%
  ##\cr
  $\m@th#1#2$\cr
  \noalign{\nointerlineskip\kern.2ex}
  $\m@th#1\sim$\cr
  \noalign{\kern-.4ex}
}%
}%
}

\newcommand{\lsim}{\undersim{<}}
\makeatother

\newcommand{\wh}{\widehat}
\newcommand{\wt}{\widetilde}

\let\R\relax
\newcommand{\R}{\mathbb{R}}

\let\R\relax
\newcommand{\R}{\mathbb{R}}

\let\norm\relax
\newcommand{\norm}[1]{\|#1\|}

\SetKw{Break}{break}

\newcommand\numberthis{\addtocounter{equation}{1}\tag{\theequation}}

\newtheorem{theorem}{Theorem}
\newreptheorem{theorem}{Theorem}
\newtheorem{lemma}{Lemma}[section]
\newreptheorem{lemma}{Lemma}

\newtheorem{claim}[lemma]{Claim}

\newtheorem{definition}[lemma]{Definition}

\newtheorem*{claim*}{Claim}
\newtheorem*{proposition*}{Proposition}
\newtheorem*{lemma*}{Lemma}
\newtheorem*{problem*}{Problem}

  \usepackage{nth}
  \usepackage{intcalc}

{\makeatletter
	\gdef\xxxmark{%
		\expandafter\ifx\csname @mpargs\endcsname\relax 
		\expandafter\ifx\csname @captype\endcsname\relax 
		\marginpar{xxx}
		\else
		xxx 
		\fi
		\else
		xxx 
		\fi}
	\gdef\xxx{\@ifnextchar[\xxx@lab\xxx@nolab}
	\long\gdef\xxx@lab[#1]#2{{\bf [\xxxmark #2 ---{\sc #1}]}}
	\long\gdef\xxx@nolab#1{{\bf [\xxxmark #1]}}
}

\title{Sublinear Time Low-Rank Approximation of Hankel Matrices}

\author[1]{Michael Kapralov}
\author[2]{Cameron Musco}
\author[1]{Kshiteej Sheth}
\affil[1]{EPFL}
\affil[2]{University of Massachusetts Amherst}

\date{}

\begin{document}

\maketitle
\thispagestyle{empty}
\begin{abstract}
Hankel matrices are an important class of highly-structured matrices, arising across computational mathematics, engineering, and theoretical computer science. It is well-known that positive semidefinite (PSD) Hankel matrices are always approximately low-rank. In particular, a celebrated result of Beckermann and Townsend shows that, for any PSD Hankel matrix $H \in \mathbb{R}^{n \times n}$ and any $\epsilon > 0$, letting $H_k$ be the best rank-$k$ approximation of $H$ (obtained via truncated singular value decomposition),  $\|H-H_k\|_F \leq \epsilon \|H\|_F$ for $k = O(\log n \log(1/\epsilon))$. I.e., the optimal low-rank approximation error decays exponentially in the rank-$k$. As such, PSD Hankel matrices are natural targets for low-rank approximation algorithms. We give the first such algorithm that runs in \emph{sublinear time}. In particular, we show how to compute, in $\polylog(n, 1/\epsilon)$ time, a factored representation of a rank-$O(\log n \log(1/\epsilon))$ Hankel matrix $\widehat{H}$ matching the error guarantee of Beckermann and Townsend up to constant factors. We further show that our algorithm is \emph{robust} -- given input $H+E$ where $E \in \mathbb{R}^{n \times n}$ is an arbitrary non-Hankel noise matrix, we obtain error $\|H - \widehat{H}\|_F \leq O(\|E\|_F) + \epsilon \|H\|_F$. Towards this algorithmic result, our first contribution is a \emph{structure-preserving} existence result - we show that there exists a  rank-$k$ \emph{Hankel} approximation to $H$ matching the error bound of Beckermann and Townsend. Our result can be interpreted as a finite-dimensional analog of the widely applicable AAK theorem, which shows that the optimal low-rank approximation of an infinite Hankel operator is itself Hankel. Armed with our existence result, and leveraging the well-known Vandermonde structure of Hankel matrices, we achieve our sublinear time algorithm using a sampling-based approach that relies on universal ridge leverage score bounds for Vandermonde matrices.
\end{abstract}

\newpage
\tableofcontents
\thispagestyle{empty}

\newpage
\pagenumbering{arabic} 
\setcounter{page}{1}

\section{Introduction}

A Hankel matrix $H\in \mathbb{R}^{n\times n}$ is constant along each of its anti-diagonals, that is $H_{i,j} = H_{k,l}$ for all $i,j,k,l\in [n]$ such that $i+j=k+l$. These highly structured matrices arise in a wide range of applications in computational mathematics, signal processing, control theory, and beyond. For example, they arise in system identification and model order reduction for linear dynamical systems \cite{Glover:1984,Fazel:2013vw,liu2010interior,MarkovskyUsevich:2012,sontag2013mathematical,Megretski:2024}. They are used in singular spectrum analysis \cite{Hassani:2007}, dynamic mode decomposition \cite{ArbabiMezic:2017}, and many other approaches in signal processing. They are also used by algorithms for approximating functions via sums of exponentials \cite{BeylkinMonzon:2005}, recovering probability measures from their polynomial moments \cite{Schmudgenothers:2017}, and efficiently converting between polynomial bases \cite{TownsendWebbOlver:2018}. In theoretical computer science, (block) Hankel matrices have recently found applications to fast algorithms for solving sparse linear systems \cite{eberly2006solving,eberly2007faster,peng2021solving,casacuberta2021faster,nie2022matrix,ghadiri2023symmetric} using Krylov subspace methods.

Row-reversing a Hankel matrix yields a Toeplitz matrix -- these matrices arise in many applications as well 
\cite{Gray:2006ta}. Hankel matrices are also closely related to other `displacement structured' matrices, such as Vandermonde and Cauchy matrices
\cite{KailathSayed:1995}.

Due to their wide applications, fast algorithms for basic linear algebraic problems on Hankel matrices have been heavily studied. A Hankel matrix $H \in \R^{n \times n}$ can be multiplied by a vector in $O(n\log n)$ time using the fast Fourier transform. Hankel linear systems can be solved exactly in $O(n^2)$ time via Levinson recursion \cite{golub2013matrix}, and to high precision in $O(n\cdot \polylog (n))$ time \cite{XiaXiGu:2012,XiXiaCauley:2014}. The full eigendecomposition of a Hankel matrix can be computed in $ O(n^2\cdot \polylog (n))$ time \cite{pan1999complexity}. In this work, we focus on algorithms with runtime scaling \emph{sublinearly in $n$}, i.e., the number of parameters required to describe $H$. Several recent papers give sublinear time algorithms for the closely related class of Toeplitz matrices \cite{EldarLiMusco:2020,KapralovLawrenceMakarov:2023,MuscoSheth:2024}. However, to the best of our knowledge, sublinear time methods for Hankel matrices have remained largely unexplored.

\subsection{Hankel Low-Rank Approximation.}

It is well-known that Hankel matrices tend to have rapidly decaying eigenvalues and be poorly conditioned \cite{Tyrtyshnikov:1994,Beckermann:2000,beckermann2017singular}. Often, they are well-approximated by a lower rank matrix, and in fact, the low-rank approximation of Hankel matrices drives their use in many of the previously mentioned applications. This includes applications to system identification and model-order reduction \cite{Glover:1984,MarkovskyUsevich:2012,Fazel:2013vw,Megretski:2024}, signal processing \cite{Hassani:2007,ArbabiMezic:2017,GillardUsevich:2022}, function approximation \cite{BeylkinMonzon:2005}, and fast basis transformations \cite{TownsendWebbOlver:2018}.

In a celebrated result, Beckermann and Townsend \cite{beckermann2017singular} precisely quantify the eigenvalue decay of \emph{positive semidefinite (PSD)} Hankel matrices, which arise in many applications. Formally they show that for any PSD Hankel matrix $H\in \mathbb{R}^{n\times n}$, if we let $H_k$ be the best rank-$k$ approximation to $H$, computed by projecting $H$ onto its top $k$ eigenvectors, then, for $k=O(\log n \log(1/\epsilon))$,
\begin{align}\label{eq:townsendBound} 
  \|H-H_k\|_2\leq \epsilon \|H\|_2 \text{ and } \|H-H_k\|_F\leq \epsilon \|H\|_F,
\end{align}
where $\norm{\cdot }_2$ and $\norm{\cdot}_F$ denote the matrix spectral and Frobenius norms, respectively.\footnote{\cite{beckermann2017singular} only claims the bound for the spectral norm, but the Frobenius norm bound follows as a simple consequence -- see \Cref{lem:townsend_corrolary}.} That is, the optimal rank-$k$ approximation to $H$ has error decaying exponentially in $k$. 
Beckermann and Townsend prove their result via a connection to an extremal problem for rational functions and using the displacement structure of Hankel matrices.

Since Hankel matrices admit near-linear time matrix-vector multiplication via fast Fourier transform, for $k = O(\log n \log(1/\epsilon))$ we can compute the factors of a near-optimal rank-$k$ approximation to a given  Hankel matrix in $ O(n \cdot \polylog(n,1/\epsilon))$ time, using an iterative method like a block power or Krylov method \cite{HalkoMartinssonTropp:2011,MuscoMusco:2015}. I.e., by \eqref{eq:townsendBound} we can compute $\wh{H}$ satisfying $\|H - \wh{H}\|_2\leq \epsilon \|H\|_2$ and $\|H-\wh{H}\|_F\leq \epsilon \|H\|_F$ in near-linear time. While this is significantly faster than what would be possible if $H$ were unstructured, given recent progress on \emph{subinear time} low-rank approximation algorithms for the closely related class of PSD Toeplitz matrices \cite{EldarLiMusco:2020,KapralovLawrenceMakarov:2023,MuscoSheth:2024}, it is natural to ask if sublinear runtime is also achievable for PSD Hankel matrices. Since row-reversing a PSD Toeplitz matrix does not in general yield a PSD Hankel matrix, these results do not directly apply to our setting. Moreover, PSD Toeplitz matrices can have a flat spectrum (consider, e.g., the identity matrix), thus behaving very differently from PSD Hankel matrices, which always have exponentially decaying eigenvalues.

\subsection{Our Contributions.}

Our main result is to show that accurate low-rank approximation of PSD Hankel matrices can in fact be performed in sublinear time. In particular, we give an algorithm with $\polylog(n,1/\epsilon)$ runtime that matches the Frobenius norm bound of Beckermann and Townsend from  \eqref{eq:townsendBound}. Further, our algorithm is robust -- it can handle inputs that differ from a Hankel matrix by an arbitrary noise matrix, paying additional error proportional to this noise. Formally, we prove:

\begin{theorem}[Sublinear Time Hankel Low-Rank Approximation]\label{thm:const_factor_approx_main}
Let $H \in \mathbb{R}^{n\times n}$ be a PSD Hankel matrix, $E\in \mathbb{R}^{n\times n}$ be an arbitrary noise matrix, and $\epsilon > 0 $ be an error parameter. Given entrywise access to $H+E$, \Cref{alg:noisy_hankel_recovery} runs in  $\polylog( n,1/\epsilon)$ time and returns a rank-$O(\log n \log(1/\epsilon))$ Hankel matrix $\wh{H}$ (in a factored representation) such that, with probability  $\ge 9/10$, for a fixed constant $C$, $$\|H-\wh{H}\|_F \le  C\|E\|_F + \epsilon \|H\|_F.$$
  \end{theorem}

  Note that the output $\wh{H}$, of 
\Cref{thm:const_factor_approx_main} is \emph{structure-preserving} -- i.e., it is itself Hankel. This is critical to achieving sublinear runtime, as it allows $\wh{H} $ to be represented using just $O(\log n \log(1/\epsilon))$ parameters (see \Cref{sec:tech_overview} for details) and output in $\polylog(n,1/\epsilon)$ time. If $\wh H$ were a general rank $O(\log n \log(1/\epsilon))$ matrix, it would require $\Omega(n \log n \log(1/\epsilon))$ time even to write down in factored form. We note that in many applications \cite{BeylkinMonzon:2005,golyandina2010choice,liu2010interior,MarkovskyUsevich:2012,Fazel:2013vw}, $\wh H$ being Hankel is a desirable property in its own right, and the problem of {structure-preserving} low-rank Hankel approximation has been studied heavily studied \cite{fazel2003log,chu2003structured,MarkovskyUsevich:2012,Ishteva:2014wm,Knirsch:2021ve,GillardUsevich:2022}. See \Cref{sec:related} for further discussion. 

In any case, towards proving \Cref{thm:const_factor_approx_main}, our first step is to prove a result analogous to that of Beckermann and Townsend \cite{beckermann2017singular}, but for the existence of a sequence of \emph{Hankel} low-rank approximations with exponentially decaying error bounds. This existence result may be of independent interest.
\begin{theorem}[Existence of Accurate Hankel Low-Rank Approximations]\label{thm:main_thm}
    For any PSD Hankel matrix $H\in \mathbb{R}^{n\times n}$ and any $\epsilon>0$, there exists a rank $O(\log n \log(1/\epsilon))$ Hankel matrix $\wh{H}$ such that $$\|H-\wh{H}\|_2\leq \epsilon \|H\|_2\text{ and }\|H-\wh{H}\|_F\leq \epsilon \|H\|_F.$$
\end{theorem}
Note that the optimal low-rank approximation $H_k$ of a Hankel matrix $H$ is not Hankel in general. Thus, the result of Beckermann and Townsend does not imply anything about the accuracy of structure-preserving Hankel low-rank approximations as considered in \Cref{thm:main_thm}. For PSD Hankel matrices, our result can be viewed as a finite-dimensional analog of the widely applied theory of Adamyan, Arov, and Krein (AAK theory) \cite{AdamyanArovKrein:1971,Pellerothers:2003,YuTownsend:2024,Megretski:2024}, which establishes that for infinite-dimensional Hankel operators, $H_k$ is itself Hankel. For the case for finite $H$, we instead show that there exists rank $k$ Hankel $\wh H$ for which the error $\|H-\wh{H}\|_2$ and $\|H-\wh{H}\|_F$ still decays exponentially in the rank with the same asymptotic rate as $\|H-H_k\|_2$ and $\|H-H_k\|_F$, respectively. Thus, requiring $\wh H$ to be Hankel incurs no significant loss.

We complement \Cref{thm:main_thm} with a lower bound showing that the dependencies on $\log n$ and $\log(1/\epsilon)$ in the rank bound are both required, even when $\wh H$ is not required to be Hankel.


\begin{theorem}[Lower Bound on the Approximate Rank of PSD Hankel Matrices]\label{thm:epsilon_rank_lower_bound}
For any $n$ sufficiently large and $\epsilon \in (0,1-c)$ for any constant $c$, there exists a PSD Hankel matrix $H\in \mathbb{R}^{n\times n}$ and a constant $C$ such that for $k\leq C(\log n +\log(1/\epsilon))$, $\|H-H_k\|_2\geq \epsilon \|H\|_2$ and $\|H-H_k\|_F\geq \epsilon \|H\|_F$.
\end{theorem}
The lower bound of \Cref{thm:epsilon_rank_lower_bound} does not quite match the $O(\log n \log(1/\epsilon))$ bound of \Cref{thm:main_thm}, or the prior of Beckermann and Townsend \cite{beckermann2017singular}. However, we believe that $O(\log n \log(1/\epsilon))$ is in fact tight, and proving a matching lower bound would be an interesting direction for future work.
\smallskip

\noindent\textbf{From existence to sublinear time approximation.}
With \Cref{thm:main_thm} in hand, to prove  \Cref{thm:const_factor_approx_main} we must show how to actually find $\wh H$ in sublinear time. Roughly speaking, our approach, which is detailed in \Cref{sec:tech_overview}, argues that $\wh H$ achieving the error bound of \Cref{thm:main_thm} can be further restricted to be a sum of two Hankel matrices $\wh H=\wh H_1+\wh H_2$. $\wh H_1$ is supported only on its top left and bottom right corners. In particular,  $(\wh H_1)_{i,j}$ is nonzero only for $i,j\in [n]$ such that $i+j\leq  O(\log n \log(1/\epsilon))$ or $ i+j \geq 2n-  O(\log n \log(1/\epsilon))$. $\wh H_2$ is low-rank with a \emph{fixed row/column span}, given by the column span of a  Vandermonde matrix whose parameters depend only on $n$ and $\epsilon$, not $H$.

With this restriction, we can frame identifying a minimizer $\wh H = \wh H_1 + \wh H_2$ of $\norm{H - \wh H}_F$ as a least squares regression problem, which we can approximately solve in sublinear time. We first learn $
\wh H_1$ simply by averaging the first and last $O(\log n \log(1/\epsilon))$ anti-diagonals of $H$. We then learn $\wh H_2$ using prior leverage-score-based random sampling techniques for linear regression \cite{EldarLiMusco:2020,KapralovLawrenceMakarov:2023,MuscoSheth:2024}. To do so, we prove new universal (ridge) leverage score bounds for Vandermonde matrices -- i.e., for the row/column span of $\wh H_2$. 

Note that, while the existence result of \Cref{thm:main_thm} holds for both the spectral and Frobenius norm, our sampling methods apply to Frobenius norm approximation only. \Cref{thm:const_factor_approx_main} implies a spectral norm error bound of $\|H - \wh H|_2 \le \epsilon\|H\|_2$ when $\wh H$ has rank $O(\log n \log(n/\epsilon))$ simply by using that $\|H \|_F \le \sqrt{n} \cdot \|H\|_2$. However, obtaining a sublinear time algorithm with tight spectral norm error bounds -- i.e., where $\wh H$ has rank $O(\log n \log(1/\epsilon))$ -- is an interesting open problem.

\smallskip
\noindent\textbf{Applications of our results.}
Our work has various applications, which we summarize here and detail in \Cref{sec:applications}.  The first is to speeding up the fast polynomial basis transform algorithm of \cite{TownsendWebbOlver:2018}. This algorithm transforms the coefficients of a degree-$n$ polynomial in one orthogonal polynomial basis (e.g., Chebyshev) to another (e.g., Legendre) in $O(n \log^4 n)$ time. They use PSD Hankel low-rank approximation is a key subroutine, and plugging in the sublinear time algorithm of Theorem \ref{thm:const_factor_approx_main} improves their overall runtime to $O(n \log^3 n)$.  

A second application is to sample-efficient covariance estimation for high-dimensional distributions with Hankel covariance matrices. 
Significant recent work focuses on Toeplitz covariance matrices \cite{EldarLiMusco:2020,MuscoSheth:2024}, which arise from stationary random processes. Hankel covariance matrices are similarly important, arising e.g., from signals generated by certain 1D autoregressive processes \cite{Fazel:2013vw,aoki2013state}. Similar to the approach of Theorem 3 in \cite{MuscoSheth:2024}, the algorithm of Theorem \ref{thm:const_factor_approx_main}
can be leveraged to estimate a Hankel covariance matrix $H$ to error $\epsilon \cdot \norm{H}_2$ using $\poly(\log n,1/\epsilon)$ samples from the distribution and $\polylog(n,1/\epsilon)$ entries read per sample.  
To the best of our knowledge, the prior best known bound here was  $\Omega(n/\epsilon^2)$  samples and $n$ entries read per sample, which follows from standard results on generic unstructured covariance estimation.

Finally, factorizations of PSD Hankel matrices can be used to compute Sum-of-Squares (SoS) decompositions of polynomials -- see Section 3.3 of \cite{ghadiri2023symmetric}. In particular, a Hankel matrix $H$ can be constructed from any polynomial $p$  and if $H$ is PSD, then $p$ has an SoS decomposition. A fast PSD low-rank factorization algorithm for $H$ with small error can be used to compute such a decomposition with few terms, up to small error over a bounded domain. 
Our algorithm does not quite apply here, since its output $\wh H$ may not be exactly PSD (although it will be near-PSD since it is close to $H$, which is PSD). Obtaining the result of \Cref{thm:const_factor_approx_main} where $\wh H$ is further guaranteed to be exactly PSD is an interesting open question.

\subsection{Related Work.}\label{sec:related}
We now discuss several areas of prior work closely related to our own.

\smallskip

\noindent\textbf{Sublinear time low-rank approximation.} Most closely related to our work is a line of work that studies sublinear query and sublinear time algorithms for low-rank approximation of PSD Toeplitz matrices \cite{Abramovich:1999vs,Chen:2015wz,Qiao:2017tp,Lawrence:2020ut,EldarLiMusco:2020,KapralovLawrenceMakarov:2023,MuscoSheth:2024}, which arise commonly as covariance matrices of stationary processes -- i.e., when the covariance structure is shift invariant. \cite{KapralovLawrenceMakarov:2023} in particular also shows the existence of a good Toeplitz low-rank approximation for PSD Toeplitz matrices.

Beyond work on sublinear time algorithms for PSD Toeplitz matrices, significant recent work has focused on sublinear time low-rank approximation algorithms for other structured matrix classes, such as general PSD matrices \cite{musco2017recursive,Bakshi:2020tl,ChenEpperlyTropp:2025}, distance matrices \cite{Bakshi:2018ul,Indyk:2019vy}, and kernel matrices \cite{musco2017recursive,Yasuda:2019vf,Ahle:2020vj}. Typically, the goal is to achieve runtime scaling linearly in $n$ for an $n \times n$ matrix, which is generally optimal for the above structured classes.  

\smallskip

\noindent\textbf{Structure-preserving low-rank approximation.}  As discussed, significant prior work in numerical linear algebra and signal processing has studied structured-preserving low-rank approximation algorithms for Hankel, and Toeplitz matrices -- see e.g., \cite{Luk:1996ur,Krim:1996wm,Park:1999ta,fazel2003log,chu2003structured,MarkovskyUsevich:2012,Ishteva:2014wm,Cai:2016wh,Ongie:2017us,Knirsch:2021ve,GillardUsevich:2022}. This work typically frames the problem somewhat differently than we do - given a fixed rank parameter $k$, they aim to find a rank-$k$ Hankel matrix $\wh H$ which minimizes or approximately minimizes $\norm{H-\wh H}_F$ over all possible rank-$k$ Hankel matrices. Unlike unconstrained low-rank approximation, this constrained low-rank approximation problem is difficult - no simple characterization of the optimal solution is known, and polynomial time computation of the optimum remains open outside a few special cases \cite{chu2003structured,Knirsch:2021ve}. Thus, prior work often tackles the problem with heuristics such as convex relaxation \cite{Fazel:2013vw,Cai:2016wh,Ongie:2017us} and alternating minimization \cite{chu2003structured,Wen:2020ub}. We note that, in practice, since the eigenvalues of PSD Hankel matrices decay so quickly, an optimal, or even near-optimal rank-$k$ Hankel approximation is typically unnecessary. Thus, we target any algorithm achieving the exponential decay bound of \Cref{eq:townsendBound}, allowing us to avoid some of the difficulties of prior work, and achieve sublinear runtime.

\smallskip

\noindent\textbf{Finite-dimensional AAK theory.} Finally, recall that \Cref{thm:main_thm} can be viewed as a finite-dimensional analog of the celebrated AAK theory \cite{AdamyanArovKrein:1971}, which establishes that infinite-dimensional Hankel operators admit optimal low-rank approximations that are themselves Hankel. Steps towards a finite-dimensional analog of this theory have been made in prior work \cite{BeylkinMonzon:2005,AnderssonCarlssonHoop:2011}. However,  while numerical experiments support the existence of highly accurate structure-preserving Hankel low-rank approximations, to the best of our knowledge, no theoretical error bound similar to \Cref{thm:main_thm} has previously been established, even in the special case of $H$ being PSD.

\subsection{Paper Organization.}
This paper is organized as follows. \Cref{sec:tech_overview} contains a technical overview of our results. \Cref{sec:prelim} introduces notation and preliminaries. \Cref{sec:existence_main} contains the proof of our existence result for an accurate Hankel low-rank approximation (\Cref{thm:main_thm}). \Cref{sec:lower_bd} contains the proof of our corresponding lower bound (\Cref{thm:epsilon_rank_lower_bound}). \Cref{sec:alg_main} contains the proof of our main algorithmic result (\Cref{thm:const_factor_approx_main}). \Cref{sec:applications} details a few applications of our results. Finally, \Cref{sec:conclusion} concludes and discusses open problems.

\section{Technical Overview.}\label{sec:tech_overview}

In this section, we give an overview of the main techniques and ideas behind our results. From classic work on Hankel matrices \cite{fiedler86,beckermann2017singular}, it is known 
that every positive definite Hankel matrix $H$ admits a \emph{Fiedler factorization} into diagonal and real Vandermonde matrices. Before introducing this factorization, we introduce notation for moment vectors and Vandermonde matrices. 
\begin{definition}[Moment vector]\label{def:moment_vector}
    For any $x\in \mathbb{R}$ and positive integer $n$, let $v_n(x)\in \mathbb{R}^n$ be the moment vector defined as the column vector $v_n(x)=[1,x,\ldots,x^{n-1}]^T$.
\end{definition}
\begin{definition}[Real Vandermonde matrix]\label{def:vandermonde_matrix}
For a set $X\subset \mathbb{R}$ let $V_X\in \mathbb{R}^{n\times |X|}$denote the Vandermonde matrix whose columns are the moment vectors $v_n(x)\in \mathbb{R}^n$ for each $x\in X$. The values in $X$ are referred to as the ``nodes'' of $V_X$.
\end{definition}
Equipped with the above, we can give the Fielder decomposition, which will be key to our approach, allowing us to view Hankel matrices as outer products of Vandermonde matrices.
\begin{claim}\label{clm:ij_entry}
For any Vandermonde matrix $V_X\in \mathbb{R}^{n\times |X|}$ and diagonal matrix $D_X=diag(\{a_x\}_{x\in X})$, the matrix $V_XD_XV_X^T$ has entries $(V_XD_XV_X^T)_{i,j}=\sum_{x\in X}a_x x^{i+j}$ that only depend on $i+j$ for all $i,j\in [n]$. Thus, $V_XD_XV_X^T$ is Hankel. Further, $V_XD_XV_X^T$ has rank $\le |X|$.
\end{claim}
The proof of the above trivially follows by expanding $V_XD_XV_X^T$. The Fiedler Factorization statement below gives a converse: all positive definite Hankel matrices admit such a factorization. 
\begin{lemma}[Fiedler Factorization\footnote{For reference, this version of the statement is stated in the discussion below Lemma 5.4 in 
\cite{beckermann2017singular}.}]\label{lem:vandermonde_decomp}
For every real positive definite Hankel matrix $H\in \mathbb{R}^{n\times n}$ there exists a set $X\subset \mathbb{R}$ of size $|X| = n$ and corresponding positive weights $a_x>0$ for all $x\in X$ such that $H=V_XD_XV_X^T$, for Vandermonde $V_X\in \mathbb{R}^{n\times n}$ and diagonal $D_X=diag(\{a_x\}_{x\in X})\in \mathbb{R}^{n\times n}$.
\end{lemma}
Since any PSD Hankel matrix can be approximated to arbitrary accuracy by a positive definite Hankel matrix, w.l.o.g., we will assume that $H$ in Theorems \ref{thm:const_factor_approx_main} and \ref{thm:main_thm} is positive definite, so that we can apply \Cref{lem:vandermonde_decomp}. 

With the above preliminaries in place, we are ready to overview our proof approach.
In \Cref{sec:existence_overview} we present the main ideas behind the existence result of Theorem \ref{thm:main_thm}. We then give the ideas behind the approximate rank lower bound of Theorem \ref{thm:epsilon_rank_lower_bound} in \Cref{sec:lower_bound_overview}. Finally, in \Cref{sec:alg_overview} we give the ideas behind the sublinear time algorithm of Theorem \ref{thm:const_factor_approx_main}.

\subsection{Existence of a Good Hankel Low Rank Approximation.}\label{sec:existence_overview}
We start by explaining the key ideas behind  Theorem \ref{thm:main_thm}.  \Cref{sec:existence_main} is dedicated to the full proof. Let $H=V_X D_X V_X^T$ be the Fiedler factorization of the given positive definite Hankel $H$, and assume $|x|\leq 1$ for all $x\in X$. We will later show that the general case reduces to this one. Letting $k = O(\log n \log(1/\epsilon))$, our goal is to find a rank-$k$ {Hankel} approximation $\wh H$ to $H$ with small error. At a very high level, we achieve this via \emph{sparsifying} $X$ to a set $T\subset \mathbb{R}$ of size $k$ and finding diagonal $D_{T}$, such that $\wh H = V_{T}D_{T}V_{T}^T$ (which is rank-$k$ Hankel by \Cref{clm:ij_entry}) serves as our approximation. We will first prove entrywise error bounds for our approximation, and then convert them into spectral and Frobenius norm bounds, as required by \Cref{thm:main_thm}. 

\smallskip

\noindent\textbf{A simple approach and why it fails.} A natural choice is to let $T$ contain the top  $k$ nodes in $X$, as ordered by the corresponding entries in $D_X$. However, this approach fails due to the ill-conditioning of real Vandermonde matrices \cite{Beckermann:2000}. In particular, consider the example $X=\{1/2+i\delta \}_{i=1}^{n/2}\cup \{2+i\delta \}_{i=1}^{n/2}$ for some $\delta \rightarrow 0$. That is, $X$ has two tight clusters of nodes, one around $1/2$ and the other around $2$. Thus $V_X$, and in turn $V_XD_XV_X^T$, tends to a rank-$2$ matrix. Moreover, suppose the diagonal entries in $D_X$ are 1 for the nodes around $1/2$ and $0.99$ for those around $2$.  Choosing the top $k\ll n/2$ nodes will miss all nodes from the cluster around $2$, thus failing to be a good low rank approximation. This example suggests a natural bucketing strategy, i.e., to bucket $X$ into sets with close nodes, and approximate the contribution of each bucket with a low-rank matrix. Ultimately, this is the approach that our proof will take.

\smallskip

\noindent\textbf{Log-scale bucketing.} Consider some  $x\in X$ whose corresponding column in $V_x$ is given by the moment vector $v_n(x)=[1,x,\ldots,x^{n-1}]^T$ (Definition \ref{def:moment_vector}). Observe that since, by our simplifying assumption, $|x| \le 1$, the entries of $v_n(x)$ are decaying exponentially. In particular, letting $y_x= \ln(1/|x|)$ denote the rate of decay and $\lambda =2\ln(1/\epsilon)$, all entries of $v_n(x)$ beyond index $i^*=\lambda /y_x$ are at most $|x|^{i^*} = e^{-y_xi^*} < \epsilon$. We can interpret $\lambda/y_x$ as a \emph{cutoff} index, beyond which all entries in $v_n(x)$ are at most $\epsilon$. 

We partition $X$ into $R\approx \log n$ buckets of nodes with similar rates of decay (formally, see \Cref{def:bucketing}):
\begin{align*}
B_0 = \left\{x\in X:  y_x\in \left[0,\frac{\lambda}{n}\right]\right\}&,\quad
B_r = \left\{x\in X:  y_x\in \left[\frac{2^{r-1}\lambda}{n},\frac{2^{r}\lambda}{n}\right]\right\} \text{ for } r \in \{1,\ldots, R\},\numberthis\label{eq:informal_bucketing}\\ B_{R+1} = \left \{x \in X: y_x > \frac{2^{R} \lambda}{n} \right \}\\
\text{ weight  } w_r = \sum_{x\in B_r}a_x,&\quad \text{ cutoff  } t_r=n/2^{r-1}.
\end{align*}
For every $r \in [R+2]$ (recall that throughout, $[n]$ denotes the set $\{0,\ldots,n-1\}$), we define the function $H_{B_r}(t)$ and the corresponding Hankel matrix $H_{B_r}$ as:
\begin{equation}\label{eq:informal_entrywise_bucket_def}
    H_{B_r}(t) = \sum_{x\in B_r}a_x x^t \quad \forall t\in [2n-1], \quad (H_{B_r})_{i,j}=H_{B_r}(i+j)\quad \forall  i,j\in [n].
\end{equation}
Since the buckets partition $X$, using Claim \ref{clm:ij_entry}, it is easy to see that $H=\sum_{r\in [R+2]} H_{B_r}$. Thus, if we can find low-rank entrywise approximations to each $H_{B_r}$ via sparsification, we can add them up to obtain an entrywise approximation to $H$, and ultimately a good approximation in the spectral and Frobenius norms.

\smallskip

\noindent\textbf{Entrywise approximation via bucket sparsification.} We will use node sparsification  to entrywise approximate each $H_{B_r}$ with a low-rank Hankel matrix. We first show how to approximate the underlying function $H_{B_r}(t)$ using a \emph{fixed} set of $O(\log(1/\epsilon))$ exponentials, with rates of decay given by scaled Chebyshev nodes.  
\begin{lemma}[Simplified, full version in  \Cref{sec:bucket_sparsification}]\label{lem:bucket_sparsification}
For $l=O(\log(1/\epsilon))$ and each $r\in [R+1]$, let $\{t_{i,r}\}_{i\in [l+1]}$ be the degree $l$ Chebyshev nodes on $[2^{r-1}\lambda/n,2^{r}\lambda/n]$ (see Definition \ref{def:chebyshev_nodes}). Let $T_r$ be the set of exponentiated Chebyshev nodes $\{e^{-t_{i,r}},-e^{-t_{i,r}}\}_{i\in [l+1]}$. Then there exists a set of coefficients $\{a_x\}_{x \in T_r}$ so that $H_{T_r}(t)=\sum_{x\in T_r}a_x x^t$ satisfies:
\begin{equation*}
    \left|H_{B_r}(t) - H_{T_r}(t)\right|\leq \epsilon w_r\quad \forall t\in [2n-1].
\end{equation*}
Recall that $w_r$ is the total weight of bucket $B_r$, as in \eqref{eq:informal_bucketing}.
\end{lemma}
\Cref{lem:bucket_sparsification} allows us to replace the nodes in each bucket $B_r$ with those in $T_r$, for all $r\in [R+1]$. See Figure \ref{fig:chebyshev_node_sparsification} for a visualization. Letting $H_{T_r}$ be the Hankel matrix given by $(H_{T_r})_{i,j}=H_{T_r} (i+j)$, the approximation bound of \Cref{lem:bucket_sparsification} can be restated as an entrywise matrix approximation bound:
\begin{equation*}
    |(H_{B_r}-H_{T_r})_{i,j}|\leq \epsilon w_r \quad \forall i,j\in [n], r\in [R+1].\numberthis\label{eq:entrywise_informal}
\end{equation*}
Furthermore, since by \Cref{clm:ij_entry}, $H_{T_R}$ has a Fiedler factorization with just $2(l+1)$ nodes, $H_{T_r}$ has rank $2(l+1)=O(\log (1/\epsilon))$. Thus, the sum of these approximations $\wh H = \sum_{r = 0}^{R+1} H_{T_r}$ is an $O(R \log(1/\epsilon)) = O(\log n \log(1/\epsilon))$ rank Hankel matrix, which will ultimately serve as our approximation to $H$ (note that $r = R+1$ is a corner case not covered by \Cref{lem:bucket_sparsification}, which we will handle later).

We now sketch the ideas behind Lemma \ref{lem:bucket_sparsification}. See  \Cref{sec:bucket_sparsification} for a full proof. First, fix $r\in [R+1]$. As discussed, the contribution of each $x\in B_r$ beyond its cutoff is small. In particular, from the definition in  \eqref{eq:informal_bucketing}, $|H_{B_r}(t)|\leq \epsilon w_r$ for all $t$ beyond the cutoff $t_r$. On the other hand, before the cutoff $t_r$, we will show that $H_{B_r}(t)$ is well approximated by a low-degree polynomial. More precisely, consider any $x\in B_r$ with $x\geq 0$. We will handle $x < 0$ with a separate but analogous argument. Let the corresponding exponential be $x^t = e^{-y_x t}$. We can see that the exponent $-y_x t$ lies in $[-2\lambda, 0]$ for $t\in [0,t_r]$ by \eqref{eq:informal_bucketing}.
This is a narrow range, since $\lambda = O(\log(1/\epsilon))$. Hence, $e^{-y_x t}$ can be well-approximated by a low-degree polynomial via Taylor expansion. In particular, for $l = O(\lambda)=O(\log(1/\epsilon))$ we can show:
\begin{align}\label{eq:cam1}
x^t = e^{-y_x t} &=\sum_{m=0}^{l-1} \frac{(-t)^m}{m!} \cdot y_x^m \pm \epsilon \implies \sum_{x\in B_r:x\geq 0}a_x e^{-y_x t} = \left ( \sum_{x\in B_r:x\geq 0}a_x \sum_{m=0}^{l-1} \frac{(-t)^m}{m!} \cdot y_x^m \right ) \pm \epsilon w_r.
\end{align}
We would like to sparsify the outer sum over $x \in B_r$. To do this, we observe that the inner sum can be written as a low-dimensional inner product. In particular, letting $v(t) = [1, -t, \ldots, (-t)^{l-1} / (l-1)!]^T$ and recalling the definition of the moment vector (\Cref{def:moment_vector}) $v_l(y_x) = [1,y_x,\ldots,y_x^{l-1}]^T$, we have $\sum_{m=0}^{l-1} \frac{(-t)^m}{m!} \cdot y_x^m = v_l(y_x)^T v(t)$. We can thus rewrite \eqref{eq:cam1} as:
\begin{align*}
    \sum_{x\in B_r:x\geq 0}a_x e^{-y_x t}  = \left( \sum_{x\in B_r:x\geq 0}a_xv_l(y_x)\right)^T v(t)\pm \epsilon w_r.
\end{align*}
%
%
%
%
The sum $\sum_{x \in B_r:x\geq 0} a_x v_l(y_x)$ consists of up to $n$ vectors in $\mathbb{R}^{l}$. Since $l \ll n$ and $a_x > 0$ for all $x$, this sum can be sparsified, e.g., using Carath\'{e}odory's theorem. We will apply a related sparsification lemma of \cite{liu2022robust} that specifically applies to moment vectors (see Lemma \ref{lem:chebyshev_sparsification}), and was originally used in the context of sparse Fourier functions. This lemma lets us to replace $\sum_{x\in B_r:x\geq 0}a_xv_l(y_x)$ with a sum of just $l+1$ moment vectors corresponding to the degree $l$ Chebyshev nodes $\{t_{i,r}\}_{i\in [l+1]}$. Furthermore, each exponential $e^{-t_{i,r}t}$ has the same cutoff $t_r$, and thus is bounded by $\epsilon$ beyond the cutoff and approximable with a low-degree polynomial before the cutoff. This allows us to convert the sparse sum of low-dimensional moment vectors back into a sparse sum of exponentials $\{e^{-t_{r,i}t}\}_{i\in [l+1]}$, which sparsifies our original sum $\sum_{x\in B_r: x \ge 0}a_x e^{-y_it} = \sum_{x\in B_r: x \ge 0}a_x x^t$. Handling negative $x$ similarly completes the proof of Lemma \ref{lem:bucket_sparsification}.  

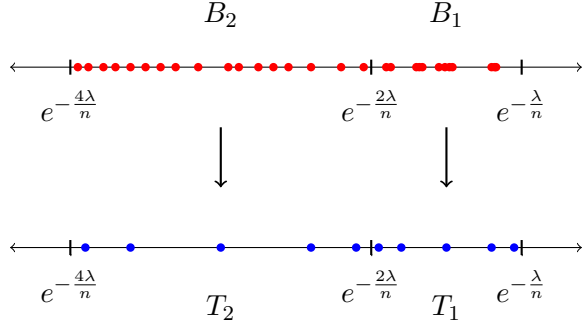
\begin{figure}\label{fig:chebyshev_node_sparsification}
\centering
\begin{tikzpicture}[scale=2]

\draw[<-]  (-0.4,1.2) -- (3,1.2) ;
\draw[<-]  (-0.4,0) -- (3,0) ;
\draw[->]  (3,1.2) -- (3.4,1.2) ;
\draw[->]  (3,0) -- (3.4,0) ;
\foreach \x/\label in {0/$e^{-\frac{4\lambda}{n}}$, 2/{$e^{-\frac{2\lambda}{n}}$}, 3/{$e^{-\frac{\lambda}{n}}$}} {
  \draw[thick] (\x,1.15) -- (\x,1.25);
  \node[below=2pt] at (\x,1.15) {\label};
}

\foreach \x in {0,2,3} {
  \draw[thick] (\x,-0.05) -- (\x,0.05);
}
\foreach \x/\label in {0/$e^{-\frac{4\lambda}{n}}$, 2/{$e^{-\frac{2\lambda}{n}}$}, 3/{$e^{-\frac{\lambda}{n}}$}} {
  \draw[thick] (\x,-0.05) -- (\x,-0.05);
  \node[below=2pt] at (\x,-0.05) {\label};
}

\node at (1,1.55) {$B_2$};
\node at (2.5,1.55) {$B_1$};
\node at (1,-0.4) {$T_2$};
\node at (2.5,-0.4) {$T_1$};

\foreach \x in {0.05, 0.12, 0.22, 0.3, 0.4, 0.5, 0.6, 0.7, 0.85, 1.05, 1.12, 1.25, 1.35, 1.45, 1.6, 1.8, 1.95,2.1, 2.13, 2.3, 2.32,2.34, 2.45, 2.49,2.52,2.54,2.8, 2.82,2.83} {
  \fill[red] (\x,1.2) circle (0.8pt);
}

\foreach \x in {0.05, 0.2, 0.5, 0.8, 0.95} {
  \fill[blue] (2+\x,0) circle (0.8pt);
}

\foreach \x in {1.1, 1.4, 2, 2.6, 2.9} {
  \fill[blue] (\x-1,0) circle (0.8pt);
}

\draw[->, thick] (1,0.8) -- (1,0.4);
\draw[->, thick] (2.5,0.8) -- (2.5,0.4);

\end{tikzpicture}
\caption{Sparsification of buckets $B_1,B_2$ to exponentiated Chebyshev nodes $T_1,T_2$ in respective intervals.}
\end{figure}

\smallskip

\noindent\textbf{From bucket approximation to low-rank approximation.} So far, we have written $H = \sum_{r = 0}^{R+1} H_{B_r}$ as a sum of node buckets, and, for $r \le R$, have shown that $H_{B_r}$ can be approximated by a rank $O(\log (1/\epsilon))$ Hankel matrix $H_{T_r}$ to small entrywise error (see \eqref{eq:entrywise_informal}). For $H_{B_{R+1}}$, which corresponds to nodes with high rates of decay, we can form the approximation $H_{T_{R+1}}$ simply by rounding any entry $(H_{B_{R+1}})_{i,j}$ with $ i + j$ larger than $O(\log(1/\epsilon))$ to $0$, yielding a sparse Hankel matrix with rank at most $O(\log(1/\epsilon))$ and small entrywise error. Our final approximation to $H$ will simply be the sum of the approximations to each $H_{B_r}$ -- i.e., $\wh H = \sum_{r = 0}^{R+1} H_{T_r}$. Recall that $\wh H$ is Hankel and has rank $O(R\cdot \log (1/\epsilon))= O(\log n\log (1/\epsilon))$ as desired.

We next need to use our entrywise approximation bounds to argue that $\|H - \wh H\|_2$ and $\norm{H - \wh H}_F$ are small. Suppose for simplicity that we only have one bucket -- i.e., $H=H_{B_0}$. Equation \eqref{eq:entrywise_informal} guarantees that for the error matrix $E_0=H_{B_0}-H_{T_0}$, we have $|(E_0)_{i,j}|\leq \epsilon w_0$ for all $i,j\in [n]$. We thus have $\norm{E_0}_2 \le \norm{E_0}_F \le \epsilon w_0 \cdot n$.
To prove an error guarantee of $\|E_0\|_2\leq \epsilon \|H_{B_0}\|_2$ as required by Theorem \ref{thm:main_thm}, we thus need to show a lower bound of $\|H_{B_0}\|_2 = \Omega( w_0\cdot n)$ (and analogous argument will yield the Frobenius norm error bound as well).

\smallskip

\noindent\textbf{Lower bounding the norm of $H$.} Continuing to focus on the single bucket case, we describe how to prove that $\norm{H_{B_0}}_2 = \Omega( w_0\cdot n)$. We then discuss how our arguments extend to the multiple bucket case. We will use repeated invocations of the following lemma, which we prove in \Cref{sec:prelim_hankel}. Recall that for two matrices $M,N \in \R^{n  \times n}$, $M \succeq N$ denotes that $M$ is larger than $N$ in the Loewner order -- i.e., that $M-N$ is PSD.
\begin{lemma}[Simplified, full version in  \Cref{sec:prelim_hankel}]\label{lem:partition_x}
Consider Hankel $H \in \R^{n \times n}$ with $H_{i,j}=\sum_{x\in X}a_x x^{i+j}$ with $a_x\geq 0$ for all $x\in X$. Consider any partition $X=X_0\cup\ldots \cup X_{l-1}$ and for $m \in [l]$, let $H_{X_m}$ be defined entrywise as $(H_{X_m})_{i,j} = \sum_{x\in X_m}a_x x^{i+j}$. We have $H=\sum_{m\in [l]}H_{X_m}$ and $H\succeq H_{X_m}$ for all $m\in [l]$.
\end{lemma}
Applying the above to $H_{B_0}$ for the partition $B_0= B^+_0\cup B^-_0$ containing positive and negative elements of bucket $B_0$ respectively, we have that $H_{B_0}\succeq H_{B^+_0}$ and $H_{B_0}\succeq H_{B^-_0}$. Thus,
\begin{equation}\label{eq:spectral_lb_informal}
    \|H_{B_0}\|_2\geq \max\{\|H_{B^+_0}\|_2,\|H_{B^-_0}\|_2\} \geq (\|H_{B^+_0}\|_2+\|H_{B^-_0}\|_2)/2,
\end{equation}
We will first lower bound $\|H_{B^+_0}\|_2$ by constructing a vector $v$ so that $\|H_{B^+_0}v\|_2$ is large. Recall that $H_{B_0} = V_{B_0} D_{B_0} V_{B_0}^T$, where all columns of the Vandermonde matrix $V_{B_0}$ are similar to the moment vector $v_n(e^{-\lambda/n})$ by our bucketing strategy (equation \eqref{eq:informal_bucketing}). Thus, we will simply choose 
%
%
%
$v$ to be this moment vector. 
We  can argue that $\|H_{B^+_0}v\|_2/\|v\|_2 \approx (\sum_{x\in B^+_0}a_x)\cdot \|v_n(e^{-\lambda/n})\|_2^2 \approx n\cdot (\sum_{x\in B^+_0}a_x)$.

We next handle $H_{B^-_0}$, which can have negative entries since $B^-_0$ has negative elements. However, observe that for any indices $i,j\in [n]$ such that $i$ is even, the sign of the entry $(H_{B^-_0})_{i,j}=\sum_{x\in B_0:x<0}a_x x^{i+j}$ only depends on $j$ as $x^{i+j}=(-1)^j|x|^{i+j}$ due to $i$ being even. This observation implies that all even rows $i\in [n]$ of $H_{B^-_0}$ have the same \emph{consistent sign pattern} not depending on $i$. Moreover, vectors $v$ whose signs correlate with the signs of even rows essentially \emph{correct for the negative signs}, avoiding cancellations in $H_{B^-_0}v$. This allows us to similarly obtain a $v$ with $\|H_{B^-_0}v\|_2/\|v\|_2 \approx n\cdot(\sum_{x\in B^-_0}a_x)$. Thus by \eqref{eq:spectral_lb_informal} we get that $\|H_{B_0}\|_2$ is approximately lower bounded by $ (\sum_{x\in B^+_0}a_x+\sum_{x\in B^-_0}a_x)\cdot n= w_0\cdot n$, as required to prove the error guarantee of Theorem \ref{thm:main_thm}. 

\smallskip

\noindent\textbf{Handling multiple buckets.}
We need to extend the above lower bound on $\|H\|_2$ to the more complex setting when the nodes of $H$ are contained in multiple buckets. Since the buckets partition $X$, applying Lemma \ref{lem:partition_x} allows us to lower bound $\|H\|_2 \ge \max_{r\in [R+2]}\|H_{B_r}\|_2$. Using an analogous argument for the single bucket case, we can obtain a lower bound for each $\|H_{B_r}\|_2$, by choosing $v$ to be the moment vector $v_n(e^{-2^r\lambda/n})$. We obtain the following bound, which is proven in \Cref{sec:hankel_spectral_lower_bd}, and strengthened to give stronger lower bounds on $\|H\|_F$.
%
%
%
\begin{lemma}[Simplified, full version in  \Cref{sec:hankel_spectral_lower_bd}]\label{lem:hankel_spectral_lower_bd}
For a PSD Hankel matrix $H = V_X D_X V_X^T\in \mathbb{R}^{n\times n}$ with all $x\in X$ satisfying $|x|\leq 1$, we have
\begin{equation*}
\|H\|_2 =\Omega\left(\max_{r\in [R+2]}w_r \cdot n/(\lambda 2^r)\right).
\end{equation*}
\end{lemma}
Equipped with the above lower bound in the presence of multiple buckets, what remains now is to prove a matching upper bound on the error, using the entrywise error bounds of \Cref{lem:bucket_sparsification}.

\smallskip

\noindent\textbf{Entrywise to matrix norm error guarantees using correlation decay.}
Let $E_r = H_{B_r}-H_{T_r}$ for all $r\in [R+1]$ be the errors in approximating each of buckets. 
Recall that, by \Cref{lem:bucket_sparsification}, equation \eqref{eq:entrywise_informal}, each $E_r$ satisfies $|(E_r)_{i,j}|\leq \epsilon w_r$ for all $i,j\in [n]$. This trivially implies a bound of $\norm{E_r}_2 \le \epsilon w_r \cdot n$ and thus $\norm{\sum_{r\in [R+2]}E_r}_2 \le \epsilon n \cdot \sum_{r \in [R+2]} w_r$. However, we will need a much stronger bound, which we state below:
\begin{lemma}[Simplified, full version in  \Cref{sec:error_correlation_analysis}]\label{lem:x_0_1_errs}
The errors $E_r=H_{B_r}-H_{T_r}$ for all $r\in [R+2]$ satisfy:
\begin{equation*}
    \left\|\sum_{r\in [R+2]}E_r\right\|_2\leq \epsilon\cdot \max_{r\in [R+2]}w_r \cdot n/2^r.
\end{equation*}
\end{lemma}
The full version of \Cref{lem:x_0_1_errs} stated in  \Cref{sec:error_correlation_analysis} also contains an analogous Frobenius norm bound. Observe that \Cref{lem:x_0_1_errs} matches the lower bound on $\|H\|_2$ given by Lemma \ref{lem:hankel_spectral_lower_bd}, thus implying the $\epsilon \norm{H}_2$ error bound of Theorem \ref{thm:main_thm}. 
Also note that if we used the trivial bound of $\norm{\sum_{r\in [R+2]}E_r}_2 \le \epsilon n \cdot \sum_{r \in [R+2]} w_r$ in place of \Cref{lem:x_0_1_errs}, our error bound would be worse by a $\poly(n)$ factor. Adjusting $\epsilon$ to account for this, we would obtain \Cref{thm:main_thm} but with rank $O(\log n \log(n/\epsilon))$ as compared to $O(\log n \log(1/\epsilon))$, quadratically worse in $\log n$.

To prove \Cref{lem:x_0_1_errs}, we drop the last term $E_{R+1}$ in $\sum_{r\in [R+2]}E_r$ for now: it is a corner case handled separately in Lemma \ref{lem:err_last_bucket}.
For each $r\in [R+1]$, we define $E_{r,1}$ to be the restriction of error $E_r$ to its top-left $t_r/2 \times t_r/2$ block (zero elsewhere), capturing the error from approximating $H_{B_r}$ up to cutoff $t_r$. We first focus on bounding the norms of these restricted error matrices -- we will separately bound the contribution of the error after the cutoff later. We use the following claim, which translates entrywise bounds into norm bounds:
\begin{claim}\label{clm:entrywise_upper_bound}
For any $A,B\in\mathbb{R}^{n\times n}$ with $|A_{i,j}|\leq B_{i,j}$ for all $i,j\in [n]$, we have: $$\|A\|_2\leq \|B\|_2 \hspace{1em}\text{ and } \hspace{1em}\|A\|_F\leq \|B\|_F.\footnote{The proof follows trivially by definition for the Frobenius norm. For the spectral norm, it follows by observing that for any vector $v\in \mathbb{R}^{n}$ and index $i\in [n]$, $(Av)_i\leq (B|v|)_i$ where $|v|$ is obtained by taking entrywise absolute values of $v$. Thus $\|Av\|_2\leq \|B|v|\|_2$, and since $v$ was arbitrary $\|A\|_2\leq \|B\|_2$.}$$
\end{claim}
Note that each $E_{r,1} $ is entrywise bounded by $\epsilon w_r \cdot v_r v_r^\top$, where $ v_r \in \mathbb{R}^n$ has ones on its first $t_r/2$ entries and is zero everywhere else. See Figure \ref{fig:error_correlation} for a visualization. Thus Claim \ref{clm:entrywise_upper_bound} allows us to bound $
\|\sum_{r \in [R+1]} E_{r,1}\|_2 \le \|\sum_{r \in [R+1]} \epsilon w_r \cdot v_r v_r^T \|_2$, a sum of rank-one terms. Crucially, the vectors $v_r$ spanning each term satisfy,
\begin{equation}\label{eq:informal_correlation_decay}
\frac{\langle v_{r_1}, v_{r_2} \rangle}{\|v_{r_1}\|\|v_{r_2}\|} \leq 2^{-\Omega(|r_1 - r_2|)}\quad \forall r_1\neq r_2.
\end{equation}
This exponential \emph{correlation decay} across bucket errors allows us to beat the triangle inequality and prove:
\begin{equation*}
   \left\|\sum_{r\in [R+1]}\epsilon w_r \cdot v_r v_r^T\right\|_2 \approx  \max_{r\in [R+1]}\|\epsilon w_r \cdot v_r v_r^T\|_2 = \epsilon \max_{r\in [R+1]} w_r\cdot n/2^r.
\end{equation*}
Our formal proof considers the Gram matrix of the normalized vectors $v_r$ for $r\in [R+1]$ and uses the correlation decay to bound the contribution of off-diagonal terms. See Lemma \ref{lem:sparsification_errs} in  \Cref{sec:error_correlation_analysis}.

For the remaining error after the cutoffs for each bucket: $\sum_{r\in [R+1]} (E_r - E_{r,1})$, we leverage the decay of the entries of both $H_{B_r}$ and $H_{T_r}$ beyond the cutoffs to give strong error bounds. In particular, for all $r\in [R+1]$ and $i,j\in [n]$ such that $i\geq t_r/2$ or $j\geq t_r/2$, 
\begin{equation*}
    (E_r - E_{r,1})_{i,j} \approx \epsilon w_r\cdot e^{-2^{r-1}(\lambda/2n)(i+j)}.
\end{equation*}
By the above, we can entrywise upper bound $E_r - E_{r,1}$ by the rank-$1$ matrix \begin{align*}
    \epsilon w_r\cdot  v_n(e^{-2^{r-1}\lambda/2n})v_n(e^{-2^{r-1}(\lambda/2n)})^T.
\end{align*} 
We show that the moment vectors $v_n(e^{-2^{r-1}(\lambda/2n)})$ also exhibit correlation decay as in equation \eqref{eq:informal_correlation_decay}, ultimately letting us show:
\begin{align*}
    \left \|\sum_{r\in [R+1]} (E_r - E_{r,1}) \right \|_2 &\le \left\|\sum_{r\in [R+1]} \epsilon w_r\cdot  v_n(e^{-2^{r-1}(\lambda/2n)})v_n(e^{-2^{r-1}(\lambda/2n)})^T\right\|_2\\ &\approx  \max_{r\in [R+1]}\|\epsilon w_r v_n(e^{-2^{r-1}\lambda/2n})v_n(e^{-2^{r-1}(\lambda/2n)})^T\|_2\\&\approx \epsilon \max_{r\in [R+1]} w_r\cdot n/2^r.
\end{align*} 
The formal proof appears in  Lemma \ref{lem:small_entries_errs} in  \Cref{sec:error_correlation_analysis}. Combined with our bound for $\sum_{r \in [R+1]} E_r$ (\Cref{lem:sparsification_errs}), this bound yields Lemma \ref{lem:x_0_1_errs}. Finally, recall that the error upper bound of Lemma \ref{lem:x_0_1_errs}, combined with the spectral norm lower bound of \Cref{lem:hankel_spectral_lower_bd}, gives the $\epsilon \|H\|_2$ error bound of Theorem \ref{thm:main_thm}, as required. The full argument is presented in Lemma \ref{lem:x_0_1}.

\begin{figure}
\centering
\begin{tikzpicture}[scale=1.2]

\begin{scope}[shift={(0,0)}]
  \draw[pattern=north east lines, pattern color=black!60] (0,0) rectangle (2,2); 
  \node at (-0.6,1.15) { $\epsilon w_{0}\cdot$};

  \draw[<->] (2.1,0) -- (2.1,2);
  \node at (2.3,1) {\small $n$};
\end{scope}

\node at (2.7,1.2) {\large $+$};

\begin{scope}[shift={(3.8,0)}]
  \draw (0,0) rectangle (2,2); 
  \draw[pattern=north east lines, pattern color=black!60] (0,2) rectangle (1,1); 
  \node at (-0.6,1.15) { $\epsilon w_{1}\cdot$};

  \draw[<->] (1.1,1) -- (1.1,2);
  \node at (1.3,1.5) {\small $\frac{n}{2}$};
\end{scope}

\node at (6.3,1.2) {\large $+$};

\begin{scope}[shift={(7.4,0)}]
  \draw (0,0) rectangle (2,2); 
  \draw[pattern=north east lines, pattern color=black!60] (0,2) rectangle (0.5,1.5); 
  \node at (-0.6,1.15) { $\epsilon w_{2}\cdot$};

  \draw[<->] (0.6,1.5) -- (0.6,2);
  \node at (0.8,1.75) {\small $\frac{n}{4}$};
\end{scope}

\node at (9.9,1.2) {\large $+$};
\node at (10.4,1.2) {\large $\cdots$};

\end{tikzpicture}
\caption{Entrywise upper bounds on $E_{0,1}+E_{1,1}+E_{2,1}+\ldots$, the shaded and unshaded region in each term are entrywise $1$ and $0$ respectively. }
\label{fig:error_correlation}
\end{figure}
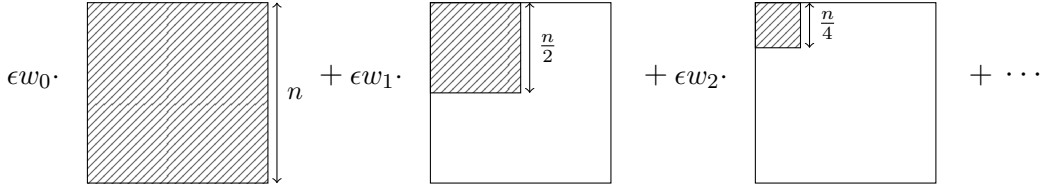

\smallskip

\noindent\textbf{Handling nodes with magnitude $> 1$.} 
Our previous discussion outlines the proof of \Cref{thm:main_thm} in the special case that all nodes $x \in X$ in the Fielder decomposition $H = V_X D_X V_X^T$ have magnitude at most $1$. \Cref{lem:x_0_1} states the theorem in this special case. We next outline how to handle $x \in X$ with $|x| > 1$.

We will use
Lemma \ref{lem:partition_x} to partition $X= X_1\cup X_2$, where $X_1$ contains all $x \in X$ with $|x| \le 1$ and $X_2$ contains all $x \in X$ with $|x| > 1$. To obtain a good Hankel low-rank approximation to $H$, it suffices to obtain good low-rank approximations to $H_{X_1}$ and $H_{X_2}$ and add them together. $H_{X_1}$ is handled by \Cref{lem:x_0_1}. For $H_{X_2}$ -- consider $x\in X_2$ (so $|x|>1$) and the corresponding moment vector $v_n(x) = [1,x,\ldots, x^{n-1}]$. Pulling out a factor of $x^{n-1}$, this moment vector can be written as $v_n(x)=x^{n-1}[1/x^{n-1},1/x^{n-2},\ldots,1] = x^{n-1} \cdot R v_n(1/x)$, where $R \in \R^{n \times n}$ is the matrix that just reverses the order of indices of any vector (see Definition \ref{def:row_reversal}). Moreover, the node $1/x$ satisfies $|1/x|\leq 1$. This insight allows us to reverse rows and columns of $H_{X_2}$, maintain invariance in spectral and Frobenius norms, and apply all previous ideas. In \Cref{sec:x_greater_than_1_reduction}, this reduction is presented in \Cref{lem:x_greater_than_1_reduction} followed by the full proof of \Cref{thm:main_thm}.

\subsection{Approximate Rank Lower Bound.}\label{sec:lower_bound_overview}


We now sketch the main ideas behind the proof of our approximate rank lower bound for PSD Hankel matrices, Theorem \ref{thm:epsilon_rank_lower_bound}. The full proof can be found in  \Cref{sec:lower_bd}. For the spectral norm lower bound, we will construct a PSD Hankel matrix $H$ with at least $\Omega(\log n +\log (1/\epsilon))$ eigenvalues larger than $\epsilon\|H\|_2$. Our lower bound for the Frobenius norm will follow similarly, using the known eigenvalue decay of PSD Hankel matrices. Depending on the relation between $\epsilon$ and $n$ we will have two different constructions. The first case is when $\epsilon \geq 1/n$, where it suffices to prove that there are $\Omega(\log n)$ large eigenvalues. The second is when $\epsilon < 1/n$, where it suffices to prove $\Omega(\log(1/\epsilon))$ large eigenvalues.

\smallskip

\noindent\textbf{Construction when $\epsilon \geq 1/n$.} We will construct PSD Hankel $H = V_X D_X V_X^T$ with $O(\log n)$ nodes $X\subset[0,1]$ and thus rank $O(\log n)$. We will show that $\norm{H}_2 = O(1)$ and that $H$ has $\Omega(\log n)$ eigenvalues that are $\Omega(1)$, proving the desired lower bound. The key idea is that if we choose sufficiently spaced out nodes, the columns of $V_X$ will have their inner products decay rapidly as the nodes get further apart. Roughly speaking, this ensures that each column of $V_X$ contributes to a separate eigenvalue, leading to $\Omega(\log n)$ eigenvalues $\approx 1$.

\smallskip

\noindent\textbf{Construction when $\epsilon \leq 1/n$.} In this case, we need to construct $H$ with $\Omega(\log 1/\epsilon)$ eigenvalues larger than $\epsilon\|H\|_2$. We show that the Hilbert matrix $H_n\in \mathbb{R}^{n\times n}$ defined as $(H_n)_{i,j}=1/(i+j+1)$ suffices. In particular, we can use a known bound, which shows that the minimum eigenvalue of $H_n$ is $\Omega(1/\exp(n))$ \cite{Beckermann:2000}. Via eigenvalue interlacing, we can extend this lower bound to show that the $k^{th}$ eigenvalue of $H$ is $\Omega(1/\exp(k))$. Setting $k = \Theta(\log(1/\epsilon))$ and observing that $\norm{H_n}_2 = O(1)$ then gives our lower bound. 
\subsection{Sublinear Time Robust Low Rank Approximation.}\label{sec:alg_overview}
Finally, we sketch the main ideas behind the sublinear time algorithm of \Cref{thm:const_factor_approx_main}. The full algorithm (Algorithm \ref{alg:noisy_hankel_recovery}), and the full proof are given in  \Cref{sec:alg_main}. Recall that we are given entrywise access to $B=H+E$ for a PSD Hankel matrix $H$ and arbitrary noise matrix $E$, and our goal is to recover a {Hankel} low rank approximation to $H$ with error $\epsilon \|H\|_F + O(\norm{E}_F)$. As discussed previously, using an arbitrarily small perturbation, we can assume $H$ to be positive definite, and thus assume it has a Fielder factorization $H = V_X D_X V_X^T$. For simplicity of exposition, we assume for now that all nodes $x \in X$ have $|x| \le 1$.

Our constructive proof of \Cref{thm:main_thm}, described in \Cref{sec:existence_overview}, implies that there is a fixed set of nodes $T$ (containing the exponentiated Chebyshev nodes for each bucket) of size $|T|=O(\log n\log(1/\epsilon))$ satisfying the following: for any positive definite Hankel $H$, there exists a diagonal matrix $D_T$ and Hankel matrix $H^*$ with only its first $O(\log(1/\epsilon))$ anti-diagonals nonzero, such that $\norm{H - (H^*+V_TD_TV_T^T)}_F \le \epsilon \norm{H}_F$. For a formal statement of this claim covering all cases, see Lemma \ref{lem:existence_of_vandermonde_good_low_rank_approximation}. Since the nodes $T$ are known, the only unknowns are the diagonal matrix $D_T$ and the sparse Hankel matrix $H^*$. Thus, our approach will be to solve the matrix regression problem $\min_{H^*, D_T}\|B-(H^*+V_TD_TV_T)\|_F$ in sublinear time. We will do this via a hybrid ridge leverage score sampling \cite{el2014fast,cohen2017input} and anti-diagonal averaging approach. 

Since the $H^*$ only contributes to the first $O(\log(1/\epsilon))$ anti-diagonals of $B$, our algorithm can simply choose $H^*$ so that the value of our Hankel low-rank approximation on each of these anti-diagonals equals the average value on the corresponding anti-diagonal in $B$. This is the optimal error solution under the Frobenius norm. 

We estimate the diagonal matrix $D_T$ by solving a two-sided sampled ridge regression problem, similar to \cite{EldarLiMusco:2020}. We first derive universal upper bounds on the ridge leverage scores of any Vandermonde matrix $V_X$, independent of $X$ (see Lemma \ref{lem:vandermonde_ridge_lev_score_bd} in  \Cref{sec:vandermonde_ridge_lev_score_bd}). To do so, we apply techniques similar to those described in \Cref{sec:existence_overview} to approximate the column span of $V_X$ using low-degree polynomials. We then apply known leverage score bounds for polynomials \cite{meyer2023near}. With our ridge leverage score bounds in hand, we can subsample and approximately solve the regression problem for $D_T$ in sublinear time. The complete algorithm is described in Algorithm \ref{alg:noisy_hankel_recovery}, and the proof of its guarantees, yielding Theorem \ref{thm:const_factor_approx_main}, is given in \Cref{sec:sublinear_time_recovery}.

\section{Preliminaries and Notation.}\label{sec:prelim}
We first introduce notation and preliminaries used throughout the paper.
\subsection{Notation}\label{sec:notation}
For any non-negative integer $n$, we let $[n]=\{0,\ldots, n-1\}$. For integers $a < b$, we let $[a:b]=\{a,a+1,\ldots, b\}$. For real numbers $a\leq b$, $[a,b]$ denotes the continuous interval from $a$ to $b$. For functions $f,g:\mathbb{R}\rightarrow \mathbb{R}$, we use $f(x)\lsim g(x)$ to indicate that there exists a constant $C>0$ such that $f(x)\leq Cg(x)$ for all $x\in \mathbb{R}$.
For any $x>0$ we use $\log(x) = \log_{2}(x)$ and $\ln(x) = \log_{e}(x)$. 

For any two matrices $A,B$ with matching dimensions, we let $A\circ B$ denote their Hadamard (i.e., entrywise) product. For integers $0\leq a_1<a_2\leq n$ and $0\leq b_1<b_2\leq n$ we let $A[a_1:a_2,b_1:b_2]$ denote the submatrix of $A$ containing rows from $a_1$ to $a_2$ and columns from $b_1$ to $b_2$. We let $\|A\|_2 = \max_{x\in \mathbb{R}^n} \|Ax\|_2/\|x\|_2$ denote the spectral norm and $\|A\|_F=\sqrt{\sum_{i \in [n]} \sum_{j \in [n]} A_{i, j}^2}$ denote the Frobenius norm. 

A symmetric matrix $A\in \mathbb{R}^{n\times n}$ is positive semidefinite (PSD) if for all $x\in \mathbb{R}^n$, $x^TAx \geq 0$, and positive definite if for all $x \in \R^{n}$, $x^T A x > 0$. Let $\lambda_{1}(A)\geq \ldots \geq \lambda_{n}(A)\geq 0$ denote its eigenvalues. Let $\preceq$ denote the Loewner ordering: $A\preceq B$ if and only if $B-A$ is PSD. Let $A= U\Sigma V^T$ denote the compact singular value decomposition of $A$, and when $A$ is symmetric PSD, note that $U \Sigma U^T$ is its eigendecomposition. Let $A^{1/2} = U\Sigma^{1/2}$ denote its matrix square root, where $\Sigma^{1/2}$ is obtained by taking the elementwise square root of $\Sigma$. Let $A_k = U_k \Sigma_k V_k^{T}$ denote the projection of $A$ onto its top $k$ singular vectors. Here, $\Sigma_k\in \mathbb{R}^{k\times k}$ is the diagonal matrix containing the $k$ largest singular values of $A$, and $U_k, V_k \in \mathbb{R}^{n\times k}$ denote the corresponding $k$ left and right singular vectors. Note that $A_k$ is the optimal rank $k$ approximation to $A$ in the spectral and Frobenius norms. That is, $A_{k} = \argmin_{\text{rank $k$ }\wh{A}}\|A-\wh{A}\|_2$ and $A_{k} = \argmin_{\text{rank $k$ }\wh{A}}\|A-\wh{A}\|_F$. 

\subsection{Basic results regarding Hankel matrices.}\label{sec:prelim_hankel}
We now present two basic lemmas regarding Hankel matrices. We begin by first stating a simple consequence of the eigenvalue decay bounds for positive definite Hankel matrices of \cite{beckermann2017singular}.
\begin{lemma}[Implication 1 of Corollary 5.5 of \cite{beckermann2017singular}]\label{lem:townsend_corrolary}
For any positive definite Hankel matrix $H\in \mathbb{R}^{n\times n}$ and $\epsilon>0$, for $k=O(\log n\log(1/\epsilon))$ the best rank $k$ approximation $H_k$ satisfies $\|H-H_k\|_F\leq \epsilon \|H\|_F$.     
\end{lemma}
\begin{proof}
For any $j\in [n]$, Corollary 5.5 of \cite{beckermann2017singular} implies that there is a constant $C$ such that $\lambda_{j+k}(H)\leq C^{-k/2\log n}\lambda_{j}(H)$ for even $k$. This implies:
\begin{align*}
    \|H-H_k\|_F ^2 &= \lambda_{k+1}^2(H)+\sum_{j=2}^{n-k}\lambda_{k+j}^2(H)\\
    &\leq C^{-k/\log n}\lambda_{1}^2(H) + C^{-k/\log n}\cdot \sum_{j=2}^{n-k}\lambda_{j}^2(H)\\
    & \leq C^{-k/\log n}\|H\|_F^2=\epsilon^2\|H\|_F^2, 
\end{align*}
where in the last line we used $k=O(\log n\log(1/\epsilon))$.
\end{proof}
We will also need the following lemma, which again follows easily from the eigenvalue decay bounds of \cite{beckermann2017singular}.
\begin{lemma}[Implication 2 of Corollary 5.5 of \cite{beckermann2017singular}]\label{lem:hankel_spectral_to_frob}
For any positive definite Hankel matrix $H\in \mathbb{R}^{n\times n}$, we have $\|H\|_F =O(\sqrt{\log n}\cdot \|H\|_2)$. Further, for any rank $k \le n$, the best rank-$k$ approximation $H_k$ to $H$ satisfies $\|H-H_k\|_F = O(\sqrt{\log n}\cdot \|H-H_k\|_2)$
\end{lemma}
\begin{proof}
For any $j\in [n]$, Corollary 5.5 of \cite{beckermann2017singular} implies that there is a constant $C$ such that $\lambda_{k+2j}(H)\leq C^{-j/\log n}\lambda_{k}(H)$. This implies:
\begin{align*}
    \|H-H_k\|_F ^2 &= \lambda_{k+1}^2(H)+\sum_{j=k+2}^n\lambda_{j}^2(H)\\
    &\leq 2\lambda_{k+1}^2(H) + 2\sum_{j=1}^{(n-k-1)/2}\lambda_{k+1+2j}^2(H)\\
    &\leq 2\lambda_{k+1}^2(H) + 2\lambda_{k+1}^2(H)\sum_{j=1}^{\infty} C^{-j/\log n}\\
    & = O(\log n \cdot \lambda_{k+1}^2(H)) = O(\log n \cdot \|H-H_k\|_2^2). 
\end{align*}
Repeating the above for $\|H\|_F^2$ completes the lemma.
\end{proof}
Next, we present the proof of Lemma \ref{lem:partition_x}, which we repeatedly use throughout the paper.
\begin{replemma}{lem:partition_x}
Consider $H$ defined entrywise as $H_{i,j}=\sum_{x\in X}a_x x^{i+j}$ with $a_x\geq 0$ for all $x\in X$. Then we can write $H=V_{X}D_XV_{X}^T$ with positive diagonal $D_X=diag(\{a_x\}_{x\in X})$. Furthermore for any partition $X=X_0\cup\ldots \cup X_{l-1}$, let $H_{X_m}$ be defined entrywise as $(H_{X_m})_{i,j} = \sum_{x\in X_m}a_x x^{i+j}$ for all $m\in [l]$. Then we have that $H_{X_m} = V_{X_m}D_{X_m}V_{X_m}^T$ and $D_m = diag(\{a_x\}_{x\in X_m})$, $H\succeq H_{X_m}$ for all $m\in [l]$, and $H=\sum_{m\in [l]}H_{X_m}$.
\end{replemma}
\begin{proof}
Since $D_X\succeq 0$, $V_X D_X V_X^T = (V_X \sqrt{D_X})(V_X \sqrt{D_X})^T$ which is an outer product of $V_X \sqrt{D_X}$ with itself, thus it equals $\sum_{x\in X}a_x v_n(x)v_n(x)^T = \sum_{m\in [l]}\sum_{x\in X_m}a_x v_n(x)v_n(x)^T$ where $v_n(x)$ is the moment vector of Definition \ref{def:moment_vector}. We also have that $\sum_{x\in X_m}a_x v_n(x)v_n(x)^T =  (V_{X_m} \sqrt{D_{X_m}})(V_{X_m} \sqrt{D_{X_m}})^T = V_{X_m}D_{X_m}V_{X_m}^T=H_{X_m}$ for all $m\in [l]$. Furthermore, $H\succeq H_{X_m} \succeq 0$ since $H$ and $H_{X_m}$ are PSD for all $m\in [l]$ . And finally $H=\sum_{m\in [l]}H_{X_m}$ follows easily by this discussion.
\end{proof}

\subsection{Linear algebra and polynomial sensitivity tools}\label{sec:prelim_tools}
We next present various preliminary linear algebraic results and tools to handle low-degree polynomials. To prove spectral norm upper bounds on various error matrices encountered in this paper, we will need the following inequality:
\begin{lemma}[Corollary 2.3.2 in \cite{golub2013matrix}]\label{lem:spectral_norm_bound}
For any $A\in \mathbb{R}^{m\times n}$ we have $\|A\|_2\leq \sqrt{\|A\|_1\|A\|_{\infty}}$ where $\|A\|_1,\|A\|_{\infty}$ are the max $\ell_1$ norms of any columns or row of $A$, respectively.
\end{lemma}
We next define the ridge leverage scores, which have been used in prior work on fast approximate kernel ridge regression \cite{el2014fast}, low-rank approximation \cite{cohen2017input}, and standard leverage score computation \cite{kapralov2017single,li2013iterative}.
\begin{definition}[$\gamma$-Ridge Leverage Scores]\label{def:ridge_lev_score}
For any $A\in \mathbb{R}^{n\times d}$ and ridge parameter $\gamma >0$, let $\tau_{i,\gamma}(A)$ be the $i^{th}$ $\gamma$-ridge leverage score defined as $\displaystyle \tau_{i,\gamma}(A) = \max_{y\in \mathbb{R}^{d}} \frac{(a_i^T y)^2}{(\|Ay\|_2^2 + \gamma\|y\|_2^2)},$ where $a_i$ is the $i^{th}$ row of $A$.
\end{definition}
We use the following ridge leverage score sampling based matrix approximation guarantee:
\begin{lemma}[Theorem 5 of \cite{cohen2017input}]\label{lem:ridge_lev_score_sampling_guarantee}
For any $\gamma>0$ and $\delta \in (0,1/8)$, given ridge leverage score approximations $\wt{\tau}_{i,\gamma}\geq \tau_{i,\gamma}(A)$ for all $i\in [n]$. Let $p_i = \wt{\tau}_{i,\gamma}/(\sum_{i\in [n]}\wt{\tau}_{i,\gamma})$ and $s= C\log(n/\delta)\sum_{i\in [n]}\wt{\tau}_{i,\gamma}$ for some sufficiently large constant $C$. Let $S\in \mathbb{R}^{s\times n}$ be the sampling matrix corresponding to sampling $s$ rows of $A$, where in each sample each row of $A$ is sampled with probability $p_i$ and rescaled by $1/\sqrt{sp_i}$. Then, with probability at least $1-\delta$, 
\begin{equation*}
    \|SAy\|_2^2+\gamma \|y\|_2^2 =\left(1\pm \frac{1}{4}\right)\left(\|Ay\|_2^2 +\gamma \|y\|_2^2\right).
\end{equation*}
\end{lemma}
Next, we use the Markov brother's inequality, which can be used to bound the slope of bounded polynomials, and thus their leverage scores.
\begin{lemma}[Markov brother's inequality, Theorem 2.1 in \cite{govil1999markov}\footnote{This version of the statement is taken from Theorem 3.1 in \cite{meyer2023near}}]\label{lem:markov_brother}
Suppose $q(t)$ is a polynomial of degree at most $d$ such that $\max_{t\in [-1,1]}|q(t)|\leq 1$. Then $\max_{t\in [-1,1]}|q'(t)|\leq d^2$, where $q'(t)$ denotes the derivative.
\end{lemma}
We will also make heavy use of the Chebyshev nodes, defined below:
\begin{definition}[Chebyshev Nodes]\label{def:chebyshev_nodes}
For any integer $l>0$ let the degree $l$ Chebyshev nodes $t_0,\ldots, t_l$ in $[-1,1]$ be defined as $t_i = \cos(i\pi/l)$. For any finite interval $[a,b]$, the Chebyshev nodes in this interval are obtained by scaling and shifting the range $[-1,1]$ to $[a,b]$.
\end{definition}
We will need the following result from \cite{liu2022robust}, which expresses arbitrary moment vectors as a sum over those evaluated at Chebyshev nodes, to sparsify sums of moment vectors.
\begin{lemma}[Claim F.5 and Corollary F.7, \cite{liu2022robust}]\label{lem:chebyshev_sparsification}
Let $p(x)$ be a degree-$l$ polynomial such that $|p(t_i)| \leq 1$ on all degree-$l$ Chebyshev nodes $t_i \in [a,b]$ (see Definition \ref{def:chebyshev_nodes}).\footnote{\cite{liu2022robust} present this claim for interval $[-1,1]$, but one can see that by shifting and scaling, it applies to any finite interval $[a,b]$.} Then $|p(x)| \leq 2l$ for all $x \in [a,b]$. Moreover, for any $x \in [a,b]$, $v_l(x)$ (\Cref{def:moment_vector}) lies in the convex hull of $\{\pm 2l\cdot v_l(t_i)\}_{i=0}^l$.
\end{lemma}

\section{Existence of a Good Hankel Low Rank Approximation.}\label{sec:existence_main}
In this section, we prove the existence result of Theorem \ref{thm:main_thm}, which we restate below.
\begin{reptheorem}{thm:main_thm}[Existence of Hankel low rank approximation]
    For any PSD Hankel matrix $H\in \mathbb{R}^{n\times n}$ and any $\epsilon>0$, there exists a rank $O(\log n \log(1/\epsilon))$ Hankel matrix $\wh{H}$ such that $$\|H-\wh{H}\|_2\leq \epsilon \|H\|_2\text{ and }\|H-\wh{H}\|_F\leq \epsilon \|H\|_F.$$
\end{reptheorem}
We first prove the theorem in the special case when, writing $H$ in its Fielder decomposition $H = V_X D_X V_X^T$, all nodes $x \in X$ satisfy $|x|\leq 1$. This restricted case is stated below:
\begin{lemma}\label{lem:x_0_1}
For PSD Hankel $H= V_X D_X V_X^T$, if all $x\in X$ satisfy $|x|\leq 1$, then there exists a Hankel matrix $\widehat{H}$ of rank $O(\log(n)\log (1/\epsilon))$ such that $\|H-\widehat{H}\|_2\leq \epsilon \|H\|_2$ and $\|H-\widehat{H}\|_F\leq \epsilon \|H\|_F$.
\end{lemma}
To prove \Cref{lem:x_0_1}, we will follow the technical overview of \Cref{sec:existence_overview}. We first prove entrywise error bounds for approximating $H$ via our node bucketing approach (Definition \ref{def:bucketing}) and the bucket sparsification result (Lemma \ref{lem:bucket_sparsification}). This sparsification result is proven in \Cref{sec:bucket_sparsification}. We next prove spectral and Frobenius norm lower bounds for $H$ (Lemma \ref{lem:hankel_spectral_lower_bd}) in  \Cref{sec:hankel_spectral_lower_bd}. In \Cref{sec:error_correlation_analysis}, we then give tight bounds on the spectral and Frobenius norm errors based on the entrywise errors introduced due to sparsification (Lemma \ref{lem:x_0_1_errs}). We complete the proof of Lemma \ref{lem:x_0_1} by relating these error bounds to the corresponding spectral and Frobenius norm lower bounds for $H$ from Lemma \ref{lem:hankel_spectral_lower_bd}.

In \Cref{sec:x_greater_than_1_reduction} we show a reduction from the case when $|x|>1$ for all $x\in X$ to the case when $|x|\leq 1$ for all $x\in X$, which combined with \Cref{lem:x_0_1} proves:
\begin{lemma}\label{lem:x_1_above}
For PSD Hankel $H= V_X D_X V_X^T$, if all $x\in X$ satisfy $|x|> 1$, then there exists a Hankel matrix $\widehat{H}$ of rank $O(\log(n)\log (1/\epsilon))$ such that $\|H-\widehat{H}\|_2\leq \epsilon \|H\|_2$ and $\|H-\widehat{H}\|_F\leq \epsilon \|H\|_F$.
\end{lemma}
Combining Lemmas \ref{lem:x_0_1} and \ref{lem:x_1_above}, we can then prove Theorem \ref{thm:main_thm}. This proof is presented in  \Cref{sec:x_greater_than_1_reduction}. 

\subsection{Bucket Sparsification}\label{sec:bucket_sparsification}
We now present our main idea of node bucketing, followed by the bucket sparsification result of Lemma \ref{lem:bucket_sparsification}.
\begin{definition}[Log-scale bucketing]\label{def:bucketing}
Let $\lambda = 2 \log(1/\epsilon)$ and $R = \log(\ln(1/0.9)n/\lambda)$. For $X\subset [-1,1]$ with corresponding weights $\{a_x\}_{x\in X}$, define ``buckets'' $B_0,\ldots, B_{R+1}$ which form a partition of $X$ as follows,
\begin{align*}
    B_{0} &= \left\{x\in X: \log(1/|x|)\in \left[0,\frac{\lambda}{n}\right]\right\},\\
        B_r &= \left\{x\in X:  \log(1/|x|)\in \left[\frac{2^{r-1}\lambda}{n},\frac{2^{r}\lambda}{n}\right]\right\} \text{ for } r \in [R+1], r\geq 1, \text{ and }\\
     B_{R+1} &= \{x\in X:  \log(1/|x|)\in [\ln(1/0.9),\infty)\}.
\end{align*}
Let $w_r = \sum_{x\in B_r}a_x$ denote the ``weight'' of $B_r$ and $t_r=n/2^{r-1}$ denote the ``cutoff''  for all $r\in [R+2]$. 
\end{definition}
We now present the formal statement of our bucket sparsification lemma and its proof.

\begin{replemma}{lem:bucket_sparsification}[Bucket sparsification]
Given $\epsilon>0$, let $l=O(\log(1/\epsilon))$. Fix $r\in [R+1]$ and consider bucket $B_r$ as per Definition \ref{def:bucketing}. Let $T_r=\{\pm e^{-t_{0,r}},\ldots,\pm e^{-t_{l,r}}\}$ where $\{t_{i,r}\}_{i\in [l+1]}$ is the set of $l+1$ scaled and shifted Chebyshev nodes (see Definition \ref{def:chebyshev_nodes}) in $[2^{r-1}\lambda/n,2^r\lambda/n]$ when $r\geq 1$ and in $[0,\lambda/n]$ when $r=0$. Then there exist weights $a^+_{i,r},a^-_{i,r}$ corresponding to $e^{-t_{0,r}},-e^{-t_{0,r}}$ respectively for each $i\in [l+1]$ such that,
\begin{equation*}
    \sum_{i\in [l+1]}|a^+_{i,r}|+|a^-_{i,r}| = O(l\cdot w_r). 
\end{equation*}
Moreover, for $H_{B_r}(t) = \sum_{x\in B_r}a_x x^t$ and $H_{T_r}(t)=\sum_{i\in [l+1]}a^+_{i,r} e^{-t_{i,r} t}+(-1)^t\sum_{i\in [l+1]}a^-_{i,r} e^{-t_{i,r} t}$,
\begin{equation*}
    |H_{B_r}(t) - H_{T_r}(t)|\leq \epsilon w_r\quad \forall t\in [2n-1].
\end{equation*}
\end{replemma}
\begin{proof}
Let $V_{B_r}$ be the Vandermonde matrix corresponding to nodes in bucket $B_r$. Let $D_{B_r} = diag(\{a_x\}_{x\in B_r})$ and $H_{B_r} = V_{B_r}D_{B_r}V_{B_r}^T$. For any $r\geq 1$, since each $x\in B_r$ satisfies $|x|\leq e^{-2^{r-1}\lambda/n}$, we have that $|x|^{t_r}\leq e^{-2^{r-1}\lambda t_r/n} = e^{-\lambda} = \epsilon$. That is, for all $t$ beyond the cutoff $t_r$ and $x \in B_r$, $|x|^t\le \epsilon$. Let $B_r^+ = \{x\in B_r: x\geq 0\}$ and $B_r^- = B_r \setminus B_r^+$ be the sets of positive and negative nodes in $B_r$ respectively. Let $y_x = \ln{1/|x|}$ for $x\in B_r$. Let $H_{B^+_r}(t) = \sum_{x\in B_r^+}a_x x^t = \sum_{x\in B_r^+}a_x e^{-y_x t}$ and $H_{B^-_r}(t) = \sum_{x\in B_r^-}a_x(-1)^t |x|^t = \sum_{x\in B_r^-}a_x (-1)^t e^{-y_x t}$. Finally, let  $H_{B_r}(t)= H_{B^+_r}(t)+H_{B^-_r}(t)$. By expanding the decomposition $H_{B_r} = V_{B_r}D_{B_r}V_{B_r}^T$ it can be seen that $(H_{B_r})_{i,j} = H_{B_r}(i+j)$ for all $i,j\in [n]$. Now using Taylor expansion for $e^{-y_x t}$ for all $x\in B^+_r$ we have the following for $l=50\lambda = 100\log(1/\epsilon)$:
\begin{align*}
    \sum_{x\in B_r^+}a_x e^{-y_x t} &=\sum_{x\in B_r^+} a_x \sum_{m=0}^{\infty} (-y_x t)^m/m!\\
    & = \left(\sum_{x\in B_r^+} a_x \sum_{m=0}^{l} y_x^m (-t)^m/m!\right) +\left(\sum_{x\in B_r^+} a_x \sum_{m=l}^{\infty} y_x^m (-t)^m/m!\right). \numberthis \label{eq:taylor_exp_plus}
\end{align*}
Similarly, we have for $B_r^-$:
\begin{align*}
    \sum_{x\in B_r^-}a_x (-1)^te^{-y_x t} &=  (-1)^t \sum_{x\in B_r^-} a_x \sum_{m=0}^{\infty} (-y_x t)^m/m!\\
    & = (-1)^t\left(\sum_{x\in B_r^-} a_x \sum_{m=0}^{l} y_x^m (-t)^m/m!\right) +(-1)^t\left(\sum_{x\in B_r^-} a_x \sum_{m=l}^{\infty} y_x^m (-t)^m/m!\right). \numberthis \label{eq:taylor_exp_minus}
\end{align*}

We now show that the sum of the second terms in equations \eqref{eq:taylor_exp_plus} and \eqref{eq:taylor_exp_minus} is at most $\epsilon\sum_{x\in B_r}a_x$ for $t\leq t_r$. When $r\geq 1$, since $y_x\in [2^{r-1}\lambda/n,2^r\lambda/n]$ for all $x\in B_r$ and $t\leq t_r$, we have that $|-y_xt|= 2^r\lambda/n \cdot 2n/2^r = 2\lambda=4\log(1/\epsilon)$. Moreover, when $r=0$ we have that $y_x\in [0,\lambda/n]$. Thus $|y_xt|\leq \lambda = 2\log(1/\epsilon)$ for all $t\in [t_r-1]=[2n-1]$ as $t_r = 2n$ for $r=0$. Now using Stirling's approximation for $m!$ when $m\geq l$, we have that $\log(m!) = m\log(m) - m\log(e)+\log(\sqrt{2\pi \log(m)})+O(1/m)\geq m\log(m/e)$. This implies that $m! \geq (m/e)^m$. Thus, we have the following bound for the second term in \eqref{eq:taylor_exp_plus}:
\begin{align*}
    \left|\sum_{m=l}^{\infty}\sum_{x\in B_r^+} a_x(-y_x t)^m/m!\right| &\leq \sum_{m=l}^{\infty}\sum_{x\in B_r^+} a_x |y_x t|^m/m!\\
    &\leq \sum_{m=l}^{\infty}\sum_{x\in B_r^+} a_x (4e\log(1/\epsilon)/l)^m\\
    &= \sum_{m=l}^{\infty}\sum_{x\in B_r^+} a_x (4e\log(1/\epsilon)/100\log(1/\epsilon))^m\\
    &\leq \sum_{x\in B_r^+} a_x \sum_{m=l}^{\infty} 0.1^m \leq \epsilon\sum_{x\in B_r^+}a_x. \numberthis \label{eq:taylor_approx_error}
\end{align*}
Similarly, $\left|(-1)^t \sum_{m=l}^{\infty}\sum_{x\in B_r^-} a_x(-y_x t)^m/m!\right|\leq \epsilon \sum_{x\in B_r^-}a_x$. Taking the sum and applying the triangle inequality, we obtain, for all $t\in [t_r-1]$,
\begin{align*}
    |\sum_{x\in B_r}a_x x^t - \left(\sum_{x\in B_r^+} a_x \sum_{m=0}^{l} y_x^m (-t)^m/m!\right)&-(-1)^t\left(\sum_{x\in B_r^-} a_x \sum_{m=0}^{l} y_x^m (-t)^m/m!\right)| \\ &\leq \epsilon \sum_{x\in B_r} a_x = \epsilon w_r. \numberthis \label{eq:error_bound_taylor}
\end{align*}

Next we apply Lemma \ref{lem:chebyshev_sparsification} to obtain that for any $r\in [R+1]$ and $r\geq 1$ and interval $[2^{r-1}\lambda/n,2^r \lambda /n]$ if we let $t_{0,r},\ldots,t_{l,r}$ be the $l+1$ scaled and shifted order $l$ Chebyshev nodes (see Definition \ref{def:chebyshev_nodes}) in $[2^{r-1}\lambda/n,2^r \lambda /n]$ then for any $y\in [2^{r-1}\lambda/n,2^r \lambda /n]$ the moment vector (see Definition \ref{def:moment_vector}) $v_l(y)$  lies in the convex hull of $\{\pm 2l v_l(t_{0,r}),\ldots, \pm 2l v_l(t_{l,r})\}$. For $r=0$, the $l+1$ Chebyshev nodes are considered in the interval $[0,\lambda/n]$. Instantiating this for $l=O(\lambda) = O(\log(1/\epsilon))$, we have that there exists weights $a_{x,i}$ for all $i\in [l+1]$ such that for all $x\in B_r$ and $r\in [R+1]$ we have
\begin{equation*}
    v_l(y_x) = \sum_{i\in [l+1]} 2l a_{x,i} v_l(t_{i,r}),
\end{equation*}
and $\sum_{i\in [l+1]}|a_{x,i}| \leq 1$. This implies that, for all $x\in B_r$ and $r\in [R+1]$, and for all $t\in [2n-1]$,
\begin{align*}
v_l(y_x)^T v(t) &= (2l\sum_{i\in [l+1]}a_{x,i} v_l(t_{i,r}))^T v(t),\\
\implies \sum_{m=0}^l y_x^m(-t)^m/m! &= 2l \sum_{i\in [l+1]} a_{x,i} \sum_{m=0}^l  t_{i,r}^m (-t)^m/m!.\numberthis \label{eq:equality_first_terms_taylor}    
\end{align*}
Now if we consider $e^{-t_{i,r} t}$ for some $i\in [l+1]$, then since $t_{i,r}\in [2^{r-1}\lambda/n,2^r\lambda /n]$ we have that $|t_{i,r} t|\leq 2\lambda=4\log(1/\epsilon)$ for all $t\in [t_r-1]$. Thus, similar to the error bound in  \eqref{eq:error_bound_taylor}, we obtain the following by truncating the Taylor approximation of $e^{-t_{i,r} t}$ to $l=100\log(1/\epsilon)$ terms for all $i\in [l+1]$ and $t\in [t_r-1]$ for all $r\in [R+1]$:
\begin{equation*}
\left |e^{-t_{i,r} t}-\sum_{m=0}^l t_{i,r}^m (-t)^m/m!\right|=\left|\sum_{m=l}^{\infty} t_{i,r}^m (-t)^m/m!\right|\leq \epsilon.
\end{equation*}
Combining the above with the equality of \eqref{eq:equality_first_terms_taylor}, we obtain that, for all $t\in [t_r-1]$ and $r\in [R+1]$,
\begin{align*}
| \sum_{x\in B_r}a_x x^t - 2l\sum_{x\in B^+_r}a_x \sum_{i\in[l+1]}a_{x,i}e^{-t_{i,r} t} &- 2l(-1)^t\sum_{x\in B^-_r}a_x\sum_{i\in [l+1]}a_{x,i} e^{-t_{i,r} t}| \\
    &\leq \left|\sum_{m=l}^{\infty}\sum_{x\in B_r^+} a_x(-y_x t)^m/m!+(-1)^t \sum_{m=l}^{\infty}\sum_{x\in B_r^-} a_x(-y_x t)^m/m!\right|\\
    &+ \left| 2l\sum_{x\in B^+_r}a_x \sum_{i\in [l+1]}a_{x,i}\sum_{m=l}^{\infty} t_{i,r}^m (-t)^m/m! \right| \\&+ \left| 2l(-1)^t\sum_{x\in B^-_r}a_x \sum_{i\in [l+1]}a_{x,i}\sum_{m=l}^{\infty} t_{i,r}^m (-t)^m/m!\right|\\
    &\leq \epsilon \sum_{x\in B_r}a_x + 4l\epsilon \sum_{x\in B_r}a_x\sum_{i\in [l+1]}|a_{x,i}|\\
    &\leq 5l \epsilon \sum_{x\in B_r}a_x.
\end{align*}
Let $a^+_i = 2l\sum_{x\in B^+_{r}}a_x a_{x,i}$ and $a^{-}_i = 2l\sum_{x\in B^-_r}a_x a_{x,i}$. Thus $\sum_{i\in [l+1]}|a^+_i|+|a^-_i|\leq 2l\cdot w_r$. Moreover, with this notation, we have, for all $t\in [t_r-1]$ and $r\in [R+1]$,
\begin{equation*}
    \left | \sum_{x\in B_r}a_x x^t - \left(\sum_{i\in [l+1]} a^+_i e^{-t_{i,r} t} + (-1)^t\sum_{i\in [l+1]} a^+_i e^{-t_{i,r} t}\right)\right |\leq 5l \epsilon \sum_{x\in B_r} a_x.
\end{equation*}
Moreover, for all $r\geq 1$ since each $t_{i,r}\in [2^{r-1}\lambda/n,2^r\lambda /n]$, $|e^{-t_{i,r} t_r}|\leq \epsilon$ for all $i\in [l+1]$, for all $t\in [t_r,2n-1]$,
\begin{align*}
   \left | \sum_{x\in B_r}a_x x^t - \left(\sum_{i\in [l+1]} a^+_i e^{-t_{i,r} t} + (-1)^t\sum_{i\in [l+1]} a^-_i e^{-t_{i,r} t}\right)\right |&\leq \sum_{x\in B_r}a_x |x^t| + \sum_{i\in [l+1]} |a^+_i| |e^{-t_{i,r} t}| \\&+\sum_{i\in [l+1]} |a^-_i| |e^{-t_{i,r} t}|\\
   &\leq \epsilon\sum_{x\in B_r}a_x + 2l\epsilon\sum_{x\in B^+_r}a_x + 2l\epsilon\sum_{x\in B^-_r}a_x \\&\leq 5l\epsilon \sum_{x\in B_r}a_x.
\end{align*}
The above case does not happen for $r=0$ as $t_r=2n$. Thus, overall we have, for all $r\in [R+1]$ and $t\in [2n-1]$,
\begin{equation*}
     \left | \sum_{x\in B_r}a_x x^t - \left(\sum_{i\in [l+1]} a^+_i e^{-t_{i,r} t} + (-1)^t\sum_{i\in [l+1]} a^-_i e^{-t_{i,r} t}\right)\right |\leq 5l\epsilon \sum_{x\in B_r} a_x.
\end{equation*}
Since $l=O(\log(1/\epsilon))$, replacing $\epsilon$ with $\epsilon^2$ we have that $\epsilon^2 \log (1/\epsilon)\leq \epsilon$. Thus, the above bound is at most $\epsilon w_r$ and $l$ remains $O(\log(1/\epsilon))$. This completes the proof of the Lemma.
\end{proof}
Next, we move on to lower bounding the spectral and Frobenius norms of $H$.
\subsection{Lower Bounding the Spectral and Frobenius Norms of \texorpdfstring{$H$}{H}}\label{sec:hankel_spectral_lower_bd}
We now prove the norm lower bound result of Lemma \ref{lem:hankel_spectral_lower_bd}. We first state the lemma for the spectral and Frobenius norms, followed by its proof.
\begin{replemma}{lem:hankel_spectral_lower_bd}[Spectral and Frobenius norm lower bounds on $H$]
For PSD Hankel matrix $H = V_X D_X V_X^T\in \mathbb{R}^{n\times n}$ with all $x\in X$ satisfying $|x|\leq 1$, apply Definition \ref{def:bucketing} to it. Then we have
\begin{align*}
\|H\|_2 &=\Omega\left(\max_{r\in [R+2]}w_r \cdot n/(\lambda 2^r)\right),\\
\|H\|_F & = \Omega\left(\sqrt{\sum_{r_1, r_2=0}^{R}w_{r_1}w_{r_2} \cdot (n/2^{\max\{r_1,r_2\}})^2}+w_{R+1}\right).
\end{align*}
\end{replemma}
\begin{proof}
First, consider a $r\in [R+1]$, and let $B^+_r = \{x\in B_r: x\geq 0\}$ and $B^-_r = B_r\setminus B_r$. Now let $H_{B^+_r} = V_{B^+_r}D_{B^+_r}V_{B^+_r}^T$ and our goal will be to first lower bound $\|H_{B^+_r}\|_2$. 

\smallskip

\noindent\textbf{Lower bounding the contribution of elements of $B^+_r$.} For this we will analyze the entries for $H_{B^+_r}v$ for any vector $v\in \mathbb{R}^n$ with non-negative entries. For any index $i\in [n]$ we have, using the definition of $B_r$ (\Cref{def:bucketing}),
\begin{align*}
    (H_{B^+_r}v)_i^2 &=  \left(\sum_{j=0}^{n-1} \left(\sum_{x\in B^+_r}a_x x^{i+j}\right)v_j \right)^2\\
    &\geq \left(\sum_{j=0}^{n-1} \left(\sum_{x\in B^+_r}a_x e^{-(2^r\lambda/n)\cdot(i+j)}\right)v_j \right)^2\\
    &=\left(\left(\sum_{x\in B^+_r}a_x\right)\cdot e^{-2^r\lambda i/n}\cdot \left(\sum_{j=0}^{n-1}  e^{-2^r\lambda j/n}v_j\right) \right)^2\\
    &= (w^+_r)^2 \cdot e^{-2^{r+1}\lambda i/n} \cdot \left(\sum_{j=0}^{n-1}  e^{-2^r\lambda j/n}v_j\right)^2,
\end{align*}
where $w^+_r = \sum_{x\in B^+_r}a_x$. Thus we have the following lower bound for $\|H_{B^+_r}v\|_2^2 = \sum_{i=0}^{n-1}(H_{B^+_r}v)_i^2$:
\begin{align*}
    \sum_{i=0}^{n-1}(H_{B^+_r}v)_i^2 &\geq \sum_{i=0}^{n-1} (w^+_r)^2 \cdot e^{-2^{r+1}\lambda i/n} \cdot \left(\sum_{j=0}^{n-1}  e^{-2^r\lambda j/n}v_j\right)^2\\
    &\geq (w^+_r)^2 \cdot \left(\sum_{i=0}^{n-1} e^{-2^{r+1}\lambda i/n}\right)\cdot \left(\sum_{j=0}^{n-1}  e^{-2^r\lambda j/n}v_j\right)^2\\
    & = (w^+_r)^2 \cdot \left(\frac{1-e^{-2^{r+1}\lambda}}{1-e^{-2^{r+1}\lambda/n}}\right) \cdot \left(\sum_{j=0}^{n-1}  e^{-2^r\lambda j/n}v_j\right)^2 \\
    & =\Omega\left( (w^+_r)^2 \cdot \left(\frac{1}{1-e^{-2^{r+1}\lambda/n}}\right) \cdot \left(\sum_{j=0}^{n-1}  e^{-2^r\lambda j/n}v_j\right)^2\right) \numberthis \label{eq:spectral_norm_lower_bd_1}.
\end{align*}
Consider the lower bound above. It is clearly maximized for the unit norm vector that is correlated with the moment vector $v_n(e^{-2^r\lambda/n})$. Since the above holds for any positive vector $v$, it holds for the vector that maximizes the lower bound. Thus, we have the following lower bound for the righthand side of  \eqref{eq:spectral_norm_lower_bd_1} above:
\begin{align*}
    \max_{v\in \mathbb{R}_+^n }\|H_{B_r}v\|_2^2&\geq \Omega\left((w^+_r)^2 /(1-e^{-2^{r+1}\lambda/n})\right) \cdot \max_{v\in \mathbb{R}_+^n} \left(\sum_{j=0}^{n-1}  e^{-2^r\lambda j/n}v_j\right)^2 \\& = \Omega\left((w^+_r)^2 /(1-e^{-2^{r+1}\lambda/n})\right)\cdot \left(\sum_{j=0}^{n-1}  e^{-2^{r+1}\lambda j/n}\right) \quad (\text{max occurs for $v_j\propto e^{-2^r\lambda j/n}$})\\
     &= \Omega\left((w^+_r)^2 /(1-e^{-2^{r+1}\lambda/n})^2\right)\\
     &= \Omega\left((w^+_r)^2/(1-(1-2^{r+1}\lambda/n + O((2^{r+1}\lambda/n)^2)))^2\right)\\
     &= \Omega((w^+_r)^2 \cdot (n/2^r\lambda)^2),\numberthis \label{eq:spectral_norm_lower_bd_2}
\end{align*}
where the approximation for $e^{-2^{r+1}\lambda/n}$ is valid as long as $r\leq \log(0.1\cdot n/\lambda)$, and for $\log(0.1\cdot n/\lambda)\leq r\leq R$ we have that $(1-e^{-2^{r+1}\lambda/n})$ and $n/2^r\lambda$ are both $\Theta(1)$.
Now, since all entries of $H_{B^+_r}$ are positive, by Perron-Frobenius the vector $v$ that maximizes $\|H_{B^+_r}v\|_2^2$ also has positive entries. Thus $\|H_{B^+_r}\|_2^2 = \max_{v\in \mathbb{R}_+^n}\|H_{B^+_r}v\|_2^2$, and using the above lower bound of \eqref{eq:spectral_norm_lower_bd_2} we have that, for all $r\in [R+1]$,
\begin{equation*}
    \|H_{B^+_r}\|_2^2 = \Omega((w^+_r)^2 \cdot (n/2^r\lambda)^2).
\end{equation*}
For the last bucket $r=R+1$ we have that, for any $v\in \mathbb{R}^n_{+}$,
\begin{align*}
    \|H_{B^+_r}v\|_2&\geq (H_{B^+_{R+1}}v)_0\\
    & = \left(\sum_{j=0}^n \sum_{x\in B^+_{R+1}}a_{x}x^{0+j}v_j\right)\\
    &\geq \left(\sum_{x\in B^{+}_{R+1}}a_x\right)v_0\\
    &= w^+_{R+1} v_0,
\end{align*}
where $v_0$ value of $v$ at the first coordinate. Thus $\|H_{B^+_{R+1}}v\|_2\geq w^+_{R+1} v_0$ and,
\begin{equation*}
    \|H_{B^+_{R+1}}\|_2 = \max_{v\in \mathbb{R}_+^n}\|H_{B^+_{R+1}}v\|_2 \geq w^+_{R+1} = \Omega((w^+_{R+1})^2\cdot(n/2^{R+1}\lambda)),
\end{equation*}
since for $r=R+1$ we have that $(n/2^{R+1}\lambda)\leq 10$ as $R=\log(\ln(1/0.9)n/\lambda)$. Overall, we have that $\|H_{B^+_r}\|_2 = \Omega((w^+_r)^2 \cdot(n/2^r\lambda))$ for all $r\in [R+2]$.

\smallskip

\noindent\textbf{Lower bounding the contribution of $B^-_r$.} Now we consider bounding $\|H_{B^-_r}\|_2$ for any $r\in [R+2]$. Recall that $B^-_r = B_r\setminus B^+_r = \{x\in B_r: x< 0\}$ and $H_{B^-_r} = V_{B^-_r}D_{B^-_r}V_{B^-_r}^T$. Let $w^-_r = \sum_{x\in B^{-}_r}a_x$. For this, we will analyze the entries for $H_{B^-_r}v$ for any vector $v\in \mathbb{R}^n$, possibly with negative entries. This is because $H_{B^-_r}$ might have negative entries and thus Perron-Frobenius does not apply. Now for any $i\in [n]$ with $i$ being even we have
\begin{align*}
    (H_{B^-_r}v)_i &=  \sum_{j=0}^{n-1} \sum_{x\in B^-_r}a_x x^{i+j}v_j\\
    &=\sum_{j=0}^{n-1} \sum_{x\in B^-_r}a_x |x|^{i+j}(-1)^{i+j}v_j = \sum_{j=0}^{n-1} \sum_{x\in B^-_r}a_x |x|^{i+j} (-1)^jv_j \quad (\text{since $i$ even}).
\end{align*}
This shows that, since the signs of entries in even rows of $H_{B^-_r}$ are consistent as they do not depend on the row, the vector $v$ that maximizes $\sum_{i\in [n],i \text{ even}}(H_{B^-_r}v)_i^2 $ must have signs that correlate with the signs of entries in even rows of $H_{B_r}$. Thus, after multiplying each coordinate $v_j$ with $(-1)^j$, we can, without loss of generality, assume that $v$ has positive entries and we thus have for any $v\in \mathbb{R}^n$ with $(-1)^jv_j\geq 0$ for all $j\in [n]$:
\begin{align*}
    \|H_{B^-_{r}}v\|_2^2\geq \sum_{i\in [n],i \text{ even}}(H_{B^-_r}v)_i^2 &=\sum_{i\in [n],i \text{ even}}\left(\sum_{j=0}^{n-1} \sum_{x\in B^-_r}a_x |x|^{i+j} v_j\right)^2\\
    & \geq \sum_{i\in [n],i \text{ even}}\left(\sum_{j=0}^{n-1} \left(\sum_{x\in B^-_r}a_x\right) e^{-(2^{r}\lambda/n)(i+j)} v_j\right)^2\\
    & = (w^-_r)^2 \cdot \left(\sum_{i\in [n],i \text{ even}}e^{-2^{r+1}\lambda i/n}\right)\cdot \left(\sum_{j=0}^{n-1}e^{-2^{r}\lambda j/n} v_j\right)^2\\
    & = \Omega((w^-_r)^2/(1-e^{-2^{r+2}\lambda/n}))\left(\sum_{j=0}^{n-1}e^{-2^{r}\lambda j/n} v_j\right)^2.
\end{align*}
The above lower bound is clearly maximized for a unit $\ell_2$ norm $v$ with $v_j \propto e^{-2^{r}\lambda j/n}$ for all $j\in [n]$. Moreover, since the above holds for any $v$ with $(-1)^jv_j\geq 0$ for all $j\in [n]$, it holds for the vector maximizing the above lower bound. Thus,
\begin{align*}
    \max_{v\in \mathbb{R}^{n}:(-1)^jv_j\geq 0\forall j\in [n]}\|H_{B^-_r}v\|_2^2&\geq\Omega((w^-_r)^2/(1-e^{-2^{r+2}\lambda/n}))\cdot\max_{v\in \mathbb{R}^n_+}\left(\sum_{j=0}^{n-1}e^{-2^{r}\lambda j/n} v_j\right)^2 \\
     &= \Omega((w^-_r)^2/(1-e^{-2^{r+2}\lambda/n}))\cdot \left(\sum_{j=0}^{n-1}e^{-2^{r+1}\lambda j/n}\right)\\
     & = \Omega((w^-_r)^2/((1-e^{-2^{r+2}\lambda/n})\cdot (1-e^{-2^{r+1}\lambda/n})))\\
     & = \Omega((w^-_r)^2 \cdot (n/2^r\lambda)^2).
\end{align*}
Based on this lower bound, we can now bound $\|H_{B^-_r}\|_2$ by
\begin{equation*}
    \|H_{B^-_r}\|_2^2\geq \max_{v\in \mathbb{R}^{n}:(-1)^jv_j\geq 0\forall j\in [n]}\|H_{B^-_r}v\|_2^2 \geq  \Omega((w^-_r)^2 \cdot (n/2^r\lambda)^2).
\end{equation*}
Thus we have that $\|H_{B^-_r}\|_2 = \Omega(w^-_r \cdot (n/2^r\lambda))$ for all $r\in [R+1]$. For $r=R+1$ we have that, for any $v\in \mathbb{R}^n$ such that $(-1)^jv_j\geq 0$ for all $j\in [n]$,
\begin{equation*}
    \|H_{B^-_{R+1}}v\|_2\geq (H_{B^-_{R+1}}v)_0=\sum_{j=0}^{n-1} \sum_{x\in B^-_{R+1}}a_x |x|^{0+j} (-1)^jv_j\geq \left(\sum_{x\in B^-_{R+1}}a_x\right) v_0.
\end{equation*}
Since $\|v\|_2=1$, the above lower bound is maximized for $v_0=1$, thus $\|H_{B^-_{R+1}}\|_2\geq  w^-_{R+1}=\Omega(w^-_{R+1} \cdot n/2^{R+1}\lambda)$.

Consider the partition of $X = \left(\bigcup_{r\in [R+2]}B^+_r\right)\cup \left(\bigcup_{r\in [R+2]}B^-_r\right)$. By Lemma \ref{lem:partition_x} we know that $H = \sum_{r\in [R+2]} H_{B^+_r}+H_{B^{-}_r}$. Since $a_x>0$ for all $x\in X$ and each $H_{B^+_r} = V_{B^+_r} D_{B^+_r}V_{B^+_r}^T = (V_{B^+_r}\sqrt{D_{B^+_r}})(V_{B^+_r}\sqrt{D_{B^+_r}})^T$ is an outer product, $H_{B^+_r}\succeq 0$ and similarly $H_{B^-_r}\succeq 0$ for all $r\in [R+2]$. Thus $H\succeq H_{B^+_r}$ and $\|H\|_2\geq \|H_{B^+_r}\|_2 = \Omega(w^+_r\cdot n/(2^r\lambda))$ for all $r\in [R+2]$. Similarly, $\|H\|_2\geq \|H_{B^-_r}\|_2 = \Omega(w^-_r\cdot n/(2^r\lambda))$. Thus, we have
\begin{align*}
\|H\|_2 &\geq  \left(\max_{r\in [R+2]}\|H_{B^+_r}\|_2+\max_{r\in [R+2]}\|H_{B^-_r}\|_2\right)/2\\&=\Omega\left(\max_{r\in [R+2]}w^+_r\cdot n/(2^r\lambda)+\max_{r\in [R+2]}w^-_r\cdot n/(2^r\lambda)\right) \\&= \Omega\left(\max_{r\in [R+2]}(w^+_r+w^-_r)\cdot n/(2^r\lambda)\right) \\&=\Omega\left(\max_{r\in [R+2]}w_r\cdot n/(2^r\lambda)\right).
\end{align*}
Next, we move on lower bounding the Frobenius norm of $\|H\|_F$.

\smallskip

\noindent\textbf{Lower bounding $\|H\|_F$.} Note that $(H)_{0,0}\geq \sum_{x\in X}a_x \geq w_{R+1}$, and also for $i,j\in [n]$ with $i+j$ even we have that $(H)_{i,j} = \sum_{r=0}^{R+1}\sum_{x\in B_{r}}a_x |x|^{i+j}\geq \sum_{r=0}^{R}w_r e^{-2^{r}(\lambda/n) (i+j)}$.
This implies the following lower bound for $\sum_{i,j\in [n]}(H)_{i,j}^2$:
\begin{align*}
    \sum_{i,j\in [n]}(H)_{i,j}^2 &\geq  \sum_{i,j\in [n]:i+j \text{ even}}\left(\sum_{r=0}^{R}w_r e^{-2^{r}(\lambda/n) (i+j)}\right)^2\\
    & = \sum_{i,j\in [n/2]}\left(\sum_{r=0}^{R}w_r e^{-2^{r+1}(\lambda/n) (i+j)}\right)^2\\
    &= \sum_{i,j\in [n/2]}\left(\sum_{r=0}^{R}w^2_r e^{-2^{r+2}(\lambda/n) (i+j)} +2\sum_{\substack{r_1, r_2=0\\r_1\neq r_2}}^{R}w_{r_1}w_{r_2} e^{-(2^{r_1+1}+2^{r_2+1})(\lambda/n) (i+j)}\right)\\
    &=\sum_{r=0}^{R} w^2_r \sum_{i,j\in [n/2]} e^{-2^{r+2}(\lambda/n) (i+j)}+2\sum_{\substack{r_1, r_2=0\\r_1\neq r_2}}^{R}w_{r_1}w_{r_2}\sum_{i,j\in [n/2]} e^{-(2^{r_1+1}+2^{r_2+1})(\lambda/n) (i+j)}\\
    &\geq \sum_{r=0}^{R} w^2_r \left(\sum_{i\in [n]} e^{-2^{r+2}(\lambda/n)i}\right)^2 + 2\sum_{\substack{r_1, r_2=0\\r_1\neq r_2}}^{R}w_{r_1}w_{r_2}\left(\sum_{i\in [n]} e^{-(2^{r_1+1}+2^{r_2+1})(\lambda/n)i}\right)^2.
\end{align*}
We now continue to bound the above as follows:
\begin{align*}
    \sum_{i,j\in [n]}(H)_{i,j}^2
    & \geq \sum_{r=0}^{R} w^2_r \left(\sum_{i\in [n]} e^{-2^{r+2}(\lambda/n)i}\right)^2 + 2\sum_{\substack{r_1, r_2=0\\r_1\neq r_2}}^{R}w_{r_1}w_{r_2}\left(\sum_{i\in [n]} e^{-(2^{r_1+1}+2^{r_2+1})(\lambda/n)i}\right)^2
     \\
    &= \sum_{r=0}^{R} w^2_r \cdot \frac{(1-e^{-2^{r+2}\lambda})^2}{(1-e^{2^{r+1}\lambda/ n})^2}+2\sum_{\substack{r_1, r_2=0\\r_1\neq r_2}}^{R}w_{r_1}w_{r_2} \cdot \frac{(1-e^{-(2^{r_1+1}+2^{r_2+1})\lambda})^2}{(1-e^{-(2^{r_1+1}+2^{r_2+1})\lambda/ n})^2}\\
    &= \Omega\left(\sum_{r=0}^{R} w^2_r \cdot (n/2^r)^2 + \sum_{\substack{r_1, r_2=0\\r_1\neq r_2}}^{R}w_{r_1}w_{r_2} \cdot (n/(2^{r_1}+2^{r_2}))^2\right)\\
    & = \Omega\left(\sum_{r_1, r_2=0}^{R}w_{r_1}w_{r_2} \cdot (n/2^{\max\{r_1,r_2\}})^2\right),
\end{align*}
where the second last inequality follows since $e^{2^{r+1}\lambda}\leq \epsilon$ and $1-e^{2^{r+1}\lambda/ n}=O(n/(2^r\lambda))$ for all $r\in [R+2]$. Thus,
\begin{align*}
    \|H\|_F  = \Omega\left((H)_{0,0}+\sqrt{\sum_{i,j\in [n]}(H)_{i,j}^2}\right) = \Omega\left(\sqrt{\sum_{r_1, r_2=0}^{R}w_{r_1}w_{r_2} \cdot (n/2^{\max\{r_1,r_2\}})^2}+w_{R+1}\right).
\end{align*}
\end{proof}
Next, we convert our the entrywise guarantees in approximating each bucket from Lemma \ref{lem:bucket_sparsification} to guarantees on matrix norms matching the lower bounds on the norms of $H$ shown in this section.
\subsection{Analysis of Error Correlation across Buckets}\label{sec:error_correlation_analysis}
We now present the formal statement and proof of the approximation guarantee of Lemma \ref{lem:x_0_1_errs}.
\begin{replemma}{lem:x_0_1_errs}[Spectral and Frobenius norm error bounds]
Let $\wh{H}_{B_{R+1}}$ as $(\wh{H}_{B_{R+1}})_{i,j} = (H_{B_{R+1}})_{i,j}$ for all $i,j\in [n]$ such that $i+j\leq  t_{R+1}=O(\log(1/\epsilon))$ and $0$ otherwise. Let $D_{T_r} = diag(\cup_{i\in [l+1]}\{a^+_{i,r},a^{-}_{i,r}\})$ and $H_{T_r} = V_{T_r}D_{T_r}V_{T_r}^T$ where $T_r$ for all $r\in [R+1]$ and $a^+_{i,r},a^{-}_{i,r}$ for all $i\in [l+1],r\in [R+1]$ are as per Lemma \ref{lem:bucket_sparsification} in  \Cref{sec:bucket_sparsification}. 

Then for $\wh{H} = \wh{H}_{B_{R+1}}+\sum_{r\in [R+1]}H_{T_r}$, we have that $\wh{H}$ is Hankel  with rank at most $O(\log n \log(1/\epsilon))$ and satisfies 
\begin{align*}
    \|H-\wh{H}\|_2&\leq \epsilon\cdot \max_{r\in [R+2]}w_r \cdot n/(2^r),\\
    \|H-\wh{H}\|_F &\leq \epsilon\left(\sqrt{\sum_{r_1,r_2=0}^R w_{r_1}w_{r_2} (n/ 2^{\max\{r_1,r_2\}})^2} + w_{R+1}\right).
\end{align*}
\end{replemma}

We will build up to the proof of the above lemma in a few steps. We denote the per bucket errors as $E_r = H_{B_r} - H_{T_r}$ for all $r\in [R+1]$ and $E_{R+1} = H_{B_{R+1}}-\wh{H}_{B_{R+1}}$. We start with analyzing the error corresponding to the last bucket $E_{R+1}$.

\begin{lemma}[Error in approximating the last bucket]\label{lem:err_last_bucket}
The error corresponding to the last bucket $E_{R+1} = H_{B_{R+1}}-\wh{H}_{B_{R+1}}$ satisfies
\begin{equation*}
    \|E_{R+1}\|_2\leq \|E_{R+1}\|_F= O(\epsilon\log(1/\epsilon)w_{R+1}).
\end{equation*}
\end{lemma}
\begin{proof}
By definition of $E_{R+1}$ we have the following for any entry $i,j\in [n]$,
 \begin{equation*}
     (E_{R+1})_{i,j} = \begin{cases}
         (H_{B_{R+1}})_{i,j}=(\sum_{x\in B_{R+1}}a_x x^{i+j})\quad  \text{if } i+j\geq t_{R+1},\\
         0 \quad \text{otherwise}.
     \end{cases}
 \end{equation*}
Thus, observe that for any $i,j\in [n]$ with $i+j\geq t_{R+1}$ using the definition of $B_{R+1}$ (\Cref{def:bucketing}) to bound each $x\in B_{R+1}$ as $|x|\leq 0.9$, we have that $|(E_{R+1})_{i,j}|\leq w_{R+1} \cdot (0.9)^{i+j}$ and $(E_{R+1})_{i,j}=0$ otherwise. Using this entrywise bound we bound the Frobenius norm of $E_{R+1}$ as
\begin{align*}
    \|E_{R+1}\|_F^2 &= \sum_{i,j\in [n]} (E_{R+1})_{i,j}^2\\
    &\leq w_{R+1}^2\left(\sum_{i=0}^{t_{R+1}} \sum_{j=t_{R+1-i}}^n 0.9^{2i+2j} + \sum_{i=t_{R+1}}^{n} \sum_{j=0}^n 0.9^{2i+2j} \right)\\
    & \leq w_{R+1}^2\left(\sum_{i=0}^{t_{R+1}} 0.9^{2i}\sum_{j=t_{R+1-i}}^{\infty} 0.9^{2j} + \sum_{i=t_{R+1}}^{n} 0.9^{2i}\sum_{j=0}^{\infty} 0.9^{2j} \right)\\
    &\leq  10 w_{R+1}^2\left(\sum_{i=0}^{t_{R+1}} 0.9^{2i} \cdot 0.9^{2t_{R+1}-2i} + \sum_{i=t_{R+1}}^{n} 0.9^{2i} \right)\\
    &\leq 10 w_{R+1}^2\left(t_{R+1}\cdot 0.9^{2t_{R+1}} + \sum_{i=t_{R+1}}^{\infty} 0.9^{2i} \right)\\
    &\leq 10 w_{R+1}^2\left(t_{R+1}\cdot 0.9^{2t_{R+1}} + 10 0.9^{2t_{R+1}} \right) \leq 100 \epsilon^2\log(1/\epsilon)^2 w_{R+1}.
\end{align*}
Thus $\|E_{R+1}\|_F = O(\epsilon\log(1/\epsilon)w_{R+1})$. Finally, since $\|E_{R+1}\|_2\leq \|E_{R+1}\|_F$, this implies the proof.
\end{proof}
Next, our goal will be to analyze norms of $E=\sum_{r=0}^{R}E_r$. We will split each error $E_r = H_{B_r}- H_{T_r}$ into three parts $E_{r,1}, E_{r,2},E_{r,3}$. Let $E_{r,1}$ be defined as follows for all $i,j\in [n]$:
\begin{equation}\label{eq:e1_defn}
    (E_{r,1})_{i,j} = \begin{cases}
        (H_{B_r}- H_{T_r})_{i,j}\quad \forall i\leq t_r/2\text{ and }j\leq t_r/2,\\
        0\quad \text{otherwise}.        
    \end{cases}
\end{equation}
Let $ E_{r,2}$ and $E_{r,3}$ be defined as follows for all $i,j\in [n]$:
\begin{align}
    (E_{r,2})_{i,j} &=\begin{cases}
        (H_{B_r})_{i,j} \quad \forall i\geq t_r/2\text{ or }j\geq t_r/2 \\
        0\quad \text{otherwise}.
    \end{cases}\label{eq:e2_defn},\\
    (E_{r,3})_{i,j} &=\begin{cases}
        (-H_{T_r})_{i,j}\quad \forall i\geq t_r/2\text{ or }j\geq t_r/2 \\
        0\quad \text{otherwise}.
    \end{cases}\label{eq:e3_defn}
\end{align}
Thus error $E_{r,1}$ is due to \emph{sparsification} of $B_r$, and $E_{r,2},E_{r,3}$ are due to \emph{small entries beyond the cutoff}.

Thus $E_r = E_{1,r}+E_{r,2}+E_{r,3}$ for all $r\in [R+1]$, and let $E^1 = \sum_{r=0}^R E_{r,1}$, $E^2 = \sum_{r=0}^R E_{r,2}$, $E^3 = \sum_{r=0}^R E_{r,3}$. The total error $E$ is $E = E^1+E^2+E^3$, and we will bound the spectral norm of $E$ as $\|E\|_2\leq \|E^1\|_2+\|E^2\|_2+\|E^3\|_2$, and similarly for the Frobenius norm. We will first bound the norms of $E^1$, followed by $E^2$ and $E^3$.
\begin{lemma}[Bounding $E^1$ using correlation decay.]\label{lem:sparsification_errs}
    The bucket sparsification error $E^1=\sum_{r=0}^R E_{r,1}$ with each $E_{r,1}$ defined in \eqref{eq:e1_defn} satisfies,
    \begin{align*}
    \|E^1\|_2 &\leq  10\epsilon\max_{r\in [R+1]}w_r \cdot n/2^r\\
        \|E^1\|_F&\leq 10\epsilon\cdot \sqrt{\sum_{r_1,r_2=0}^{R} w_{r_1}w_{r_2}\cdot (n/2^{\max\{r_1,r_2\}})^2 }.
    \end{align*}
\end{lemma}
\begin{proof}
We first analyze $\|E^{1}\|_2$. Consider any vector $v$ with $\|v\|_2=1$, and consider any index $i\in [n]$. Then we have the following bound on $(E_{r,1}v)_i$ for any $r\in [R+1]$ using its definition as per \eqref{eq:e1_defn}:
\begin{align*}
    (E_{r,1}v)_i &= \begin{cases}
        \sum_{j=0}^{t_r/2} (H_{B_r}- H_{T_r})_{i,j}v_j \text{ if } i\leq t_r/2,\\
        0 \text{ otherwise. }
    \end{cases}
    \\
    &\leq \begin{cases}
        \epsilon w_r\sum_{j=0}^{t_r/2} |v_j| \text{ if } i\leq t_r/2,\\
        0 \text{ if } i>t_r/2,
    \end{cases}, \numberthis \label{eq:e1_entrywise_bound}
\end{align*}
where in the last inequality we used the fact that $(E_{r,1})_{i,j}\leq \epsilon w_r$ for all $i\leq t_r/2$ and $j\leq t_r/2$ by Lemma \ref{lem:bucket_sparsification}. Let $v_r\in \mathbb{R}^{n}$ such that $(v_r)_{j}=1$ for all $j\leq t_r/2$ and $0$ elsewhere. Furthermore, let $V\in \mathbb{R}^{n\times R}$ have its $r^{th}$ column be equal to $v_r$ for all $r\in [R+1]$. Let $D=diag(\{w_r\}_{r=0}^R)$. Suppose $|v|$ is the vector whose $j^{th}$ entry is $|v_j|$ for all $j\in [n]$. Then we have the following bound on $\|E^{1}v\|_2$:
\begin{align*}
    \|E^{1}v\|_2^2 = \sum_{i=0}^n(E^{1}v)_i^2 &= \sum_{i=0}^n\left(\sum_{r=0}^R (E_{r,1} v)_i\right)^2\\
    &\leq \epsilon^2 \sum_{i=0}^n\left(\sum_{r=0}^R w_r \sum_{j=0}^{t_r}|v_j|\right)^2\\
    & = \epsilon^2 \sum_{i=0}^n\left(\sum_{r=0}^R w_r (v_r v^T_r |v|)_i\right)^2\\
    & = \epsilon^2 \|VDV^T |v|\|_2^2\leq \epsilon^2 \|VDV^T\|_2^2. \quad \text{(since $\||v|\|_2=1$)}\numberthis \label{eq:e1_bound}
\end{align*}
Furthermore, the bound of \eqref{eq:e1_entrywise_bound} when applied for $v$ chosen to be basis vectors also implies an entrywise bound. That is, for any $i,j\in [n]$ we have that $(E^1)_{i,j} = \sum_{r=0}^R(E_{r,1})_{i,j}\leq \sum_{r=0}^R \epsilon w_r(v_r^Tv_r)_{i,j} = \epsilon (VDV^T)_{i,j}$.
Thus, we have
\begin{equation}\label{eq:e1_fro_bound}
    \|E^1\|_F \leq \epsilon \|VDV^T\|_F. 
\end{equation}
Next, we analyze the spectral norm of $VDV^T$. We first rescale each column of $V$ by its $\ell_2$ norm and scale up the corresponding entry in $D$ by the $\ell_2$ norm squared. For any $r\in [R+1]$, $\|v_r\|_2 = \sqrt{t_r/2} = \sqrt{n/2^r}$ since it contains $t_r$ many ones and rest of the entries are $0$. Thus we rescale $v_r$ by $\sqrt{t_r/2}$ and scale up $D_{r,r} = w_r$ by $t_r/2$. Hence $D_{r,r} = w_r\cdot n/2^r$ for all $r\in [R+1]$.

Observe that since $D\succeq 0$, $VDV^T$ can be written as the outer product $(\sqrt{D}V)^T (\sqrt{D}V)$. Since the eigenvalues of $A^TA$ are equal to those of $AA^T$ for any matrix $A$, we thus get that the eigenvalues of $VDV^T$ are equal to the eigenvalues of $\sqrt{D}V^T V\sqrt{D}$. First we bound $\|V^TV\|_2$ using Lemma \ref{lem:spectral_norm_bound}. Since $V^TV$ is symmetric, its max $\ell_1$ norm of all rows and columns is identical. Thus, we have by Lemma \ref{lem:spectral_norm_bound},
\begin{align*}
    \|V^TV\|_2&\leq \sqrt{\|V^TV\|_{1}\cdot\|V^TV\|_{\infty}}\\
    &=\|V^TV\|_{\infty}\\
    &= \max_{r_1\in [R+1]} \sum_{r_2\in [R+1]} |(V^TV)_{r_1,r_2}|. \numberthis \label{eq:spectral_norm_inequality}
\end{align*}
Now for any $r_1,r_2\in [R+1]$ such that $r_1\neq r_2$ we have for $(V^T V)_{r_1,r_2}$,
\begin{align*}
    (V^T V)_{r_1,r_2} & = \frac{v^T_{r_1}v_{r_2}}{\|v_{r_1}\|_2\|v_{r_2}\|_2}\\
    & = \frac{\sum_{j=0}^{\sqrt{\min\{t_{r_1}/2,t_{r_2}/2\}}} 1}{\sqrt{t_{r_1}/2}\sqrt{t_{r_2}/2}}\\
    &=\min \left\{ \sqrt{\frac{t_{r_1}/2}{t_{r_2}/2}},\sqrt{\frac{t_{r_2}/2}{t_{r_1}/2}}\right\}\\
    &= \min \left\{ \sqrt{\frac{2^{r_2}}{2^{r_1}}},\sqrt{\frac{2^{r_1}}{2^{r_2}}}\right\}\\
    &= \frac{1}{\sqrt{2}^{|r_2-r_1|}}. \numberthis \label{eq:correlation_bd_1}
\end{align*}
For a fixed $r_1\in [R+1]$, the above implies for $\sum_{r_2=0}^R(V^TV)_{r_1,r_2}$,
\begin{align*}
    \sum_{r_2=0}^R(V^TV)_{r_1,r_2} &= 1+ \sum_{r_2\neq r_1}\frac{1}{\sqrt{2}^{|r_2-r_1|}}\\
    &\leq  1+2\sum_{r=0}^{\infty}\frac{1}{\sqrt{2}^r}\\
    & \leq 10.
\end{align*}
Thus, we have that $\|V^T V\|_1 \leq 10$ and furthermore, by \eqref{eq:spectral_norm_inequality}, we have that $\|V^TV\|_2 \leq 10 $. Thus $\|\sqrt{D}V^T V \sqrt{D}\|_2 = O(\|D\|_2) = 10\max_{r\in [R+1]}w_r \cdot n/2^r$. This combined with \eqref{eq:e1_bound} implies that $\|E^1v\|_2 \leq 10\epsilon\max_{r\in [R+1]}w_r \cdot n/2^r$ for any $v$ with $\|v\|_2=1$. Thus, we have
\begin{equation}
\|E^1\|_2 \leq  10\epsilon\max_{r\in [R+1]}w_r \cdot n/2^r.
\end{equation}
Next, we analyze the Frobenius norm of $VDV^T$. Since the eigenvalues of $VDV^T$ and $\sqrt{D}V^TV \sqrt{D}$ are equal, we have that $\|VDV^T\|_F =\|\sqrt{D}V^TV \sqrt{D}\|_F$. This implies
\begin{align*}
    \|VDV^T\|_F^2 &=\|\sqrt{D}V^TV \sqrt{D}\|_F^2\\
    & = \sum_{r_1,r_2=0}^{R} D_{r_1,r_1}D_{r_2,r_2}(V^TV)_{r_1,r_2}^2\\
    & = 100\sum_{r_1,r_2=0}^{R} w_{r_1}w_{r_2}(n^2/2^{r_1+r_2})\cdot (1/2^{|r_1-r_2|}) \quad \text{(From  \eqref{eq:correlation_bd_1}.)}\\
    & = 100\sum_{r_1,r_2=0}^{R} w_{r_1}w_{r_2}\cdot (n/2^{\max\{r_1,r_2\}})^2.
\end{align*}
Thus combining this with \eqref{eq:e1_fro_bound}, we have
\begin{equation}
    \|E^1\|_F\leq \epsilon  \|VDV^T\|_F\leq 10\epsilon\cdot \sqrt{\sum_{r_1,r_2=0}^{R} w_{r_1}w_{r_2}\cdot (n/2^{\max\{r_1,r_2\}})^2 }.
\end{equation}
\end{proof}
Next, we analyze the errors $E^2$ and $E^3$ using carefully constructed entrywise upper bounds, and using decay of correlations across column vectors to bound the spectral and Frobenius norms of the upper bounds.
\begin{lemma}[Bounding $E^2,E^3$ using correlation decay]\label{lem:small_entries_errs}
The errors corresponding to small entries beyond cutoffs $E^2 = \sum_{r=0}^R E_{r,2}$ and $E^3 = \sum_{r=0}^R E_{r,3}$ for each $ E_{r,2}, E_{r,3}$ defined as per equations \eqref{eq:e2_defn} and \eqref{eq:e3_defn} satisfy
\begin{align}
    \|E^2\|_2+\|E^3\|_2 &\leq  1000\sqrt{\epsilon}\log(1/\epsilon)\cdot \max_{r\in [R+1]}w_r  (n/2^r\lambda),\\
    \|E^2\|_F+\|E^3\|_F &\leq 1000\sqrt{\epsilon}\log(1/\epsilon) \cdot \sqrt{\sum_{r_1,r_2=0}^R w_{r_1}w_{r_2} (n/ 2^{\max\{r_1,r_2\}})^2}.
\end{align}
\end{lemma}
\begin{proof}
We first analyze $E^2$. Consider a vector $v$ with $\|v\|_2=1$, and consider $(E^2v)_i = \sum_{r=0}^R (E_{r,2}v)_i$ for any row $i\in [n]$. Then using the fact that $|x|\leq e^{-2^{r-1}\lambda/n}$ for all $x\in B_r$, we have the following by definition of $E_{r,2}$ as per \eqref{eq:e2_defn} for any $r\in [R+1]$ with $r\geq 1$,
\begin{align*}
    (E_{r,2} v)_i &= \begin{cases}
        \sum_{j=t_r/2}^{n-1} \left(\sum_{x\in B_r}a_x x^{i+j}\right)v_j \text{ if }i\leq t_r/2\\
        \sum_{j=0}^{n-1} \left(\sum_{x\in B_r}a_x x^{i+j}\right)v_j \text{ if }i> t_r/2
    \end{cases}
    \\
    & \leq \begin{cases}
        \sum_{j=t_r/2}^{n-1} \left(\left(\sum_{x\in B_r}a_x\right) e^{-2^{r-1}(\lambda/n)(i+j)}\right)|v_j| \text{ if }i\leq t_r/2\\
        \sum_{j=0}^{n-1} \left(\left(\sum_{x\in B_r}a_x\right) e^{-2^{r-1}(\lambda/n)(i+j)}\right)|v_j|  \text{ if }i> t_r/2
    \end{cases}
    \\
    & = \begin{cases}
        w_r e^{-2^{r-1}(\lambda/n)i}\sum_{j=t_r/2}^{n-1}  e^{-2^{r-1}(\lambda/n)j}|v_j| \text{ if }i\leq t_r/2\\
        w_r e^{-2^{r-1}(\lambda/n)i}\sum_{j=0}^{n-1}  e^{-2^{r-1}(\lambda/n)j}|v_j|\text{ if }i> t_r/2
    \end{cases}
\end{align*}
Now if $i\leq t_r/2$, we have the following bound on $\sum_{j=t_r/2}^{n-1}  e^{-2^{r-1}(\lambda/n)j}|v_j|$:
\begin{align*}
    \sum_{j=t_r/2}^{n-1}  e^{-2^{r-1}(\lambda/n)j}|v_j| &= e^{-2^{r-1}(\lambda/2n)t_r/2}\sum_{j=t_r}^{n-1}e^{-2^{r-1}(\lambda/n)j + 2^{r-1}(\lambda/2n)t_r}|v_j|\\
    &\leq e^{-2^{r-1}(\lambda/2n)t_r/2}\sum_{j=t_r}^{n-1}e^{-2^{r-1}(\lambda/2n)j}|v_j|,
\end{align*}
where in the last line we used the fact that for $j\geq t_r/2$ we have that $-2j+t_r/2\leq -j$. We also have that $e^{-2^{r-1}(\lambda/n)i}\leq e^{-2^{r-1}(\lambda/2n)i}$. Thus, for $i\leq t_r/2$ we have
\begin{align*}
(E_{r,2}v)_i&\leq  e^{-2^{r-1}(\lambda/2n)t_r/2}\cdot e^{-2^{r-1}(\lambda/2n)i}\sum_{j=t_r/2}^{n-1}e^{-2^{r-1}(\lambda/2n)j}|v_j|  \\
&\leq \sqrt{\epsilon}\cdot e^{-2^{r-1}(\lambda/2n)i}\sum_{j=t_r/2}^{n-1}e^{-2^{r-1}(\lambda/2n)j}|v_j|\\
&\leq \sqrt{\epsilon}\cdot \sum_{j=0}^{n-1}e^{-2^{r-1}(\lambda/2n)(i+j)}|v_j|.
\end{align*}
Now when $i>t_r/2$, we have that $-2i\leq -t_r/2-i$. This implies 
\begin{equation*}
 e^{-2^{r-1}(\lambda/n)i}\leq e^{-2^{r-1}(\lambda/2n)t_r/2}\cdot e^{-2^{r-1}(\lambda/2n)i} = \sqrt{\epsilon}\cdot e^{-2^{r-1}(\lambda/2n)i}.
\end{equation*}
Moreover, $\sum_{j=0}^{n-1}e^{-2^{r-1}(\lambda/n)j}|v_j|\leq \sum_{j=0}^{n-1}e^{-2^{r-1}(\lambda/2n)j}|v_j|$ also holds easily. Overall we get that when $i\geq t_r/2$, $\sum_{j=0}^{n-1}  e^{-2^{r-1}(\lambda/n)j}|v_j|$ also satisfies
\begin{equation*}
    e^{-2^{r-1}(\lambda/n)i}\sum_{j=0}^{n-1}  e^{-2^{r-1}(\lambda/n)j}|v_j|\leq \sqrt{\epsilon} \sum_{j=0}^{n-1}e^{-2^{r-1}(\lambda/2n)(i+j)}|v_j|.
\end{equation*}
This implies that $(E_{r,2} v)_i \leq \sqrt{\epsilon}\cdot w_r \sum_{j=0}^{n-1}e^{-2^{r-1}(\lambda/2n)(i+j)}|v_j|$ for all $i\in [n]$ and $r\in [R+1]$ with $r\geq 1$. In particular, this also implies an entrywise bound. For any $i,j\in [n]$ and for all $r\in [R+1]$  with $r\geq 1$, we have
\begin{equation}\label{eq:e2_entrywise_bd}
(E_{r,2})_{i,j}\leq \sqrt{\epsilon}w_r\cdot e^{-2^{r-1}(\lambda/2n)(i+j)}.
\end{equation}

Moreover, for $r=0 $ $E_{r,2}=0$ by definition. Thus, we only need to consider $r\geq 1$ when analyzing $\|E^2\|_2$ and $\|E^2\|_F$.

Let $\wt{E}_{r,2}  = w_r v_n(e^{-2^{r-1}\lambda/2n})v_n(e^{-2^{r-1}\lambda/2n})^T$ be a rank $1$ matrix, and define it for all $r\in [R+1]$  with $r\geq 1$. Let $\wt{E}^2 = \sum_{r=1}^R \wt{E}_{r,2} $. Thus for $\Lambda = \{e^{-2^{r-1}\lambda/2n}\}_{r=1}^{R}$ and corresponding Vandermonde matrix $V_{\Lambda}$, and $D_{\Lambda} = diag(\{w_r\}_{r=1}^{R})$, $\wt{E}^2 = V_{\Lambda}D_{\Lambda}V_{\Lambda}^T$. Also let $|v|$ be the vector whose $j^{th}$ entry is $|v_j|$. With this notation, we get the following bound on $\|E^2 v\|_2^2$:
\begin{align*}
    \|E^2v\|_2^2 &= \sum_{i=0}^{n-1}\left(\sum_{r=1}^{R}(E_{r,2} v)_i\right)^2\\&\leq \epsilon \sum_{i=0}^{n-1}\left(\sum_{r=1}^R w_r\sum_{j=0}^{n-1}e^{-2^{r-1}(\lambda/2n)(i+j)}|v_j|\right)^2\\
    & = \epsilon\sum_{i=0}^{n-1}\left(\sum_{r=1}^R (\wt{E}_{r,2}  |v|)_i)\right)^2 \\&= \epsilon \|\wt{E}^2 |v| \|_2^2\leq \epsilon \|\wt{E}^2\|_2^2 \quad (\text{since $\|v\|_2=1$ implies $\||v|\|_2=1$}).
\end{align*}
Since this is true for any $v$ with $\|v\|_2=1$ we have
\begin{equation}\label{eq:e2_bound}
    \|E^2v\|_2 \leq \sqrt{\epsilon}\|\wt{E}^2\|_2.
\end{equation}
Moreover, using the entrywise bound of \eqref{eq:e2_entrywise_bd} we also have
\begin{align*}
   \|E^2\|_F^2 &= \sum_{i,j\in [n]} \left(\sum_{r=1}^R(E_{r,2})_{i,j}\right)^2\\
   &\leq \epsilon \sum_{i,j\in [n]} \left(\sum_{r=1}^Rw_r e^{-2^{r-1}(\lambda/2n)(i+j)}\right)^2\\
   &= \epsilon \sum_{i,j\in [n]} (\wt{E}^2)_{i,j}^2 = \epsilon \|\wt{E}^2\|_F^2. \numberthis\label{eq:e2_fro_bound}
\end{align*}
We now bound $\|\wt{E}^2\|_2$. Observe that $\wt{E}^2$ can be written as the following outer product:
\begin{equation*}
\wt{E}^2 = V_{\Lambda}D_{\Lambda}V_{\Lambda}^T = (V_{\Lambda}\sqrt{D_{\Lambda}}) (V_{\Lambda}\sqrt{D_{\Lambda}})^T.
\end{equation*}
Thus, the largest eigenvalue of $\wt{E}^2$ is the same as the largest eigenvalue of $\sqrt{D_{\Lambda}}V_{\Lambda}^TV_{\Lambda}\sqrt{D_{\Lambda}}$. First, we rescale each column of $V_{\Lambda}$ by its $\ell_2$ norm, and scale up the corresponding diagonal entry by its $\ell_2$ norm squared. For any $r\in [R+1]$, the $\ell_2$ norm of the $r^{th}$ column of $V_{\Lambda}$, which is $v_n(e^{-2^{r-1}\lambda/2n})^T$, is as follows:
\begin{equation*}
\|v_n(e^{-2^{r-1}\lambda/2n})\|_2^2 = \sum_{i=0}^{n-1}e^{-2^{r-1}\lambda i/n} = \frac{1-e^{-2^{r-1}\lambda}}{(1-e^{-2^{r-1}\lambda/n})} = \Theta\left(\frac{n}{2^r\lambda}\right).
\end{equation*}
Here, the last equality is true because the numerator is at most $1$ and at least $1-e^{-\lambda}=1-\epsilon$, and thus it is $\Theta(1)$. For the denominator, we observe that $e^{-2^{r-1}\lambda/n} = 1-\Theta(2^{r-1}\lambda/n)$ as long as $2^{r-1}\lambda/n\leq c$ for some small enough constant $c>0$ using the inequality $|e^{-x}-(1-x)|\leq O(x^2)$ for any $x<c$. On the other hand when $2^{r-1}\lambda/n\geq c$, $e^{-2^{r-1}\lambda/n}\in [e^{-2^{R-1}\lambda/n},e^{-c}] = [0.9,e^{-c}]$, thus $1/(1-e^{-2^{r-1}\lambda/n}) = \Theta(n/(2^r\lambda))$ also holds since both the sides of the equality are $\Theta(1)$. 

Now observe that we can bound the spectral norm of $\sqrt{D_{\Lambda}}V_{\Lambda}^TV_{\Lambda}\sqrt{D_{\Lambda}}$ as
\begin{equation*}
\|\sqrt{D_{\Lambda}}V_{\Lambda}^TV_{\Lambda}\sqrt{D_{\Lambda}}\|_2 \leq \|D_{\Lambda}\|_2\cdot \|V_{\Lambda}^TV_{\Lambda}\|_2.
\end{equation*}
First we bound $\|V_{\Lambda}^TV_{\Lambda}\|_2$. As per Lemma \ref{lem:spectral_norm_bound} and since $V_{\Lambda}^TV_{\Lambda}$ is symmetric, we have
\begin{align*}
    \|V_{\Lambda}^TV_{\Lambda}\|_2&\leq \sqrt{\|V_{\Lambda}^TV_{\Lambda}\|_{1}\cdot\|V_{\Lambda}^TV_{\Lambda}\|_{\infty}}\\
    &=\|V_{\Lambda}^TV_{\Lambda}\|_{\infty}\\
    &= \max_{r_1\in [R+1]} \sum_{r_2\in [R+1]} |(V_{\Lambda}^TV_{\Lambda})_{r_1,r_2}|.
\end{align*}
Now for any $r_1,r_2\in [R+1]$ with $r\geq 1$ such that $r_1\neq r_2$, we have that $(V_{\Lambda}^TV_{\Lambda})_{r_1,r_2}$ after rescaling is as follows:
\begin{align*}
    (V_{\Lambda}^TV_{\Lambda})_{r_1,r_2} &= \frac{v_n(e^{-2^{r_1-1}\lambda/n})^T v_n(e^{-2^{r_2-1}\lambda/n})}{\|v_n(e^{-2^{r_1-1}\lambda/n})\|_2\|v_n(e^{-2^{r_2-1}\lambda/n})\|_2}\\
    & = \frac{\sum_{i=0}^{n-1}e^{-(2^{r_1-1}+2^{r_2-1})\lambda i/n}}{\|v_n(e^{-2^{r_1-1}\lambda/n)}\|_2\|v_n(e^{-2^{r_2-1}\lambda/n)}\|_2}\\
    & \leq  \frac{1}{1-e^{-(2^{r_1-1}+2^{r_2-1})\lambda/n}}\cdot \frac{1}{\|v_n(e^{-2^{r_1-1}\lambda/n)}\|_2\|v_n(e^{-2^{r_2-1}\lambda/n)}\|_2}\\
    & \leq 100 \frac{n}{\lambda(2^{r_1}+2^{r_2})}\cdot \frac{\sqrt{2^{r_1}2^{r_2}}\lambda}{n} \leq 100\frac{\sqrt{2^{r_1} 2^{r_2}}}{2^{r_1}+2^{r_2}}\numberthis \label{eq:inner_product_bound}.
\end{align*}
Thus, for any $r_1\in [R+1]$, we have the following bound on $\sum_{r_2=0}^R |(V_{\Lambda}^TV_{\Lambda})_{r_1,r_2}$|:
\begin{align*}
    \sum_{r_2=1}^R|(V_{\Lambda}^TV_{\Lambda})_{r_1,r_2}| & \leq 1+100\sum_{r_2\neq r_1}\frac{\sqrt{2^{r_1} 2^{r_2}}}{2^{r_1}+2^{r_2}}\\
    &\leq 1+ 100\sum_{r_2=-\infty}^{r_1-1}\frac{\sqrt{2^{r_1} 2^{r_2}}}{2^{r_1}+2^{r_2}} + 100\sum_{r_2=r_1+1}^{\infty}\frac{\sqrt{2^r_1 2^{r_2}}}{2^{r_1}+2^{r_2}}\\
    &\leq 1+ 100\sum_{r_2=-\infty}^{r_1-1}\sqrt{2^{r_2-r_1}} + 100\sum_{r_2=r_1+1}^{\infty}\frac{\sqrt{2^{r_2-r_1}}}{2^{r_2-r_1}}\\
    &\leq 500.
\end{align*}
Thus we have that $\|V_{\Lambda}^TV_{\Lambda}\|_2\leq 500$. Thus $\|\sqrt{D_{\Lambda}}V_{\Lambda}^TV_{\Lambda}\sqrt{D_{\Lambda}}\|_2 \leq 500\|D_{\Lambda}\|_2 = 500\max_{r\in [R+1]}w_r \cdot (n/2^r\lambda)$. Overall we get that $\|\wt{E}^2\|_2 =500\max_{r\in [R+1]}w_r \cdot (n/2^r\lambda))$, which combined with \eqref{eq:e2_bound} implies that $\|E^2\|_2 = 500\sqrt{\epsilon}\max_{r\in [R+1]}w_r \cdot (n/2^r\lambda))$.
Next we bound $\|E^2\|_F$. Recall from \eqref{eq:e2_fro_bound} that $\|E^2\|_F\leq \epsilon\|\wt{E}^2\|_F = \epsilon\|V_{\Lambda}D_{\Lambda}V_{\Lambda}^T\|_F$. As we argued before that the eigenvalues of $V_{\Lambda}D_{\Lambda}V_{\Lambda}^T$ are identical to those of $\sqrt{D_{\Lambda}}V_{\Lambda}^TV_{\Lambda}\sqrt{D_{\Lambda}}$, thus $\|V_{\Lambda}D_{\Lambda}V_{\Lambda}^T\|_F = \|\sqrt{D_{\Lambda}}V_{\Lambda}^TV_{\Lambda}\sqrt{D_{\Lambda}}\|_F$. Thus, we have 
\begin{align*}
    \|V_{\Lambda}D_{\Lambda}V_{\Lambda}^T\|_F^2 &= \|\sqrt{D_{\Lambda}}V_{\Lambda}^TV_{\Lambda}\sqrt{D_{\Lambda}}\|_F^2\\
    & = \sum_{r_1,r_2=1}^R (D_{\Lambda})_{r_1,r_1}(D_{\Lambda})_{r_2,r_2} (V_{\Lambda}^TV_{\Lambda})_{r_1,r_2}^2\\
    &\leq 100\sum_{r_1,r_2=0}^R w_{r_1}w_{r_2} (n^2/\lambda 2^{r_1+r_2})\cdot(2^{r_1+r_2}/(2^{r_1}+2^{r_2})^2)\quad \text{(using \eqref{eq:inner_product_bound})}\\
    &\leq 100\sum_{r_1,r_2=0}^R w_{r_1}w_{r_2} (n/ 2^{\max\{r_1,r_2\}})^2.
\end{align*}
Combining the above with \eqref{eq:e2_fro_bound} we obtain
\begin{equation}\label{eq:e2_fro_error_bd}
    \|E^2\|_F\leq 500 \sqrt{\epsilon} \sqrt{\sum_{r_1,r_2=0}^R w_{r_1}w_{r_2} (n/\lambda 2^{\max\{r_1,r_2\}})^2}.
\end{equation}
Since each $|x|\leq e^{-2^{r-1}\lambda/n}$ for all $x\in T_r$ as well, and $\sum_{i\in [l+1]}|a^+_{i,r}|+|a^{-}_{i,r}|\leq O(l\cdot w_r)$ for all $r\in [R+1]$ as per Lemma \ref{lem:bucket_sparsification}, we can obtain the following using the same argument as above and redefining $D_{\Lambda}$ as $D_{\Lambda}=diag(\cup_{i\in [1,R]}\{\sum_{i\in [l+1]}|a^+_{i,r}|+|a^{-}_{i,r}|\})$:
\begin{align*}
    \|E^3\|_2 &\leq 500\sqrt{\epsilon}l \cdot \max_{r\in [R+1]}\cdot w_r \cdot (n/2^r)),\\
\|E^3\|_F &\leq 500 \sqrt{\epsilon}l \cdot \sqrt{\sum_{r_1,r_2=0}^R w_{r_1}w_{r_2} (n/ 2^{\max\{r_1,r_2\}})^2}.
\end{align*}
Combining the bounds on $E^2$ and $E^3$ completes the proof.
\end{proof}

We now present the proof of Lemma \ref{lem:x_0_1_errs} by combining these intermediate lemmas.
\begin{proof}[Proof of Lemma \ref{lem:x_0_1_errs}]
Combining the spectral norm bounds from Lemmas \ref{lem:err_last_bucket},\ref{lem:sparsification_errs} and \ref{lem:small_entries_errs} we have
\begin{align*}
    \|H-\wh{H}\|_2 & = \left\|H_{B_{R+1}}-\wh{H}_{B_{R+1}}+\sum_{r=0}^{R}H_{B_r}-H_{T_r}\right\|_2\\
    &= \left\|\sum_{r=0}^{R+1} E_{r}\right\|_2\\
    &\leq \left\|\sum_{r=0}^R E_{r}\right\|_2 + \|E_{R+1}\|_2\\
    &= \left\|E^1+E^2+E^3\right\|_2 + O(\epsilon\log(1/\epsilon)w_{R+1})\\
    &\leq \|E^1\|_2+\|E^2\|_2+\|E^3\|_2+O(\epsilon\log(1/\epsilon)w_{R+1})\\
    &\leq O(\sqrt{\epsilon}\log^2(1/\epsilon)\max_{r\in [R+2]} w_r \cdot n/(2^r)).
\end{align*}
Similarly, combining the Frobenius norm bounds, we have,
\begin{align*}
      \|H-\wh{H}\|_F 
    &\leq \|E^1\|_F+\|E^2\|_F+\|E^3\|_F +\|E_{R+1}\|_F\\
    &\leq O\left(\sqrt{\epsilon}\log^2(1/\epsilon) \left(\sqrt{\sum_{r_1,r_2=0}^R w_{r_1}w_{r_2} (n/ 2^{\max\{r_1,r_2\}})^2} + w_{R+1}\right)\right).
\end{align*}
Replacing $\epsilon$ with $\epsilon^4$ we get that the above becomes
\begin{align*}
    \|H-\wh{H}\|_2& \leq \epsilon\max_{r\in [R+2]} w_r \cdot n/(2^r)),\\
    \|H-\wh{H}\|_F &\leq \epsilon\left(\sqrt{\sum_{r_1,r_2=0}^R w_{r_1}w_{r_2} (n/ 2^{\max\{r_1,r_2\}})^2} + w_{R+1}\right).
\end{align*}
Finally, we conclude the proof of the lemma by bounding the rank of $\wh{H}$. Clearly $\wh{H}_{B_{R+1}}$ is Hankel and $rk(\wh{H}_{B_{R+1}})\leq O(\log(1/\epsilon))$ as it only has $O(\log(1/\epsilon))$ nonzero rows. Moreover, $H_{T_r}$ is Hankel with rank $|T_r|= l = O(\log(1/\epsilon))$ for all $r\in [R+1]$. Thus since $\wh{H} = \wh{H}_{B_{R+1}}+\sum_{r=0}^{R}H_{T_r}$, $\wh{H}$ is Hankel as it is the sum of Hankel matrices and $rk(\wh{H})\leq rk(\wh{H}_{B_{R+1}})+ \sum_{r=0}^{R}rk(H_{T_r}) = O(\log n \log(1/\epsilon))$.
\end{proof}

Equipped with the proofs of Lemmas \ref{lem:bucket_sparsification}, \ref{lem:x_0_1_errs}, and \ref{lem:hankel_spectral_lower_bd}, we now present the proof of Lemma \ref{lem:x_0_1}, which proves the main existence result of Theorem \ref{thm:main_thm} when all $x\in X$ satisfy $|x|\leq 1$.
\begin{proof}[Proof of Lemma \ref{lem:x_0_1}]
Applying Lemma \ref{lem:x_0_1_errs} on $H$ to obtain symmetric Hankel $\wh{H}$ and applying Lemma \ref{lem:hankel_spectral_lower_bd} to lower bound $\|H\|_2$ we have that $\|H-\wh{H}\|_2\leq \epsilon \cdot \max_{r\in [0,R+1]}w_r \cdot n/2^r = O(\epsilon \lambda \|H\|_2) = O(\epsilon\log(1/\epsilon) \|H\|_2)$. Setting $\epsilon$ as $\epsilon^2$ we get that $\|H-\wh{H}\|_2\leq\epsilon \|H\|_2$. Similarly, using the bounds for the Frobenius norm, we get that $\|H-\wh{H}\|_F\leq\epsilon \|H\|_F$.
\end{proof}
Now we consider the case when all $x\in X$ have $|x| > 1$ and reduce it to the case when all $x\in X$ have $|x| \le 1$. We then present the proof of Theorem \ref{thm:main_thm}, which combines our results for these two cases.
\subsection{Analysis of the Case when \texorpdfstring{$|x|>1$}{|x|>1} for all \texorpdfstring{$x\in X$}{x in X}}\label{sec:x_greater_than_1_reduction}

To present our reduction, we need the following definition of the row order reversal operator.
\begin{definition}[Row order reversal operator]\label{def:row_reversal}
Let $R\in \mathbb{R}^{n\times n}$ be the row order reversal operator defined as follows,
\begin{equation*}
R_{i,j} = \begin{cases}
    1 \text{ if } i+j=n\\
    0 \text{ otherwise}
\end{cases}\quad \forall i,j\in [n].
\end{equation*}
Note that $R^T = R$, and for any matrix $A\in \mathbb{R}^{n\times n}$ $\|RA\|_2 = \|A\|_2$ and $\|RA\|_F = \|A\|_F$.
\end{definition}

We now give the formal statement of our reduction.
\begin{lemma}\label{lem:x_greater_than_1_reduction}
For Hankel $H= V_X D_X V_X^T$,where all $x\in X$ satisfy $|x|>1$ and $D_X= diag(\{a_x\}_{x\in X})$, let $X^{-1} = \{1/x\}_{x\in X}$, $D_{X^{-1}}= diag(\{a_x x^{2n-2}\}_{x\in X})$ and $H_{X^{-1}} = V_{X^{-1}}D_{X^{-1}}V_{X^{-1}}^T$. 
Then $H_{X^{-1}}$ is also Hankel, and for the row order reversal operator $R$ (\Cref{def:row_reversal}) we have that $H = R H_{X^{-1}}R^T$.
\end{lemma}
\begin{proof}
For any $x\in X$, the corresponding column in $V_{X}$ is $v_{n}(x)$ and it can be written as
\begin{align*}
  v_n(x)&= [1,x,\ldots,x^{n-1}]^T \\
  &= x^{n-1}[1/x^{n-1},1/x^{n-2},\ldots,1]^T \\
  &= x^{n-1} R [1,1/x,\ldots,1/x^{n-1}]^T \\&= x^{n-1} R v_n(1/x)  .
\end{align*}
Hence if we let $D=diag(\{x^{n-1}\}_{x\in X})$, we have that $V_{X} = R V_{X^{-1}}D$. Thus,
\begin{equation*}
   H = V_X D_X V_X^T = R V_{X^{-1}}D D_X D V_{X^{-1}}^T R^T = R V_{X^{-1}}D_{X^{-1}} V_{X^{-1}}^T R^T=R H_{X^{-1}} R^T.
\end{equation*}
\end{proof}

Next, we present the proof of Lemma \ref{lem:x_1_above} by combining the above with Lemma \ref{lem:x_0_1}.
\begin{proof}[Proof of Lemma \ref{lem:x_1_above}]
Every $x\in X^{-1}$ satisfies $|x|\leq 1$, and $H_{X^{-1}}$ is Hankel as well. Thus applying Lemma \ref{lem:x_0_1} to $H_{X^{-1}}$ we get that there is a Hankel matrix $\wh{H}_{X^{-1}}$ of rank $O(\log(n)\log(1/\epsilon))$ such that $\|H_{X^{-1}}-\wh{H}_{X^{-1}}\|_2\leq \epsilon\|H_{X^{-1}}\|_2$. Now note that since $R$ is just the row order reversal operator, the rank of $R\wh{H}_{X^{-1}}R^T$ is still $O(\log(n)\log(1/\epsilon))$. Moreover, $\|R(H_{X^{-1}}-\wh{H}_{X^{-1}})R^T\|_2 = \|H_{X^{-1}}-\wh{H}_{X^{-1}}\|_2\leq \epsilon \|H_{X^{-1}}\|_2 = \epsilon \|RH_{X^{-1}}R^T\|_2$. Thus if we let $\wh{H} = R \wh{H}_{X^{-1}} R^T$ which is also Hankel from the definition of $R$, $\|H- \wh{H}\|_2 = \|RH_{X^{-1}}R^T-R\wh{H}_{X^{-1}}R^T\|_2 \leq \epsilon \|R H_{X^{-1}} R^T\|_2 = \epsilon \|H\|_2$. Repeating this argument for the Frobenius norm completes the proof of the lemma.
\end{proof}

Equipped with Lemmas \ref{lem:x_0_1} and \ref{lem:x_1_above}, we are finally prove our main existence result of Theorem \ref{thm:main_thm}, which covers all cases of values in $X$.
\begin{proof}[Proof of Theorem \ref{thm:main_thm}]
Let $\Delta\in \mathbb{R}^{n\times n}$ be a positive definite Hankel matrix scaled so that $\|\Delta\|_F\leq \epsilon/\poly(n)\|H\|_2$. Perturb $H$ by adding $\Delta$ to it, then $H+\Delta$ is positive definite Hankel and it differs from $H$ in spectral and Frobenius norm by at most $ \epsilon/\poly(n)\|H\|_2$. Thus proving the guarantee of the theorem for $H+\Delta$ proves it for $H$ up to rescaling $\epsilon$ by a constant. Thus, w.l.o.g., we can assume $H$ is positive definite, and apply Lemma \ref{lem:vandermonde_decomp} to it to obtain $H=V_XD_XV_X^T$.  For given $X$ consider the partition $X=X_1\cup X_2$ where $X_1 = \{x\in X: |x|\leq 1\}$ and $X_2 = \{x\in X: |x|>1\}$. Thus by Lemma \ref{lem:partition_x}, $H=H_{X_1}+H_{X_2}$ where $H_{X_1} = V_{X_1}D_{X_1}V_{X_1}^T$ and $H_{X_2} = V_{X_2}D_{X_2}V_{X_2}^T$. Since $H_{X_1},H_{X_2}\succeq 0$, $H\succeq H_{X_1}$ and $H\succeq H_{X_2}$. Thus we have that $\|H\|_2\geq \|H_{X_1}\|_2$ and $\|H\|_2\geq \|H_{X_2}\|_2$, which implies $\|H\|_2 \geq  (\|H_{X_1}\|_2+\|H_{X_2}\|_2)/2$. Now let $\wh{H}_{X_1}$ be as per Lemma \ref{lem:x_0_1} applied to $H_{X_1}$ and $\wh{H}_{X_2}$ be as per Lemma \ref{lem:x_1_above} applied to $H_{X_2}$. Let $\wh{H} = \wh{H}_{X_1}+\wh{H}_{X_2}$. Then $\wh{H}$ is Hankel and its rank is at most $O(\log n \log(1/\epsilon))$. We also have
\begin{equation*}
    \|H-\wh{H}\|_2\leq \|H_{X_1}-\wh{H}_{X_1}\|_2+\|H_{X_2}-\wh{H}_{X_2}\|_2\leq \epsilon(\|H_{X_1}\|_2+\|H_{X_2}\|_2)\leq 2\epsilon \|H\|_2.
\end{equation*}
Rescaling $\epsilon$ by $1/2$ and repeating the above argument identically for the Frobenius norm completes the proof of the Theorem.
\end{proof}

\section{Lower Bound on the  Approximate Rank of Hankel Matrices}\label{sec:lower_bd}
We now give the proof of the lower bound result of Theorem \ref{thm:epsilon_rank_lower_bound}. The proof will be based on two cases, depending on the relation between $\epsilon$ and $n$.
\begin{proof}[Proof of Theorem \ref{thm:epsilon_rank_lower_bound}]
The first case is when $\epsilon\geq 1/n$. In this case, we will show the existence of a Hankel matrix with rank roughly $\log n$ and all eigenvalues very close to $1$. In the other case, when $\epsilon\leq 1/n$, we will argue that the $n$ by $n$ Hilbert matrix has at least roughly $\log(1/\epsilon)$ eigenvalues above $\epsilon$ times the largest eigenvalue.

\smallskip

\noindent\textbf{Case when $\epsilon\geq 1/n$.} In this case, we construct a PSD Hankel matrix with top eigenvalue roughly $1$ and $\Omega(\log n)$ eigenvalues that are $\Omega(1)$.  Let $R = c'\log n$ for some small enough constant $c'>0$ and consider the set $X=\{e^{-2^{Ci}/n}: \text{integer }i\in [1,R]\}$ for constant $C=1/10c'$, corresponding moment vectors $v_n(e^{-2^{Ci}/n})$ (see  \Cref{def:moment_vector}) and weights $w_i= 1/\|v_n(e^{-2^{Ci}/n})\|_2^2$ for all $i\in [1,R]$. Consider the diagonal matrix $D_X\in \mathbb{R}^{R \times R}$ defined as $(D_X)_{i,i} = w_i$ for all integer $i\in [1,R]$ and the PSD Hankel matrix $H=V_XD_XV_X^T\in \mathbb{R}^{n\times n}$ of rank $R$. Since $D_X$ has positive entries on the diagonal, we can express $H$ as the outer product $H=(V_X\sqrt{D_X})(V_X\sqrt{D_X})^T$, hence its eigenvalues are identical to those of $(V_X\sqrt{D_X})^T V_X\sqrt{D_X} = \sqrt{D_X}V_X^T V_X\sqrt{D_X} \in \mathbb{R}^{R\times R}$. Next we have that by definition the $i^{th}$ diagonal entry of $\sqrt{D_X}V_X^T V_X\sqrt{D_X}$ for any integer $i\in [1,R]$ are equal to $w_i v_n(e^{-2^{Ci}/n})^T v_n(e^{-2^{Ci}/n}) = 1 $. Moreover, for any integer $i\in [1,R]$ we have the following bound on $w_i$:
\begin{align*}
    w_i &= \frac{1}{\|v_n(e^{-2^{Ci}/n})\|_2^2}\\
    &= \frac{1}{\sum_{l=0}^{n-1}e^{-2^{Ci+1}l/n}} \\
    & = \frac{1-e^{-2^{Ci+1}/n}}{1-e^{-2^{Ci+1}}}\\
    & \leq 10\cdot \frac{2^{Ci}}{n}\numberthis \label{eq:bound_w_i},
\end{align*}
where the last inequality followed from the fact that $2^{Ci}/n\leq 2^{CR}/n =o(1)$ for $R=c'\log n$ and $C=1/10c'$, thus $1-e^{-2^{Ci+1}/n} = (1\pm o(1))\cdot 2^{Ci+1}/n$. 

We now argue that $\sqrt{D_X}V_X^T V_X\sqrt{D_X}$ is very close the $R\times R$ identity matrix $I_R$. Formally, let $E = \sqrt{D_X}V_X^T V_X\sqrt{D_X} - I_R$. Our first goal is to bound $\|E\|_2$. Then, using Weyl's inequality (see Section 1.3 in \cite{tao2012topics}), we have that $|\lambda_i(\sqrt{D_X}V_X^T V_X\sqrt{D_X}) - \lambda_i(I_R)|\leq \|E\|_2$ for all $i\in [1,R]$, and since the $R$ eigenvalues of $I_R$ are all $1$, we have that $\lambda_i(\sqrt{D_X}V_X^T V_X\sqrt{D_X}) = 1\pm \|E\|_2$ for all $i\in [1,R]$. Now we bound $\|E\|_2$ using Lemma \ref{lem:spectral_norm_bound}. For this, note that since $E$ is the difference of two symmetric matrices, $E$ is symmetric. Thus we have that $\|E\|_{\infty} = \|E\|_1$, and from Lemma \ref{lem:spectral_norm_bound} we have
\begin{equation*}
    \|E\|_2\leq \|E\|_{\infty} = \max_{i\in [R]}\sum_{j\in [R]:j\neq i}|(\sqrt{D_X}V_X^T V_X\sqrt{D_X})_{i,j}|.
\end{equation*}
Now, for any $i,j\in [1,R]$ with $i\neq j$ we have the following bound on $(\sqrt{D_X}V_X^T V_X\sqrt{D_X})_{i,j}$:
\begin{align*}
    (\sqrt{D_X}V_X^T V_X\sqrt{D_X})_{i,j} & = \sqrt{w_i w_j} \cdot v_n(e^{-2^{Ci}/n})^T v_n(e^{-2^{Ci}/n})\\
    &\leq 100 \frac{\sqrt{2}^{C(i+j)}}{n}\cdot \sum_{l=0}^{n-1} e^{-(2^{Ci}+2^{Cj})l/n}\quad \text{( from  \eqref{eq:bound_w_i})}\\
    &\leq 100 \frac{\sqrt{2}^{C(i+j)}}{n}\cdot \frac{1}{1-e^{-(2^{Ci}+2^{Cj})/n}}\\
    &\leq 100 \frac{\sqrt{2}^{C(i+j)}}{n}\cdot \frac{n}{2^{Ci}+2^{Cj}}\\
    &=100 \frac{\sqrt{2}^{C(i+j)}}{2^{Ci}+2^{Cj}}.
\end{align*}
This implies the following bound on $\sum_{j\in [1,R]:j\neq i}|(\sqrt{D_X}V_X^T V_X\sqrt{D_X})_{i,j}|$ for any $i\in [R]$:
\begin{align*}
   \sum_{j\in [1,R]:j\neq i}|(\sqrt{D_X}V_X^T V_X\sqrt{D_X})_{i,j}| &\leq  100 \sum_{j\in [1,R]:j\neq i}\frac{\sqrt{2}^{C(i+j)}}{2^{Ci}+2^{Cj}}\\
   &\leq 100\sum_{j=-\infty}^{i-1}\frac{\sqrt{2}^{C(i+j)}}{2^{Ci}+2^{Cj}}+\sum_{j=i+1}^{\infty}\frac{\sqrt{2}^{C(i+j)}}{2^{Ci}+2^{Cj}}\\
   &\leq 100 \sum_{j=-\infty}^{i-1}\sqrt{2}^{C(j-i)}+ 100 \sum_{j=i+1}^{\infty}\sqrt{2}^{C(i-j)}\\
   & = 100\cdot \frac{2\cdot \sqrt{2}^{-C}}{1-\sqrt{2}^{-C}}\leq 2^{-C/10}.
\end{align*}
This implies that since $\|E\|_2\leq  \max_{i\in [1,R]}\sum_{j\in [1,R]:j\neq i}|(\sqrt{D_X}V_X^T V_X\sqrt{D_X})_{i,j}|$, $\|E\|_2\leq 2^{-C/10}$. Hence we have that $\lambda_{i}(\sqrt{D_X}V_X^T V_X\sqrt{D_X}) = 1\pm 2^{-C/10}$, for all $i\in [1,R]$. Thus, for any $\epsilon<(1-2^{-C/10})/(1+2^{-C/10})$ we have that the number of eigenvalues of $H$ above $\epsilon \lambda_1(H)$ are at least $R=\Omega(\log n) = \Omega(\log n + \log(1/\epsilon))$ since $\log(1/\epsilon)\leq \log n$. Setting $C$ to a large enough constant, the above argument works for any $\epsilon\leq 1-c$ for any desired constant $c>0$. Moreover, for any $k\leq R$ for any such $\epsilon\leq 1-c$  if  $\|H-H_k\|_2\leq \epsilon \|H\|_2$, $k$ has to be $\Omega(R)=\Omega(\log n) = \Omega(\log n+\log(1/\epsilon))$. For any such $k$ we also have that $\|H-H_k\|_F = (1\pm 2^{-C/10})\sqrt{(R-k)}$ and $\|H\|_F = (1\pm 2^{-C/10})\sqrt{R}$. Thus for any $\epsilon\leq 1-c$ if  $\|H-H_k\|_F\leq \epsilon \|H\|_F$, $k$ has to be $\Omega(R)=\Omega(\log n) = \Omega(\log n+\log(1/\epsilon))$. The converse of these two statements finishes the proof when $\epsilon\geq 1/n$.

\smallskip

\noindent\textbf{Case when $\epsilon\leq 1/n$.} In this regime, consider the $n$ by $n$ Hilbert matrix $H_n$, defined as $H_{i,j} = 1/(i+j+1)$ for all $i,j\in [n]$. $H_n$ is a positive definite Hankel matrix - it is clearly Hankel by definition and its positive definiteness follows from the fact that $(H_n)_{i,j} = 1/(i+j+1) = \int_{t=0}^{1}t^{i+j-2} dt$. Thus for any $x\in \mathbb{R}^{n}$ we have that $x^T H_nx = \sum_{i,j\in [n]}x_ix_j \int_{t=0}^{1}t^{i+j-2} dt =\int_{t=0}^1 (\sum_{i\in [n]}x_i t^{i-1})^2 dt>0$. It is known that $\|H_n\|_2\in [1,\pi]$ (Example 3.3 in \cite{Beckermann:2000}), and moreover, the smallest singular value $\lambda_n(H_n)$ satisfies the following lower bound (Equation 3.35 in \cite{wilf1970finite}),
\begin{equation}\label{eq:hilbert_lower_bd}
    \lambda_n(H_n) = \Theta(\sqrt{n}\cdot (1+\sqrt{2})^{-4n}).
\end{equation}
Our first goal will be to use these facts to first show that $\lambda_k(H_n)\geq \epsilon \|H_n\|_2$ for $k = \Omega(\log(1/\epsilon))$. 

Let $k=\log(1/(C'\epsilon \log n))$ for some large enough constant $C'>0$, then since $\epsilon\leq 1/n$ we have that $k=\Omega(\log(1/\epsilon)+\log(n))$. Let $H_k$ be the $k$ by $k$ Hilbert matrix, which is also a principal submatrix of $H_n$. Thus, by the interlacing property for eigenvalues (Theorem 8.17 in \cite{golub2013matrix}) and applying the lower bound of \eqref{eq:hilbert_lower_bd} to $H_k$ we have 
\begin{align*}
    \lambda_k(H_n)&\geq \lambda_k(H_k)\\
    &= \Omega(\sqrt{k}\cdot (1+\sqrt{2})^{-4k})\\
    &= \Omega(\pi \cdot 2^{-k})\\
    &=\Omega(\|H_n\|_2\cdot 2^{-k})\\
    &= C'\cdot\epsilon\cdot \log n \|H_n\|_2.
\end{align*}
Thus for the best rank $k$ approximation $(H_n)_k$, $\|H_n-(H_n)_k\|_2 \geq C'\cdot\epsilon\cdot \log n \|H_n\|_2 \geq \epsilon\|H_n\|_2$. Moreover, since $H_n$ is positive definite Hankel, thus applying Lemma \ref{lem:hankel_spectral_to_frob} we have
\begin{align*}
    \|H_n-(H_n)_k\|_F&\geq \|H_n-(H_n)_k\|_2\geq C'\cdot\epsilon\cdot \log n \|H_n\|_2 \\&\geq \epsilon\cdot \log n \cdot (\|H_n\|_F/\sqrt{\log n})\geq \epsilon\|H_n\|_F.
\end{align*}
Thus, again in this case there exists a constant $C$ such that if $k\leq C(\log n+\log(1/\epsilon))$, then $\|H_n-(H_n)_k\|_F\geq \epsilon\|H_n\|_F$ and $\|H_n-(H_n)_k\|_2\geq \epsilon\|H_n\|_2$.  
\end{proof}


\section{Sublinear Time Algorithm for Hankel Low Rank Approximation.}\label{sec:alg_main}
We next present the proof of the sublinear time robust low-rank approximation algorithm of Theorem \ref{thm:const_factor_approx_main}, which we restate below.
\begin{reptheorem}{thm:const_factor_approx_main}[Robust Hankel low rank approximation algorithm]
For any PSD Hankel $H\in \mathbb{R}^{n\times n}$ suppose we are given entrywise access to $H+E$ for arbitrary noise matrix $E\in \mathbb{R}^{n\times n}$. Then Algorithm \ref{alg:noisy_hankel_recovery} in time $\poly(\log n,\log(1/\epsilon))$ returns a rank $O(\log n \log(1/\epsilon))$ rank Hankel matrix $\wh{H}$ (in a factored representation) such that $\|H-\wh{H}\|_F \leq O(\|E\|_F) + \epsilon \|H\|_F$ holds with probability at least $0.99$.
\end{reptheorem}

Our starting point is the following lemma, which presents the results and properties needed from the existence proof results in \Cref{sec:existence_main} for the algorithm. This lemma allows us to set up the task of recovering a good low-rank approximation to PSD Hankel $H$ plus adversarial noise $E$ as a matrix regression task,

\begin{lemma}\label{lem:existence_of_vandermonde_good_low_rank_approximation}
Let $T=\cup_{i\in [R+1]}T_r$ for $\{T_r\}_{r\in [R+1]}$ defined as per Lemma \ref{lem:bucket_sparsification} in \Cref{sec:bucket_sparsification}, and let $R$ be the row order reversal operator (\cref{def:row_reversal}). Let $V\in \mathbb{R}^{n\times 2|T|}$ be the matrix defined as $[V_T; RV_T]$ obtained by concatenating the columns of the Vandermonde matrix $V_T$ corresponding to the set $T$ and $R V_T$. Then there exists a diagonal matrix $D^*\in \mathbb{R}^{2|T|\times 2|T|}$ and a Hankel matrix $H^*\in \mathbb{R}^{n\times n}$ satisfying $(H^*)_{i,j}=0$ for all $i,j\in [n]$ such that $i+j\in [C\log(1/\epsilon),2n-C\log(1/\epsilon)]$ for some constant $C>0$ and,
\begin{equation*}
    \|VD^*V^T + H^*  -H\|_F \leq \epsilon\|H\|_F.
\end{equation*}
Moreover, the sum of absolute entries of $D^*$ is bounded, satisfying 
\begin{equation*}
    \sum_{i\in [2|T|]}|D^*_{i,i}| = O(\log(1/\epsilon)\log(n) \|H\|_F).
\end{equation*}
\end{lemma}
\begin{proof}
Again, using an arbitrarily small positive definite Hankel perturbation, we can assume $H$ to be positive definite. Consider the application of Lemma \ref{lem:vandermonde_decomp} to $H$ and let $H=V_X D_X V_X^T$ be its Fiedler/Vandermonde decomposition. Now let $X_1 = \{x\in X: |x|\leq 1\}, X_2 = X\setminus X_1 = \{x\in X: |x|>1\}$ and $H_1 = V_{X_1}D_{X_1}V_{X_1}^T$ and $H_2 = H-H_1 = V_{X_2}D_{X_2}V_{X_2}^T$. 

First, consider the application of Definition \ref{def:bucketing} to bucket the elements of $X_1$ as $X_1=\cup_{r\in [R+2]}B_r$ with corresponding weights $\{w_r\}_{r\in [R+2]}$ and $H_{B_r} = V_{B_r}D_{B_r}V_{B_r}^T$ for all $r\in [R+2]$, and Lemmas \ref{lem:x_0_1_errs} and \ref{lem:hankel_spectral_lower_bd} to $H_1$ to obtain Hankel matrices $V_{T} D_{T} V_{T}^T, \wh{H}_{B_{R+1}}$ and the lower bound on $\|H_1\|_F$.  From the guarantees of Lemma \ref{lem:x_0_1_errs} and \ref{lem:hankel_spectral_lower_bd} we have that $\|V_{T} D_{T} V_{T}^T +\wh{H}_{B_{R+1}}- H_1\|_F\leq O(\epsilon \log(1/\epsilon)\|H_1\|_F)$. Moreover, $ (\wh{H}_{B_{R+1}})_{i,j}\neq 0$ only for $i,j\in [n]$ such that $i+j\leq C\log(1/\epsilon)$ for some constant $C>0$. Setting $\epsilon$ to $\epsilon^2$ we have
\begin{equation*}
    \|V_{T} D_{T} V_{T}^T +\wh{H}_{B_{R+1}}- H_1\|_F\leq \epsilon \|H_1\|_F.
\end{equation*}
Moreover, by Lemma \ref{lem:bucket_sparsification} we have that $\sum_{i\in [|T|]}|{D_T}_{i,i}| = \sum_{r\in [R+1]}\sum_{i,r}|a^+_{i,r}|+|a^-_{i,r}| = O(\log(1/\epsilon)\cdot \sum_{x\in X}a_x)$. It further satisfies
\begin{align*}
    \sum_{i\in [|T|]}|{D_T}_{i,i}| &= O(\log(1/\epsilon)\cdot \sum_{x\in X}a_x)\\
    &=O(\log(1/\epsilon)\cdot \sum_{r\in [R+2]}w_r)\\
    &= O(\log(1/\epsilon) \cdot R\max_{r\in [R+2]}w_r(n/2^r\lambda))\\
    &= O(\log(1/\epsilon)\log(n) \|H_1\|_2)\numberthis \label{eq:d_t_l1_bound}.
\end{align*}

Now consider $X_2$ and $H_2 = V_{X_2}D_{X_2}V_{X_2}^T$, and consider the row order reversal operator $R$ (\Cref{def:row_reversal}) and apply Lemma \ref{lem:x_greater_than_1_reduction} to $H_2$ to obtain $H_{X_2^{-1}}$ satisfying $H_2 = R H_{X_2^{-1}} R^T$. We repeat the argument of the previous paragraph for $H_{X_2^{-1}}$ since all $x\in X_2^{-1}$ satisfy $|x|\leq 1$ to obtain Hankel matrices $V_{T}D'_TV_{T}^T$ and $\wh{H}_{B'_{R+1}}$ such that $\|V_{T}D'_TV_{T}^T+\wh{H}_{B'_{R+1}} - H_{X_2^{-1}}\|_F\leq \epsilon\|H_{X_2}^{-1}\|_F$. Moreover, $ (\wh{H}_{B'_{R+1}})_{i,j}\neq 0$ only for $i,j\in [n]$ such that $i+j\leq C\log(1/\epsilon)$ for some constant $C>0$. This also implies
\begin{align*}
  \|RV_{T}D'_TV_{T}^TR^T + R\wh{H}_{B'_{R+1}}R^T - H_2\|_F &= \|V_{T}D'_{T}V_{T}^T + \wh{H}_{B'_{R+1}} - RH_2R^T\|_F\\ &\leq \epsilon\|H_{X_2}^{-1}\|_F = \epsilon\|H_2\|_F.
\end{align*}
Finally, similar to \eqref{eq:d_t_l1_bound}, we have the following bound on $\sum_{i\in [|T|]}|{D'_T}_{i,i}|$:
\begin{equation}\label{eq:d_prime_t_l1_bound}
    \sum_{i\in [|T|]}|{D'_T}_{i,i}| = O(\log(1/\epsilon)\log(n) \|H_2\|_2).
\end{equation}
Thus if we let $V = [V_T;R V_T]\in \mathbb{R}^{n\times 2|T|}$ and diagonal $D^*\in \mathbb{R}^{2|T|\times 2|T|}$ obtained by concatenating the diagonals of $D_T$ and $D'_T$, and using Lemma \ref{lem:partition_x} to obtain $\|H\|_F\geq \max\{\|H_1\|_F,\|H_2\|_F\}$ (since $H\succeq H_1$ and $H\succeq H_2)$, then we have
\begin{align*}
    \|VD^*V^T + \wh{H}_{B_{R+1}}+  R\wh{H}_{B'_{R+1}}R^T  -H\|_F &= \|VD^*V^T + \wh{H}_{B_{R+1}}+  R\wh{H}_{B'_{R+1}}R^T  -H_1+H_2\|_F \\
    &\leq \epsilon (\|H_1\|_F+\|H_2\|_F)\leq 2\epsilon\|H\|_F.
\end{align*}
Rescale $\epsilon$ by $2$ so the above upper bound is $\epsilon \|H\|_2$, and set $H^* =\wh{H}_{B_{R+1}}+  R\wh{H}_{B'_{R+1}}R^T$ which satisfies the claimed guarantee in the statement of the lemma by its definition. Moreover, we have the following bound on $\sum_{i\in [2|T|]}|D^*_{i,i}|$:
\begin{align*}
    \sum_{i\in [2|T|]}|D^*_{i,i}|&= \sum_{i\in [|T|]}|{D_T}_{i,i}|+\sum_{i\in [|T|]}|{D'_T}_{i,i}|
    \\&= O(\log(1/\epsilon)\log(n)(\|H_1\|_2+\|H_2\|_2)
    \\&= O(\log(1/\epsilon)\log(n) \|H\|_2)\numberthis \label{eq:d_star_l1_bound},
\end{align*}
where in the last inequality we used the $\|H\|_2\geq \max\{\|H_1\|_2,\|H_2\|_2\}$, as again implied by Lemma \ref{lem:partition_x}. Finally, using $\|H\|_2\leq \|H\|_F$, this completes the proof of the lemma.
\end{proof}

The above lemma says that there exists a good Hankel low rank approximation to $H$ of the form $VD^*V^T+ H^*$ where $V$ is fixed and diagonal $D^*$, with bounded entries, and Hankel $H^*$, with only $O(\log(1/\epsilon))$ nonzero rows, is unknown. The first part of our algorithm will estimate $H^*$, which can be done efficiently since it has few nonzero rows. We will then mask out those entries of the input where $H^*$ is nonzero, and use a ridge leverage score sampling approach to estimate $D^*$. This section is organized as follows: In \Cref{sec:vandermonde_ridge_lev_score_bd} we focus on proving good upper bounds on the ridge leverage scores of masked versions of real Vandermonde matrices, which can then be used for $V$ as well. Then, in  \Cref{sec:sublinear_time_recovery} we present Algorithm \ref{alg:noisy_hankel_recovery} and proof of its guarantees, giving Theorem \ref{thm:const_factor_approx_main}.
\subsection{Ridge Leverage Score Bounds for Vandermonde Matrices}\label{sec:vandermonde_ridge_lev_score_bd}

In this section, we prove the following lemma, which upper bounds the ridge leverage scores of Vandermonde matrices where the first and last few rows have been zeroed out. As discussed above, we need to have bounds for such matrices with some rows zeroed out since we will be using them to solve a masked version of a matrix regression problem. Throughout this section, $\gamma>0$ will be the ridge parameter. We now present notation that will be used throughout this section.
\begin{definition}\label{def:lev_score_setup}
For any $X\subset \mathbb{R}$ of size $|X|=k\leq n$, consider the Vandermonde matrix $V_X\in \mathbb{R}^{n\times k}$ (\Cref{def:vandermonde_matrix}). For any $n_0\leq O(\log(1/\epsilon))$, let $V_X[n_0:n-n_0,:]$ be the set of all except the first and last $n_0$ rows of $V_X$. We will consider diagonal $D = diag(\{\max\{1,x^{n-1}\}\}_{x\in X})$ to rescale columns of $V_X[n_0:n-n_0,:]$ as $V_X[n_0:n-n_0,:]D^{-1}$ so all columns have entries at most $1$.
\end{definition}
With this notation, we now present the universal upper bounds on the ridge leverage scores of $V_X[n_0:n-n_0,:]D^{-1}$.
\begin{lemma}\label{lem:vandermonde_ridge_lev_score_bd}
Consider the setup of Definition \ref{def:lev_score_setup}. Let $r'$ be such that $n/2^{r'}\leq n_0\leq n/2^{r'-1}$ and $r_0 = \min\{r',\log(n/1000\log^2(n/\gamma))\} $. There exists closed form expressions for upper bounds on the $\gamma$-ridge leverage scores $\wt{\tau}_{i,\gamma}\geq \tau_{i,\gamma}(V_X[n_0:n-n_0,:]D^{-1})$ as follows:
\begin{equation*}
    \wt{\tau}_{i,\gamma}= \begin{cases}
        4000\log(n/\gamma)^2/(n/2^r)\hspace{0.1cm} \forall i\in [n/2^r,n/2^{r-1}], r\in [2,r_0],\\
        4000\log(n/\gamma)^2/(n/2^r)\hspace{0.1cm} \forall i\in [n-n/2^{r-1},n-n/2^r], r\in [2,r_0],\\
        1\quad \text{otherwise}. 
    \end{cases}
\end{equation*}
Moreover, their sum satisfies
\begin{equation*}
\sum_{i\in [n_0,n-n_0]}\wt{\tau}_{i,\gamma} \leq 8000(\log^3(n/\gamma)+n_0).
\end{equation*}
\end{lemma}
To prove the above lemma, we will need the following bound on the leverage scores of low-degree polynomials. This is the discrete version of Lemma 4.13 in \cite{meyer2023near}.
\begin{lemma}[Discrete version of Lemma 4.13 in \cite{meyer2023near}]\label{lem:polynomial_lev_score_bd}
For any polynomial $p(t)$ of degree $d$ defined over integers $t$ in an interval $[a,b]$ for integers $a\leq b$ satisfying $L=b-a\geq 100d^2$, we have the following,
\begin{equation*}
    |p(t)|^2\leq 100d^2\left(\frac{(\sum_{t\in [a,b]}p(t)^2)}{L}\right)\quad \forall t\in [a,b].
\end{equation*}
\end{lemma}
\begin{proof}
Let $L= b-a$ and let $I$ be the uniform grid of $L$ points with width $2/L$ in $[-1,1]$. That is, $I = \{-1,-1+2/L,\ldots,1-2/L,1\}$. Let $q(.)$ be the polynomial defined over $t\in I$ obtained by scaling and shifting the polynomial $p(.)$ defined over the interval $[a,b]$. Fix a $t\in [a,b]$ at which we want to prove the leverage score bound, and let $t_I\in I$ be the corresponding point in $I$. Without loss of generality, assume that $q(t_I)=1$ by rescaling. Let $C = \max_{t\in I}|q(t)|$, thus $C\geq 1$ after the rescaling described previously, and $s^* = \argmax_{t\in [-1,1]}|q(t)|$. By Markov brother's inequality (Lemma \ref{lem:markov_brother}) we have that $|q'(s^*)|\leq C d^2$, then for $s^*_I\in I$ to be the nearest point in $I$ to $s^*$ we have that since $|s^*-s^*_I|\leq 2/L$, $q(s^*_I)\geq q(s^*) - Cd^2\cdot 2/L = C - Cd^2\cdot 2/L\geq 0.99 C$, since $L\geq 100d^2$. Now again by Markov brother's inequality, we have that $q(s^*_I+s)\geq q(s^*_I)- Cd^2s \geq 0.99C - C d^2 s$ for all $|s|\leq 1/d^2$. Let $I^* = I \cap [s^*_I-1/d^2,s^*_I+1/d^2]$, then $|I^*|\geq L/d^2$. Thus, we have
\begin{align*}
    \sum_{t\in I} |q(t)|^2 &\geq \sum_{t\in I^*}|q(t)|^2\\
    & \geq \sum_{t\in I^*}(0.99C- Cd^2 |t-s^*_I|)^2\\
    &\geq \sum_{\Delta=2/L}^{1/2d^2} (0.99C- Cd^2 \Delta)^2 \\
    &\geq \sum_{\Delta=2/L}^{1/2d^2} C^2/8\\
    & = (C^2/8)\cdot (L/4d^2-1)\\
    & \geq (C^2/8)\cdot L/8d^2 \quad (\text{ since $L>100d^2$})\\
    &\geq 0.01L/d^2 \quad (\text{ since $C\geq 1$}).
\end{align*}
This implies that $|q(t)^2|\leq 100d^2 (\sum_{t\in I} |q(t)|^2/L)$ for all $t\in I$. Thus as per the definition of $q(t)$ and $I$ in terms of $L$ and $p(t)$, this statement is equivalent to $|p(t)|^2 \leq 100d^2 (\sum_{t\in [a,b]} |p(t)|^2/L$.
\end{proof}
Equipped with the above, we now present the proof of Lemma \ref{lem:vandermonde_ridge_lev_score_bd}.
\begin{proof}[Proof of Lemma \ref{lem:vandermonde_ridge_lev_score_bd}]
We will index the rows of $V_{X}[n_0:n-n_0,]D^{-1}$ with the interval $[n_0,n-n_0]$. We will bound the ridge leverage scores in two parts. First we consider bounding them for the first half of the rows $i\in [n_0,n/2]$, and then for $i\in [n/2,n-n-n_0]$. Consider the following geometric grouping of the first $n/2$ rows of $V_X D^{-1}$ -- Let $V_{X,r} = V_X[n/2^r:n/2^{r-1},:]D^{-1}\in \mathbb{R}^{n/2^r\times k}$ be the matrix containing the subset of rows of $V_XD^{-1}$ from $[n/2^r:n/2^{r-1}]$ for all $r\in [2,r_0]$. Now consider a fixed $r\in [2,r_0]$. Since for any $v\in \mathbb{R}^{k}$ we have that 
\begin{equation*}
 \|V_X[n_0:n-n_0] D^{-1}v\|_2^2 \geq \|V_{X,r}v\|_2^2,   
\end{equation*}
it implies that $\tau_{i,\gamma}(V_X[n_0:n-n_0] D^{-1})\leq \tau_{i,\gamma}(V_{X,r})$. Let $\lambda =2\log(1/\gamma)$. Now for any $x\in X$ such that $|x|\leq e^{-2^r\lambda/n}$, we have that $|x|^i \leq \gamma$ for all $i\in [n/2^r, n/2^{r-1}]$. Moreover, for any $x\in X$ with $|x|\geq 1$ we have that $D_{x,x} = x^{n-1}$. Hence for all $x$ such that $|x|\geq e^{4\lambda/n}$ $|x^iD_{x,x}^{-1}|= |x|^{-(n-i-1)}\leq \gamma$ for all $i\in [n/2]$. This means that only those $x\in X$ with $e^{-2^r\lambda/n}\leq |x|\leq e^{4\lambda/n}$ effectively contribute to $V_{X,r}v$, as the rest are at most $\gamma$. This is now formalized as follows -- for any $v\in \mathbb{R}^k$ such that $\|v\|_2=1$ and any $i\in [n/2^r, n/2^{r-1}]$ we have
\begin{align*}
(V_{X,r}v)_i  = \sum_{x\in X}v_x D_{x,x}^{-1} x^{i} =&\sum_{x\in X: |x|\in  [e^{-2^r\lambda/n},e^{4\lambda/n}]}v_x D_{x,x}^{-1} x^{i} + \sum_{x\in X: |x|\leq e^{-2^r\lambda/n}}v_x x^{i} \\&+ \sum_{x\in X: |x|\geq e^{4\lambda/n}}v_x x^{-(n-1)} x^{i}.
\end{align*}
 To ease notation, we absorb the scaling of $D_{x,x}^{-1}$ into $v_x$. Thus, we have 
\begin{equation}\label{eq:row_bound_1}
    |(V_{X,r}v)_i - \sum_{\substack{x\in X:\\ |x|\in [e^{-2^r\lambda/n},e^{4\lambda/n}]}}v_x x^{i}|\leq \gamma \|v\|_1\leq \gamma \sqrt{n} \|v\|_2.
\end{equation}
Rescaling $\gamma$ by $\sqrt{n}$, the above bound becomes $\gamma$. Next, our goal is to show that the expression $\sum_{x\in X:\\ |x|\in [e^{-2^r\lambda/n},e^{4\lambda/n}]}v_x x^{i}$ can be well approximated using a \emph{low-degree polynomial} by truncating its Taylor series, using the fact that over $i\in [n/2^r,n/2^{r-1}]$ it behaves smoothly. Let $y_x = \ln(1/|x|)$ for all $x\in X$. Then we have, for $l=10\log(n/\gamma)$,
\begin{align*}
    \sum_{\substack{x\in X:\\ |x|\in [e^{-2^r\lambda/n},e^{4\lambda/n}]}}v_x x^{i} &= \sum_{\substack{x\in X: x\geq 0\\  y_x\in [-4\lambda/n,2^r\lambda/n]}}v_x e^{-y_x i} + (-1)^i \sum_{\substack{x\in X: x\leq 0\\  y_x\in [-4\lambda/n,2^r\lambda/n]}}v_x e^{-y_x i} \\
    &= \sum_{\substack{x\in X: x\geq 0\\  y_x\in [-4\lambda/n,2^r\lambda/n]}}v_x \sum_{m=0}^{l} (-y_x i)^m/m! + (-1)^i \sum_{\substack{x\in X: x\leq 0\\  y_x\in [-4\lambda/n,2^r\lambda/n]}}v_x \sum_{m=0}^{l} (-y_x i)^m/m!\\
    &+ \sum_{\substack{x\in X: x\geq 0\\  y_x\in [-4\lambda/n,2^r\lambda/n]}}v_x \sum_{m=l}^{\infty} (-y_x i)^m/m! + (-1)^i \sum_{\substack{x\in X: x\leq 0\\  y_x\in [-4\lambda/n,2^r\lambda/n]}}v_x \sum_{m=l}^{\infty} (-y_x i)^m/m! \\
    &= p^+(i)+(-1)^i p^-(i)\\
    &+ \sum_{\substack{x\in X: x\geq 0\\  y_x\in [-4\lambda/n,2^r\lambda/n]}}v_x \sum_{m=l}^{\infty} (-y_x i)^m/m! + (-1)^i \sum_{\substack{x\in X: x\leq 0\\  y_x\in [-4\lambda/n,2^r\lambda/n]}}v_x \sum_{m=l}^{\infty} (-y_x i)^m/m! \numberthis \label{eq:taylor_approx_error_lev_score_bd},
\end{align*}
where $p^+(i)$ and $p^-(i)$ are the following degree $l=10\log(n/\gamma)$ polynomials defined over $i\in [n/2^r,n/2^{r-1}]$:
\begin{align*}
p^+(i) &= \sum_{\substack{x\in X: x\geq 0\\  y_x\in [-4\lambda/n,2^r\lambda/n]}}v_x \sum_{m=0}^{l} (-y_x i)^m/m!, \\p^-(i) &= \sum_{\substack{x\in X: x\leq 0\\  y_x\in [-4\lambda/n,2^r\lambda/n]}}v_x \sum_{m=0}^{l} (-y_x i)^m/m!.
\end{align*}
Now for $i\in [n/2^r,n/2^{r-1}]$, $|y_x i|\leq 2\lambda = 4\log(n/\gamma)$. Now using Stirling's approximation for $m!$ when $m\geq l$, we have that $\log(m!) = m\log(m) - m\log(e)+\log(\sqrt{2\pi \log(m)})+O(1/m)\geq m\log(m/e)$. This implies that $m! \geq (m/e)^m$. Thus, we have the following for the second term in \eqref{eq:taylor_approx_error_lev_score_bd} for $l=10\log(n/\gamma)$:
\begin{align*}
    |\sum_{\substack{x\in X: x\geq 0\\  y_x\in [-4\lambda/n,2^r\lambda/n]}}v_x \sum_{m=l}^{\infty} (-y_x i)^m/m! &+ (-1)^i \sum_{\substack{x\in X: x\leq 0\\  y_x\in [-4\lambda/n,2^r\lambda/n]}}v_x \sum_{m=l}^{\infty} (-y_x i)^m/m!| \\&\leq \sum_{m=l}^{\infty}\sum_{\substack{x\in X:\\  y_x\in [-4\lambda/n,2^r\lambda/n]}} |v_x| |y_x t|^m/m!\\
    &\leq \sum_{m=l}^{\infty}\sum_{\substack{x\in X:\\  y_x\in [-4\lambda/n,2^r\lambda/n]}} |v_x|  (4\log(1/\gamma)/l)^m\\
    &\leq\sum_{\substack{x\in X:\\  y_x\in [-4\lambda/n,2^r\lambda/n]}}  |v_x|  \cdot\sum_{m=l}^{\infty}0.1^m\leq (\gamma/n)\|v\|_1\leq \gamma \|v\|_2.
\end{align*}
 Note since the scaling of $D^{-1}$ was absorbed into $v$, $\|v\|_2$ is not necessarily $1$ anymore. Thus overall we have that for all $i\in [n/2^r,n/2^{r-1}]$,  
\begin{equation}\label{eq:row_bound_2}
    |\sum_{\substack{x\in X:\\ |x|\in [e^{-2^r\lambda/n},e^{4\lambda/n}]}}v_x x^{i}-p^+(i)-(-1)^i p^{-}(i)|\leq \gamma\|v\|_2.
\end{equation}
Combining the bounds from \eqref{eq:row_bound_1} and \eqref{eq:row_bound_2} and rescaling $\gamma$ by $2$, we have, for all $i\in [n/2^r,n/2^{r-1}]$,
\begin{equation}\label{eq:row_bound_3}
    |(V_{X,r}v)_i - p^{+}(i) - (-1)^i p^{-}(i)|\leq \gamma\|v\|_2,
\end{equation}
which formalizes the fact that any vector in the column span of $V_{X,r}$ can be approximated using low-degree polynomials.

For every $i\in [n/2^{r+1},n/2^r]$ let $p_{even}(i) = p^+(2i)+p^-(2i)$ and $p_{odd}(i) = p^+(2i+1)-p^-(2i+1)$ be polynomials of degree also $l=10\log(n/\gamma)$. Then for every even $i\in [n/2^r,n/2^{r-1}]$ $p^{+}(i) - (-1)^i p^{-}(i) = p_{even}(i/2)$ and for every odd $i\in [n/2^r,n/2^{r-1}]$ $p^{+}(i) - (-1)^i p^{-}(i) = p_{odd}((i-1)/2)$. Thus we get that for all even $i\in [n/2^r,n/2^{r-1}]$ as per \eqref{eq:row_bound_3},
\begin{equation}\label{eq:even_indices}
|(V_{X,r}v)_i - p_{even}(i/2)|\leq \gamma\|v\|_2.
\end{equation}
This implies the following after applying Lemma \ref{lem:polynomial_lev_score_bd} for polynomial $p_{even}(t)$ of degree $l=10\log(n/\gamma)$ defined over interval $[n/2^{r+1},n/2^r]$ of length $L= n/2^r - n/2^{r+1} = n/2^{r+1}$:
\begin{align*}
    |(V_{X,r}v)_i|^2 &\leq 2|p_{even}(i/2)|^2 + 2\gamma^2\|v\|_2^2 \quad  \text{(Using \eqref{eq:even_indices})}\\
    &\leq  10l^2 \left(\sum_{j=n/2^r:j \text{ even}}^{n/2^{r-1}}|p_{even}(j/2)|^2/L\right) + 2\gamma^2\|v\|_2^2\\
    &\leq 20l^2\left(\sum_{i=n/2^r:i \text{ even}}^{n/2^{r-1}}|(V_{X,r}v)_i|^2/L\right) + 20\gamma^2l^2\|v\|_2^2 \quad \text{(Using \eqref{eq:even_indices})}\\
    &\leq 40l^2 \|V_{X,r}v\|_2^2/(n/2^r) + 20\gamma^2l^2\|v\|_2^2,
\end{align*}
as long as the degree  $l=10\log(n/\gamma)$ and length of the interval $n/2^{r+1}$ satisfies $L\geq 100l^2$. That is,
\begin{align*}
&n/2^{r+1}\geq 100(10\log(n/\gamma))^2\\
&\implies r\leq \log(n/1000\log^2(n/\gamma)).
\end{align*}
Similarly applying the argument to all odd $i\in [n/2^r,n/2^{r-1}]$ and polynomial $p_{odd}(i)$, we have that for all odd $i\in [n/2^r,n/2^{r-1}]$,
\begin{equation*}
 |(V_{X,r}v)_i|^2\leq  40l^2 \|V_{X,r}v\|_2^2/(n/2^r) + 20\gamma^2l^2\|v\|_2^2.   
\end{equation*}
Thus, the above holds for all $i\in [n/2^r,n/2^{r-1}]$. Rescaling $\gamma$ to $\gamma^4/n$, we have that $|(V_{X,r}v)_i|^2\leq 4000\log(n/\gamma)^2 \|V_{X,r}v\|_2^2/(n/2^r) +\gamma\|v\|_2^2/(n/2^r)$ for all $i\in [n/2^r,n/2^{r-1}]$. Thus we have
\begin{align*}
    \frac{|(V_{X,r}v)_i|^2}{\|V_{X,r}v\|_2^2 +\gamma\|v\|_2^2}&\leq \frac{4000\log(n/\gamma)^2 \|V_{X,r}v\|_2^2/(n/2^r) +\gamma\|v\|_2^2/(n/2^r)}{\|V_{X,r}v\|_2^2 +\gamma\|v\|_2^2}\\&\leq 4000\log(n/\gamma)^2/(n/2^r).
\end{align*}
This implies that, for all $i\in [n/2^r,n/2^{r-1}]$ and for all $r\in [2,\log(n/1000\log^2(n/\gamma))]$,
\begin{equation*}
    \tau_{i,\gamma}(V_X[n_0:n-n_0,:]D^{-1})\leq \tau_{i,\gamma}(V_{X,r})\leq 4000\log(n/\gamma)^2/(n/2^r).
\end{equation*}
Let $r'$ be such that $n/2^{r'}\leq n_0\leq n/2^{r'-1}$ and $r_0 = \min\{r',\log(n/1000\log^2(n/\gamma))\} $. Then we set $\wt{\tau}_{i,\gamma}$ as 
\begin{equation*}
    \wt{\tau}_{i,\gamma}= \begin{cases}
        4000\log(n/\gamma)^2/(n/2^r)\hspace{0.1cm} \forall i\in [n/2^r,n/2^{r-1}], r\in [2,r_0],\\
        1\quad \forall i\in [n_0,n/2^{r_0}]. \numberthis\label{eq:tau_tilde_first_n/2_bound}
    \end{cases}
\end{equation*}

Now, for the case when $i\in [n/2,n-n_0]$, the above argument applies identically -- if we consider the row order reversal operator $R\in \mathbb{R}^{n-n_0\times n-n_0}$ (\Cref{def:row_reversal}), then $\tau_{i,\gamma}(V_X[n_0:n-n_0,:]D^{-1}) = \tau_{(n-n_0)-i,\gamma}(RV_X[n_0:n-n_0,:]D^{-1})$. Furthermore by definition the matrix $RV_X[n_0:n-n_0,:]D^{-1}$ can be viewed as $V_{X^{-1}}[n_0:n-n_0,:]D_{X^{-1}}^{-1}$ where $X^{-1}= \{1/x:x\in X, x\neq 0\}$ contains the inverse of all nodes in $X$, $V_{X^{-1}}$ is the corresponding Vandermonde matrix and $D_{X^{-1}} = diag(\{\max\{1,x^{n-1}\}\}_{x\in X^{-1}})$. Thus for $i\in [n/2,n-n_0]$, $(n-n_0)-i\in [n_0,n/2]$ and the above argument can be repeated for $RV_X[n_0:n-n_0,:]D^{-1}$ directly. Thus, we have
\begin{equation*}
    \wt{\tau}_{i,\gamma}= \begin{cases}
        4000\log(n/\gamma)^2/(n/2^r)\hspace{0.1cm} \forall i\in [n/2^r,n/2^{r-1}], r\in [2,r_0],\\
        4000\log(n/\gamma)^2/(n/2^r)\hspace{0.1cm} \forall i\in [n-n/2^{r-1},n-n/2^r], r\in [2,r_0],\\
        1\quad \text{otherwise}. 
    \end{cases}
\end{equation*}
This implies
\begin{align*}
    \sum_{i\in [n_0,n-n_0]}\wt{\tau}_{i,\gamma}
    &\leq 2\sum_{r\in [2,r_0]}\sum_{i\in [n/2^r,n/2^{r-1}]}4000\log(n/\gamma)^2/(n/2^r)\\
        &+ 2\sum_{i\in [n_0,n/2^{r_0}]}1\\
        &\leq 2\sum_{\substack{r\in [2,\log(n/1000\log^2(n/\gamma))]}}n/2^r\cdot 4000\log(n/\gamma)^2/(n/2^r) \\&+ 2n_0\\
        &\leq 8000 (\log^3(n/\gamma)+n_0).
\end{align*}
\end{proof}

Since we will need sampling guarantees for Vandermonde matrices where the first and last few rows have been masked, we introduce the following notation for masking such rows.
\begin{definition}[Row mask]\label{def:row_mask}
Let $M_i\in \mathbb{R}^{n\times |X|}$ defined as $M_i[:i,:] = 0$ and $1$ otherwise for all $i\in [n]$, i.e., $M_i$ zeroes out the first $i$ rows of $V_{X}$ after the Hadamard product $M_i\cdot V_{X}$.
\end{definition}

Equipped with the above definition, we now show the following ridge leverage score-based spectral approximation lemma for masked Vandermonde matrices using the upper bounds on ridge leverage scores from Lemma \ref{lem:vandermonde_ridge_lev_score_bd} and the sampling guarantee of Lemma \ref{lem:ridge_lev_score_sampling_guarantee}.
\begin{lemma}\label{lem:vandermonde_ridge_lev_score_sampling_guarantee}
Given $\delta\in (0,1/8)$, consider the setup of Definition \ref{def:lev_score_setup}. Consider the $\gamma$-ridge leverage score upper bounds $\{\wt{\tau}_{i,\gamma}\}_{i\in [n_0,n-n_0]}$ from Lemma \ref{lem:vandermonde_ridge_lev_score_bd}. Define $\wt{\tau}_{i,\gamma}=1$ for all $i\in [n_0]\cup [n-n_0,n]$, and let $S\in \mathbb{R}^{s\times n}$ be the sampling matrix of Lemma \ref{lem:ridge_lev_score_sampling_guarantee} instantiated with $\{\wt{\tau}_{i,\gamma}\}_{i\in [n]}$. 

Let mask $M_i$ for any $i\in [n]$ (\Cref{def:row_mask}) and $R$ be the row order reversal operator (\Cref{def:row_reversal}). Then $s=8000(n_0+\log^3(n/\gamma))\log(n/\delta)$, and with probability at least $1-\delta$ for all $y\in \mathbb{R}^{|X|}$ and all $i\in [n_0]$,
\begin{align*}
    \|SM_i\circ V_XD^{-1}y\|_2^2+\gamma \|y\|_2^2 &=\left(1\pm \frac{1}{4}\right) \left(\|M_i\circ V_XD^{-1}y\|_2^2 +\gamma \|y\|_2^2\right),\\
    \|S(RM_i)\circ V_XD^{-1}y\|_2^2+\gamma \|y\|_2^2 &=\left(1\pm \frac{1}{4}\right)\left(\|(RM_i)\circ V_XD^{-1}y\|_2^2 +\gamma \|y\|_2^2\right).
\end{align*}
Furthermore ,the sampling matrix $S$ satisfies $|S_{i,j}|\leq \poly(n)$ for all $i\in [s],j\in [n]$, and can be generated and multiplied with $M_i\circ V_XD^{-1}$ for any $i\in [n_0]$ in $|X|\cdot \polylog(n,1/\gamma)$ time.
\begin{proof}
For any $i\in [n_0]$ consider $M_i\circ V_X D^{-1}$. Then since $V_X[n_0:n-n_0,:] D^{-1}$ contains a subset of the rows of $M_i\circ V_X D^{-1}$ as long as $i\leq n_0$, we have that $\tau_{i,\gamma}(M_i\circ V_X D^{-1})\leq\tau_{i,\gamma}(V_X[n_0:n-n_0,:] D^{-1})\leq \wt{\tau}_{i,\gamma}$ for all $i\in [n_0,n-n_0]$. Moreover, $\tau_{i,\gamma}(M_i\circ V_X D^{-1})\leq 1$ holds trivially for all $i\in [n_0]\cup [n-n_0]$. Hence the guarantee of Lemma \ref{lem:ridge_lev_score_sampling_guarantee} holds when applied to $M_i\circ V_X D^{-1}$ for sampling matrix $S$ generated using upper bounds $\{\wt{\tau}_{i,\gamma}\}_{i\in [n]}$ for any fixed $i\in [n_0]$. Similarly, it also holds when applied to $RM_i\circ V_X D^{-1}$ as it just zeroes out the last $n_0$ rather than the first $n_0$ rows. Taking a union bound for the guarantee to hold for all $i\in [n_0]$ simultaneously, since $S$ is generated using sampling probabilities independent of $i$, over all $n_0$ possible values of $i$. Finally, note that we can sample and generate $S$ quickly using Lemma C.3 of \cite{bakshisubquadratic}. Since for any $m,n\in [n]$ the sum $\sum_{i=m}^n\wt{\tau}_{i,\gamma}$ can be computed exactly in constant time using their explicit expressions as per Lemma \ref{lem:vandermonde_ridge_lev_score_bd}, we can generate one sample using Lemma C.3 of \cite{bakshisubquadratic} in $\polylog(n,1/\gamma)$ time and thus construct $SM_i\circ V_XD^{-1}$ in time $|X|\cdot \polylog(n,1/\gamma)$. Finally note that $\min_{i\in [n]}\wt{\tau}_{i,\gamma}\geq 1/n$ and $s\geq 1$ as per Lemma \ref{lem:vandermonde_ridge_lev_score_bd}. Thus, each entry of $S$ is at most $\poly(n)$ by its definition as per Lemma \ref{lem:ridge_lev_score_sampling_guarantee}. This completes the proof of the lemma.
\end{proof}
\end{lemma}
In the next section, we present the proof of Theorem \ref{thm:const_factor_approx_main}.
\subsection{Noisy Hankel recovery}\label{sec:sublinear_time_recovery}
Using the tools developed in the previous section, our main goal in this section is to present the proof of Theorem \ref{thm:const_factor_approx_main}.
\begin{reptheorem}{thm:const_factor_approx_main}
For any PSD Hankel $H\in \mathbb{R}^{n\times n}$ suppose we are given entrywise access to $H+E$ for arbitrary noise matrix $E\in \mathbb{R}^{n\times n}$. Then Algorithm \ref{alg:noisy_hankel_recovery}, in time $\poly(\log n,\log(1/\epsilon))$, finds a $k=O(\log n \log(1/\epsilon))$ rank Hankel matrix $\wh{H}$ (in a compressed representation) such that $\|H-\wh{H}\|_F \leq O(\|E\|_F) + \epsilon \|H\|_F$ holds with probability at least $0.9$. Moreover, $\wh{H} = VD'V^T + \wh{H}'$ where $D'$ is a rank $k$ diagonal matrix and $\wh{H}'$ is a Hankel matrix with only $O(k)$ nonzero entries. 
\end{reptheorem}

We will build up to the proof of Theorem \ref{thm:const_factor_approx_main} using various intermediate lemmas. Consider $T, V, D^*, H^*$ as per Lemma \ref{lem:existence_of_vandermonde_good_low_rank_approximation}, let $B=H+E = VD^*V^T + H^* + E + E_H$ where $E_H=H- (VD^*V^T + H^*)$, and let $R$ be the row order reversal operator (\Cref{def:row_reversal}). Recall as per Lemma \ref{lem:existence_of_vandermonde_good_low_rank_approximation} Hankel $H^*$ only has the first and last $n_0=O(\log(1/\epsilon))$ many nonzero anti-diagonals, In Algorithm \ref{alg:noisy_hankel_recovery} we estimate $H^*$ and diagonal $D^*$ in $VD^*V^T$ separately. We will estimate $H^*$ by considering the first and last $n_0$ many nonzero anti-diagonals of $B$, and then, while estimating the diagonal matrix $D^*$ of $VD^*V^T$ we will have to \emph{mask out} these first and last $n_0$ many anti-diagonals. Thus, we need the following definition of an anti-diagonal Hankel mask.
\begin{definition}[Anti-diagonal mask]\label{def:antidiagonal_hankel_mask}
Let $M\in \mathbb{R}^{n\times n}$ be the mask matrix that masks out the first and last $n_0=C\log(1/\epsilon)$ ($C$ as per Lemma \ref{lem:existence_of_vandermonde_good_low_rank_approximation}) anti-diagonals defined as
\begin{equation}\label{eq:def_mask_m}
    M_{i,j} = \begin{cases}
        1 \text{ for all } i,j\in [n] \text{ s.t. } i+j\in [n_0,2n-n_0] \\
         0 \text{ otherwise}.
    \end{cases} .
\end{equation} 
Let $M^C$ be the complement mask defined as $M^C = \mathbbm{1}_n\mathbbm{1}_n^T -M$, for the all ones matrix $\mathbbm{1}_n\mathbbm{1}_n^T$.
\end{definition}
Equipped with the above notation and high-level description, Algorithm \ref{alg:noisy_hankel_recovery} is as follows.

\begin{algorithm}[H]
		\caption{\textsc{NoisyHankelRecovery}}
		\label{alg:noisy_hankel_recovery}
		\begin{algorithmic}[1]
			
			\STATE $\textbf{Input:}$ Query access to $B=H+E\in \mathbb{R}^{n\times n}$ for PSD Hankel $H$ and arbitrary noise $E$, $\epsilon>0$.
			\STATE $\textbf{Init:}$ Set $n_0=C\log(1/\epsilon)$ for $C$ as per Lemma \ref{lem:existence_of_vandermonde_good_low_rank_approximation}, Anti-diagonal mask $M$ (\Cref{def:antidiagonal_hankel_mask}).
			\STATE  Consider $B_1 = M^C\circ B$ and $B_2 = M\circ B$.
			\STATE Let Hankel $\wh{H}_1$ obtained by averaging the first and last $n_0$ anti-diagonals of $B_1$. That is, for all $m\in [n_0]$ define the $m^{th}$  anti-diagonal entries $(\wh{H}_1)_{i,m-i}$ for every $i\in [m]$ as follows,
            \begin{equation*}
                (\wh{H}_1)_{i,m-i} = (\sum_{j\in [m]}(B_1)_{j,m-j})/m,
            \end{equation*}
            and for all $m\in [2n-n_0,2n]$ the $m^{th}$ anti-diagonal entries defined similarly.
			
			\STATE Let $S_1,S_2\in \mathbb{R}^{s\times}$ be sampling matrices drawn as per Lemma \ref{lem:vandermonde_ridge_lev_score_sampling_guarantee} for ridge parameter $\gamma = \epsilon^2/\poly(n)$, and let $V\in \mathbb{R}^{n\times 2|T|}$ as per Lemma \ref{lem:existence_of_vandermonde_good_low_rank_approximation}.
			
			\STATE Let $D'(\gamma)$ be obtained from solving the following sketched ridge regression problem,
            \begin{equation*}
    D'(\gamma) = \argmin_{\text{diag } D\in \mathbb{R}^{2|T|\times 2|T|}} \|S_1 (M\circ VD V^T) S_2^T - S_1 B_2 S_2^T \|_F^2 + \gamma \|D\|_F^2,
\end{equation*}
			\STATE Return $\wh{H}_1,D'(\gamma)$.
		\end{algorithmic}
	\end{algorithm}

Let $B_1 = M^C \circ B$ and $B_2 = M \circ  B$ be as per line 3 of Algorithm \ref{alg:noisy_hankel_recovery}. By definition, $M\circ H^* = 0$. That is, $H^*$ only contributes to the first and last $n_0$ many anti-diagonals of $B$. Thus, we have that,
\begin{align*}
B_1 &= M^C\circ (VD^*V^T+H^*) + M^C\circ (E+E_H),\numberthis \label{eq:b1_expression} \\
B_2 &= M \circ V D^*V^T + M\circ (E + E_H). \numberthis \label{eq:b2_expression}
\end{align*}
To prove the guarantees of Algorithm \ref{alg:noisy_hankel_recovery} in Theorem \ref{thm:const_factor_approx_main}, we will present intermediate lemmas providing guarantees of various lines of Algorithm \ref{alg:noisy_hankel_recovery}. First, we will approximate $B_1$, the top left and bottom right $n_0$ many anti-diagonals of $B$, with a Hankel matrix. We show this can be done by \emph{averaging the anti-diagonals} of $B_1$ as done in line 4 of Algorithm \ref{alg:noisy_hankel_recovery}.
\begin{lemma}[Estimating the first and last $O(\log(1/\epsilon))$ anti-diagonals]\label{lem:learning_anti_diagonals_top_left_bottom_right}
In time $O(\log(1/\epsilon)^2)$ we can find a Hankel $\wh{H}_1$ of rank $O(\log(1/\epsilon))$ such that $M\circ \wh{H}_1=0$ and,
\begin{equation*}
    \|\wh{H}_1 - B_1\|_F^2 \leq  \|M^C\circ (E+E_H)\|_F^2.
\end{equation*}
This procedure corresponds to Line 4 of Algorithm \ref{alg:noisy_hankel_recovery}.
\end{lemma}
\begin{proof}
Note that by definition of mask $M$ as per \eqref{eq:def_mask_m}, $(B_1)_{i,j}\neq 0$ only for $i,j\in [n]$ such that $i+j\leq n_0$ or $i+j\geq 2n-n_0$. Consider any Hankel matrix $H\in \mathbb{R}^{n\times n}$ which is only defined on the first and last $n_0$ many anti-diagonals. That is, $M\circ H= 0$. Let $H(m) = (H)_{i,m-i}$ for all $i\in [m]$ and $m\in [2m]$ be the value of $H$ on its $m^{th}$ anti-diagonal. Then the error $\|H-B_1\|_F$ is 
\begin{align*}
    \|H-B_1\|_F^2 &= \sum_{m\in [n_0]} \sum_{i\in [m]} (H-B_1)_{i,m-i}^2+\sum_{m\in [n_0]} \sum_{i\in [m]} (R(H-B_1)R^T)_{i,m-i}^2\\
    &= \sum_{m\in [n_0]} \sum_{i\in [m]} (H(m) - (B_1)_{i,m-i})^2+\sum_{m\in [n_0]} \sum_{i\in [m]} ((RHR^T)(m)-(RB_1R^T)_{i,m-i}))^2.
\end{align*}
Let $\wh{H}_1 $ defined as the optimal $H$ for the above error follows $\wh{H}_1 = \argmin_{\text{Hankel }H: M\circ H =0}\|H-B_1\|_F^2$. Then by taking the derivative of the error it is easy to see that $\wh{H}_1$ is given by averaging the anti-diagonals of $B_1$. That is, $\wh{H}_1(m) = (\sum_{i\in [m]}(B_1)_{i,m-i})/m$ and $\wh{H}_1(2n-m) = (\sum_{i\in [m]}(RB_1R^T)_{i,m-i})/m$ for all $m\in [n_0]$. Moreover, $\wh{H}_1$ satisfies the following due to optimality:
\begin{equation*}
    \|\wh{H}_1 - B_1\|_F^2 \leq \|M^C\circ (VD^*V^T+H^*) - B_1\|_F^2= \|M^C\circ (E+E_H)\|_F^2.
\end{equation*}
Also note that $\wh{H}_1$ can be computed in $O(n_0^2)= O(\log^2(1/\epsilon))$ time by computing the averages, and rank of $\wh{H}_1 = O(\log(1/\epsilon))$ as it has only $O(\log(1/\epsilon))$ nonzero rows. This is presented in line 4 of Algorithm \ref{alg:noisy_hankel_recovery}.
\end{proof}
Now we focus on finding a good Hankel low rank approximation to $B_2$ defined as per \eqref{eq:b2_expression}. This is done in lines 5 and 6 of Algorithm \ref{alg:noisy_hankel_recovery}. Observe that if we can obtain a Hankel matrix $\wh{H}_2$ such that $M^C \circ \wh{H}_2 = 0$, rank of $\wh{H}_2$ is at most $O(\log n \log(1/\epsilon))$ and satisfying
\begin{equation}\label{eq:masked_regression_target_guarantee}
    \|\wh{H}_2 - B_2\|_F^2 \leq \alpha \|M\circ (E+E_H)\|_F^2+ \epsilon \|H\|_F^2,
\end{equation}
for some $\alpha>0$, then we have the following error guarantee for $\wh{H}_1+\wh{H}_2$ for $\wh{H}_1$ as per Lemma \ref{lem:learning_anti_diagonals_top_left_bottom_right}. Note $\wh{H}_1+\wh{H}_2$ is also a Hankel matrix of rank at most $O(\log n \log(1/\epsilon))$.
\begin{align*}
    \|\wh{H}_1+\wh{H}_2- B_1-B_2\|_F^2 &= \|\wh{H}_1-B_1\|_F^2+ \|\wh{H}_2-B_2\|_F^2\\
    &\leq \|M^C\circ (E+E_H)\|_F^2 +\alpha \|M\circ (E+E_H)\|_F^2\\
    &\leq \alpha \|E+E_H\|_F^2 ,
\end{align*}
where in the first equality we used the fact that the supports of $\wh{H}_1-B_1$ and $\wh{H}_2 - B_2$ are disjoint by definition of the mask $M$. Hence, for any $\alpha = O(1)$ we have
\begin{align*}
    \|\wh{H}_1+\wh{H}_2- B\|_F&\leq \alpha \|E\|_F + \alpha  \|E_H\|_F\\
    &=\alpha \|E\|_F + \alpha  \|H-(VD^*V^T+H^*)\|_F\\
    &\leq \alpha \|E\|_F+ \epsilon \alpha\|H\|_F\leq O(\|E\|_F)+ \epsilon \|H\|_F,
\end{align*}
where in the last line we used the bound of Lemma \ref{lem:existence_of_vandermonde_good_low_rank_approximation} to bound $\|H-(VD^*V^T+H^*)\|_F$ and rescaled $\epsilon$ by a constant.
We now present the following lemma, which achieves the goal of finding such a low rank Hankel $\wh{H}_2$ satisfying the target of \eqref{eq:masked_regression_target_guarantee} for some $\alpha = O(1)$.
\begin{lemma}[Estimating the remaining anti-diagonals]\label{lem:learning_anti_diagonals_remaining}
In time $\polylog(n,1/\epsilon)$  lines 5 and 6 of Algorithm \ref{alg:noisy_hankel_recovery} find a diagonal $D'$ such that the following holds with probability 0.99,
\begin{equation*}
    \|M\circ VD'V^T- B_2\|_F^2\leq 1000 \|M\circ(E+E_H)\|_F^2 + \epsilon \|H\|_F^2.
\end{equation*}
\end{lemma}
\begin{proof}
Since we know that $B_2 = M\circ VD^*V^T+M\circ (E+E_H)$ from \eqref{eq:b2_expression}, and $V = [V_T; R V_T]\in \mathbb{R}^{n\times 2|T|}$ with $T$ a fixed set as per Lemma \ref{lem:bucket_sparsification}, we phrase our objective to be to find a diagonal $D\in \mathbb{R}^{2|T|\times 2|T|}$ as per the following matrix ridge regression objective for some carefully chosen ridge parameter $\gamma >0$,
\begin{equation}\label{eq:masked_regression_objective}
    D(\gamma) = \argmin_{\text{diag }D\in\mathbb{R}^{2|T|\times 2|T|}} \| M\circ VDV^T-B_2\|_F^2 + \gamma \|D\|_F^2.
\end{equation}
Before we discuss how to solve the above efficiently in sublinear time, we first discuss why approximately solving the above suffices for our purposes. That is, to find a low rank $\wh{H}_2 $ that satisfies \eqref{eq:masked_regression_target_guarantee} for some $\alpha = O(1)$.

\smallskip

\noindent\textbf{Bounding the cost of the optimal solution.} Clearly rank of $ VDV^T$ for any diagonal $D\in \mathbb{R}^{2|T|\times 2|T|}$ is at most $2|T| = O(\log n \log(1/\epsilon))$ (bound on $|T|$ follows from Lemma \ref{lem:bucket_sparsification}). Moreover, because $M\circ VDV^T$ is differs from $VDV^T$ in at most $n_0 = O(\log(1/\epsilon))$ rows as per Definition \ref{def:antidiagonal_hankel_mask}, rank of $M\circ VDV^T$ is at most rank of $VDV^T$ plus $O(\log(1/\epsilon))$ which is $O(\log n\log(1/\epsilon))$. Moreover, $M\circ VDV^T$ is Hankel as it is the entrywise product of two Hankel matrices. And finally by optimality, the solution $D(\gamma)$ of the objective in \eqref{eq:masked_regression_objective} satisfies the target guarantee of \eqref{eq:masked_regression_target_guarantee} with $\alpha=1$ shown as follows,
\begin{align*}
      \| M\circ VD(\gamma)V^T &-B_2\|_F^2+\gamma\|D(\gamma)\|_F^2 \leq  \| M\circ VD^*V^T-B_2\|_F^2 + \gamma \|D^*\|_F^2\\& = \|M\circ (E+E_H)\|_F^2+\gamma \|D^*\|_F^2\\
    &\leq \|M\circ (E+E_H)\|_F^2 + \gamma (\sum_{i\in [2|T|]}|D^*_{i,i}|)^2 \quad (\text{since $D^*$ diagonal)}\\
    &\leq \|M\circ (E+E_H)\|_F^2 + O(\gamma\log^2 n\log^2(1/\epsilon)) \|H\|_F^2 \quad (\text{using Lemma \ref{lem:existence_of_vandermonde_good_low_rank_approximation}})
    \numberthis \label{eq:opt_bound}.
\end{align*}
Thus ,solving \eqref{eq:masked_regression_objective} with $\gamma = \epsilon^2/(\log^2(n)\log^2(1/\epsilon))$, we have
\begin{equation*}
\| M\circ VD(\gamma)V^T -B_2\|_F^2\leq \| M\circ VD(\gamma)V^T -B_2\|_F^2 +\gamma \|D(\gamma)\|_F^2 \leq   \|M\circ (E+E_H)\|_F^2+ \epsilon^2 \|H\|_F^2,
\end{equation*}
which satisfies \eqref{eq:masked_regression_target_guarantee} for $\alpha=1$. Hence, our task is reduced to finding a constant factor approximate solution to \eqref{eq:masked_regression_objective} for any given $\gamma>0$ in sublinear time. We will go about it using sampling according to the ridge leverage scores of $V$.

\smallskip

\noindent\textbf{Ridge spectral approximation guarantees for two-sided sketching.} We will show that Lemmas \ref{lem:vandermonde_ridge_lev_score_bd} and \ref{lem:vandermonde_ridge_lev_score_sampling_guarantee} are applicable to $V$ -- Recall $V = [V_T;RV_T]\in \mathbb{R}^{2|T|\times 2|T|}$, where $V_T$ is the Vandermonde matrix corresponding to the set of exponentiated Chebyshev nodes $T$ (recall $T$ is as per Lemma \ref{lem:x_0_1_errs} in  \Cref{sec:bucket_sparsification}) and $R$ is the row order reversal operator (Definition \ref{def:row_reversal}). This implies that if we let $T^{'} = \{1/x\}_{x\in T}$ and $D_{T^{'}} = diag(\{x^{n-1}\}_{x\in T^{'}})$, then $R V_T = V_{T^{'}} D_{T^{'}}^{-1}$. This implies that for $D_{T\cup T^{'}} =diag(\{1\}_{x\in T}\cup\{x^{n-1}\}_{x\in T^{'}}) $, $V$ is of the form $V = V_{T\cup T^{'}}D_{T\cup T^{'}}^{-1}$ where $V_{T\cup T^{'}}$ is the Vandermonde matrix corresponding to set $T\cup T^{'}$. Now consider applying Lemma \ref{lem:vandermonde_ridge_lev_score_sampling_guarantee} for $X= T\cup T^{'}$ with corresponding Vandermonde matrix $V_X = V_{T\cup T^{'}}$ and $n_0,\gamma$ and failure probability $\delta$ to obtain sampling matrix $S\in \mathbb{R}^{s\times n}$ for $s=O(\log^4(n/\gamma)\log(1/\delta))$. Then we have the following for all $i\in [n_0]$ and $y\in \mathbb{R}^{2|T|} $:
\begin{align*}
    \|SM_i\circ Vy\|_2^2+\gamma \|y\|_2^2 &=\left(1\pm \frac{1}{4}\right) \left(\|M_i\circ Vy\|_2^2 +\gamma \|y\|_2^2\right),\numberthis \label{eq:s_guarantee_masked_vandermonde_1_vectorwise}\\
\|SRM_i\circ Vy\|_2^2+\gamma \|y\|_2^2 &=\left(1\pm \frac{1}{4}\right)\left(\|RM_i\circ Vy\|_2^2 +\gamma \|y\|_2^2\right),\numberthis \label{eq:s_guarantee_masked_vandermonde_2_vectorwise}
\end{align*}
where $M_i\in \mathbb{R}^{n\times 2|T|}$ is the row mask matrix of Definition \ref{def:row_mask} that zeroes out the first $i$ rows of $V\in \mathbb{R}^{n\times 2|T|}$.

Suppose $S_1,S_2$ are sampled as per the distribution of $S$ as above.  Let $M_{c(i)}$ be the mask corresponding to the $i^{th}$ column, i.e., $M_{c(i)} = M_i $ for $i\leq n_0$, $M_{c(i)}=M_0$ for $i \in [n_0,n-n_0]$ and $M_{c(i)} = RM_{n-i}$ for $i\in [n-n_0,n]$. If we consider the columns of $M\circ VDV^T$ for any diagonal $D\in \mathbb{R}^{2|T|\times 2|T|}$, then by definition of mask $M$ (see \eqref{eq:def_mask_m}) and $M_{c(i)}$ for all $i\in [n]$ the $i^{th}$ column of $M\circ VDV^T$ is of the form $M_{c(i)} \circ Vy$ for some $y\in \mathbb{R}^{2|T|}$ for all $i\in [n]$. Since $M\circ VDV^T$ is symmetric, the same holds for its rows. Furthermore, since $S_2$ is a sampling matrix, $(M\circ VDV^T)S_2^T$ is a matrix whose columns are a subset of scaled columns of $M\circ VDV^T$.

Thus, applying equations \eqref{eq:s_guarantee_masked_vandermonde_1_vectorwise} and \eqref{eq:s_guarantee_masked_vandermonde_2_vectorwise} to all columns of $M\circ VDV^T$ and rows of $(M\circ VDV^T)S_2^T$ and summing them up, we have, with probability $1-\delta$ for all $D$,
\begin{align*}
    \|S_1(M\circ VDV^T)S_2^T\|_F^2+\gamma \|DV^TS_2^T\|_F^2 &=\left(1\pm \frac{1}{4}\right)\left(\|(M\circ VDV^T)S_2^T\|_F^2 +\gamma \|DV^TS_2^T\|_F^2\right), \numberthis \label{eq:s_guarantee_masked_vandermonde_1}\\
    \|(M\circ VDV^T)S_2^T\|_F^2+\gamma \|VD\|_F^2 &=\left(1\pm \frac{1}{4}\right) \left(\|M\circ VDV^T\|_F^2 +\gamma \|VD\|_F^2\right).\numberthis \label{eq:s_guarantee_masked_vandermonde_2}
\end{align*}
Now for any diagonal $D$, since $S_2$ is a sampling matrix as per Lemma \ref{lem:vandermonde_ridge_lev_score_sampling_guarantee} we have that each entry of $S_2$ is at most $\poly(n)$, thus $\|DV^TS_2^T\|_F^2 \leq \poly(n) \|DV^T\|_F^2$. Furthermore, since each entry in $V$ is at most $1$ and $D$ diagonal, we have that $\|DV^T\|_F^2\leq  \poly(n)\|D\|_F^2$ and $\|VD\|_F^2\leq  \poly(n)\|D\|_F^2$. Applying these bounds after rescaling $\gamma$ by $\poly(n)$ and combining equations \eqref{eq:s_guarantee_masked_vandermonde_1} and \eqref{eq:s_guarantee_masked_vandermonde_2}, we have, with probability at least $1-\delta$ for all $D$,
\begin{equation}\label{eq:s_guarantee_masked_vandermonde_3}
    \|S_1(M\circ VDV^T)S_2^T\|_F^2 = \left(1\pm \frac{1}{2}\right)\|M\circ VDV^T\|_F^2 \pm \gamma \|D\|_F^2.
\end{equation}

We will consider diagonal $D'\in \mathbb{R}^{2|T|\times 2|T|}$ obtained by solving the following sketched regression problem, 
\begin{equation}\label{eq:d_prime_def_sketched_regression}
    D'(\gamma) = \argmin_{\text{diag } D\in \mathbb{R}^{2|T|\times 2|T|}} \|S_1 (M\circ VD V^T) S_2^T - S_1 B_2 S_2^T \|_F^2 + \gamma \|D\|_F^2,
\end{equation}
and we will show that $D'(\gamma)$ achieves a constant factor approximation to the objective defined in \eqref{eq:masked_regression_objective}.
This is the procedure of lines 5 and 6 in Algorithm \ref{alg:noisy_hankel_recovery}. Note $D'(\gamma)$ can be obtained in $\poly(s)=\polylog(n,1/\gamma)$ time.

\smallskip

\noindent\textbf{Analyzing the optimal solution of sketched ridge regression.} In the following we denote $D'(\gamma)$ with $D'$ to ease notation, and re-introduce the notation $\gamma$ at the end of the analysis.
First, observe from the triangle inequality that 
\begin{equation*}
    \|S_1(M\circ VD'V^T)S_2^T - S_1 B_2 S_2^T\|_F = \|S_1(M\circ V(D'-D^*)V^T)S_2^T\|_F \pm  \|S_1(M\circ VD^{*}V^T)S_2^T - S_1 B_2 S_2^T\|_F.
\end{equation*}
Since $S_1$ and $S_2$ are both unbiased sampling matrices and $D^*$ does not depend on them, by applying Markov's inequality twice over randomness in $S_1,S_2$ and taking a union bound, we have, with probability at least $0.99$,
\begin{equation}\label{eq:unbiasedness_applied_to_d_star}
    \|S_1(M\circ VD^{*}V^T)S_2^T - S_1 B_2 S_2^T\|_F^2\leq 100 \|M\circ VD^{*}V^T - B_2 \|_F^2.
\end{equation}
Combining with the above bound after squaring both sides we have
\begin{align*}
    \|S_1(M\circ VD'V^T)S_2^T - S_1 B_2 S_2^T\|_F^2 &\leq 2\|S_1(M\circ V(D'-D^*)V^T)S_2^T\|_F^2 + 200\|M\circ VD^{*}V^T-  B_2 \|_F^2,\\
     \|S_1(M\circ VD'V^T)S_2^T - S_1 B_2 S_2^T\|_F^2 &\geq \frac{1}{2}\|S_1(M\circ V(D'-D^*)V^T)S_2^T\|_F^2 -  100\|M\circ VD^{*}V^T-  B_2 \|_F^2.
\end{align*}
Now consider the application of the guarantee of \eqref{eq:s_guarantee_masked_vandermonde_3} to $D=D'-D^*$ to get with probability at least $0.99$,
\begin{align*}
    \|S_1(M\circ V(D'-D^*)V^T)S_2^T\|_F^2 = \left(1\pm \frac{1}{2}\right)\|M\circ V(D'-D^*)V^T\|_F^2 \pm \gamma (\|D'\|_F^2+\|D^*\|_F^2).
\end{align*}
Plugging this into the equation above 
, we obtain
\begin{align*}
 \|S_1(M\circ VD'V^T)S_2^T - S_1 B_2 S_2^T\|_F^2 &\leq 3\|M\circ V(D'-D^*)V^T\|_F^2 + 200\|M\circ VD^{*}V^T-  B_2 \|_F^2\\
    &+\gamma (\|D'\|_F^2+\|D^*\|_F^2),\\
     \|S_1(M\circ VD'V^T)S_2^T - S_1 B_2 S_2^T\|_F^2 &\geq \frac{1}{4}\|M\circ V(D'-D^*)V^T\|_F^2 -  100\|M\circ VD^{*}V^T-  B_2 \|_F^2\\&-\gamma (\|D'\|_F^2+\|D^*\|_F^2)).
\end{align*}
Now again by triangle inequality and squaring both sides we have
\begin{align*}
    \|M\circ V(D'-D^*)V^T\|_F &= \|M\circ VD'V^T- B_2\|_F \pm \|M\circ VD^*V^T- B_2\|_F\\
    \implies \|M\circ V(D'-D^*)V^T\|_F^2 &\leq 2\|M\circ VD'V^T- B_2\|_F^2  +2\|M\circ VD^*V^T- B_2\|_F^2\\
    \text{ and }\|M\circ V(D'-D^*)V^T\|_F^2 &\geq \frac{1}{2}\|M\circ VD'V^T- B_2\|_F^2  -\|M\circ VD^*V^T- B_2\|_F^2.
\end{align*}
Plugging this into the previous equation we get
\begin{align*}
 \|S_1(M\circ VD'V^T)S_2^T - S_1 B_2 S_2^T\|_F^2 &\leq 6\|M\circ VD'V^T-B_2\|_F^2  + 300\|M\circ VD^{*}V^T-  B_2 \|_F^2\\
    &+\gamma (\|D'\|_F^2+\|D^*\|_F^2),\\
     \|S_1(M\circ VD'V^T)S_2^T - S_1 B_2 S_2^T\|_F^2 &\geq \frac{1}{8}\|M\circ VD'V^T-B_2\|_F^2 -  200\|M\circ VD^{*}V^T-  B_2 \|_F^2\\&-\gamma (\|D'\|_F^2+\|D^*\|_F^2))\numberthis\label{eq:two_sided_sketching_guarantee}.
\end{align*}
From Lemma \ref{lem:existence_of_vandermonde_good_low_rank_approximation} we can upper bound as $\gamma \|D^*\|_F^2\leq \gamma \|D^*\|_1^2 \leq O(\gamma \log(1/\epsilon)\log(n)\|H\|_F)$, which is at most $\gamma \|H\|_F$ after rescaling $\gamma$ by $n$. Now to bound $\gamma \|D'\|_F^2$ we use the optimality of $D'$ as per its definition in \eqref{eq:d_prime_def_sketched_regression} as follows:
\begin{align*}
    \gamma \|D'\|_F^2&\leq \|S_1 (M\circ VD' V^T) S_2^T - S_1 B_2 S_2^T \|_F^2 + \gamma \|D'\|_F^2\\
    &\leq \|S_1( M\circ VD^* V^T) S_2^T - S_1 B_2 S_2^T \|_F^2 + \gamma \|D^*\|_F^2 \quad\text{(from optimality of $D'$ in \eqref{eq:d_prime_def_sketched_regression})} \\
    &\leq 100\|M\circ VD^* V^T -  B_2 \|_F^2 + 2\gamma\|D^*\|_F^2 \quad \text{(from \eqref{eq:unbiasedness_applied_to_d_star})}.\numberthis\label{eq:d_prime_norm_bound}
\end{align*}
Finally, equipped with the guarantees above, we bound the performance of $D'=D'(\gamma)$ as an approximate solution for the original masked Vandermonde regression objective as defined in \eqref{eq:masked_regression_objective}.
\begin{align*}
    \|M\circ VD'(\gamma) V^T  -  B_2 \|_F^2&\leq \|M\circ VD'(\gamma) V^T  -  B_2 \|_F^2+ \gamma \|D'(\gamma)\|_F^2 \\&\leq 8\|S_1M\circ VD'(\gamma) V^TS_2^T  -  S_1B_2 S_2^T\|_F^2\quad \text{(from \eqref{eq:two_sided_sketching_guarantee})} \\
    &+(1600\|M\circ VD^{*}V^T-  B_2 \|_F^2+8\gamma (\|D'(\gamma)\|_F^2+\|D^*\|_F^2))\\
    &\leq 10\|S_1M\circ VD^{*} V^TS_2^T  -  S_1B_2 S_2^T\|_F^2 \quad \text{(as $D'$ optimal, see  \eqref{eq:d_prime_def_sketched_regression})}\\ &+ 1600\|M\circ VD^{*} V^T -  B_2 \|_F^2+8\gamma (\|D'(\gamma)\|_F^2+\|D^*\|_F^2)\\
    &\leq 1600\|M\circ VD^{*} V^T -  B_2 \|_F^2 + \gamma \|D^*\|_F^2 \text{( from  \eqref{eq:unbiasedness_applied_to_d_star} and \eqref{eq:d_prime_norm_bound})}\\
    &\leq 1600\|M\circ (E+E_H)\|_F^2+ \epsilon^2\|H\|_F^2,
\end{align*}
where the last line is obtained similarly to \eqref{eq:opt_bound} and $\gamma$ is set to $\gamma = \epsilon^2/\poly(n)$. All of the above holds with probability at least $0.9$ after a union bound, and thus this completes the proof of the lemma.
\end{proof}
We now give the proof of Theorem \ref{thm:const_factor_approx_main}, which follows easily from the previous lemmas.

\begin{proof}[Proof of Theorem \ref{thm:const_factor_approx_main}]
Let $\wh{H}_1$ and $\wh{H}_2=M\circ VD'V^T$ be as per Lemmas \ref{lem:learning_anti_diagonals_top_left_bottom_right} and \ref{lem:learning_anti_diagonals_remaining}. Let $\wh{H} = \wh{H}_1+\wh{H}_2$. Clearly rank of $\wh{H}$ is at most the rank of $\wh{H}_1$ plus the rank of $\wh{H}_2$, and thus it is at most $O(\log n \log(1/\epsilon)$ 
Moreover, $\wh{H}$ is Hankel as it is the sum of Hankel matrices. Finally, $\wh{H}$ satisfies, with  probability at least $0.9$,
\begin{align*}
    \|\wh{H} - B\|_F^2 &= \|\wh{H}_1-B_1\|_F^2 + \|\wh{H}_2-B_2\|_F^2\quad \text{(due to disjoint supports)}\\
    &\leq \|M^C\circ (E+E_H)\|_F^2 + 1600\|M\circ (E+E_H)\|_F^2 + \epsilon^2 \|H\|_F^2\quad \text{(from Lemmas \ref{lem:learning_anti_diagonals_top_left_bottom_right} and \ref{lem:learning_anti_diagonals_remaining})}\\
    &\leq 1600\|E+E_H\|_F^2 + \epsilon^2 \|H\|_F^2\\
    &\leq 1600 \|E\|_F^2 + 200 \|E_H\|_F^2 + \epsilon^2 \|H\|_F^2\\
    & \leq 1600 \|E\|_F^2 + 2\epsilon^2 \|H\|_F^2,
\end{align*}
where in the final line we used the fact that $\|E_H\|_F^2=\|VD^*V^T + H^* - H\|_F^2\leq \epsilon^2 \|H\|_F$ as per Lemma \ref{lem:existence_of_vandermonde_good_low_rank_approximation}.
Hence $\|\wh{H} - B\|_F\leq 100\|E\|_F+ \epsilon \|H\|_F$ with probability $0.9$. Moreover, the runtime of finding $\wh{H}$ is $\poly(\log n,\log(1/\epsilon))+O(\log^2(n/\epsilon))$ as per Lemmas \ref{lem:learning_anti_diagonals_top_left_bottom_right} and \ref{lem:learning_anti_diagonals_remaining}. Thus the overall runtime is $\poly(\log n,\log(1/\epsilon))$ 
, and the rank of $\wh{H}$ as per its definition in Lemma \ref{lem:existence_of_vandermonde_good_low_rank_approximation} is $O(\log n \log(1/\epsilon))$. Expressing $\wh{H} = VD'V^T + (\wh{H_2}-M^C \circ VD'V^T)$ concludes the proof of the lemma.
\end{proof}

\section{Applications.}\label{sec:applications}
Finally, we detail several applications of our main results. 

\smallskip

\noindent\textbf{Fast polynomial basis transforms.}  Townsend, Webb, and Olver \cite{TownsendWebbOlver:2018} consider the problem of transforming the coefficients of a degree $n$ polynomial in one orthogonal polynomial basis (e.g., Chebyshev) to another (e.g., Legendre). When the input polynomial is represented as a coefficient vector $v \in \R^n$, this problem can be formulated as computing a matrix-vector product $Mv$ where $M$ is the change of basis matrix. 

\cite{TownsendWebbOlver:2018} begins by observing that many important basis conversion matrices can be written as $D_1(T\circ H)D_2$ where $D_1,D_2$ are diagonal matrices, $T$ is Toeplitz, $H$ is PSD Hankel, and $\circ$ denotes the Hadamard (entrywise) matrix product. Their algorithm exploits this decomposition, developing a fast approximate matrix-vector product primitive for such structured matrices. Let $\epsilon$ be an error parameter, and consider the regime when $\epsilon=1/\poly(n)$. Let $k=O(\log n \log(1/\epsilon)) = O(\log^2n)$. Their algorithm first computes a rank $k$ approximation of $H$ denoted by $\sum_{r=1}^k a_r l_r l_r^T$ with entrywise error $\epsilon=1/\poly(n)$ in $O(n \log ^4 n)$ time. They then observe that the matrix vector product $D_1(T \circ (\sum_{r=1}^k a_r l_r l_r^T)) D_2 v$ can be computed using $k$ FFTs in time $O(n\log n \cdot k) = O(n\log^3 n)$. Thus, their overall runtime is $O(n\log^4 n)$, dominated by the cost of computing the low-rank approximation of $H$. 

Directly applying the algorithm of Theorem \ref{thm:const_factor_approx_main} with $\epsilon = 1/\poly(n)$ and so rank $k = O(\log^2 n)$, we can compute a compressed representation of an entrywise $1/\poly(n)$ approximation to $H$ of the form $VD'V+ \wh{H}'$ in  $O(\polylog (n))$ time, where $\wh{H}'$ is Hankel with only $O(k)$ nonzeros, $D' \in \R^{O(k) \times O(k)}$ is diagonal, and $V \in \R^{n \times O(k)}$ is Vandermonde. Expanding out the columns of $V$, we can explicitly write down a factorization of this matrix by writing each component in the sum $\sum_{r=1}^{O(k)} a_r l_r l_r^T$ in $O(nk) = O(n \log^2 n)$ time. Using this factorization in the algorithm if \cite{TownsendWebbOlver:2018}, we can compute $D_1(T \circ (\sum_{r=1}^k a_r l_r l_r^T)) D_2 v$ as before using $k$ FFTs in time $O(n\log^3 n)$. Thus, the overall runtime is dominated by this last multiplication step and is $O(n\log^3 n)$, improving the runtime of \cite{TownsendWebbOlver:2018} by a $\log n$ factor.

\smallskip

\noindent\textbf{Hankel covariance estimation.} Consider an $n$-dimensional Gaussian distribution $\mathcal{N}(0,H)$ where $H\in \mathbb{R}^{n\times n}$ is PSD Hankel. In covariance estimation, we seek to estimate $H$ efficiently given sample access to this distribution. We can do so following the approach of Theorem 3 in \cite{MuscoSheth:2024}, which focuses on Toeplitz covariance matrix estimation. First, one can see that Lemma 5.6 of \cite{MuscoSheth:2024} holds in our setting as well. That is, if we let $XX^T$ be the empirical covariance matrix formed from $s = \widetilde{O}(k^4/\epsilon^2)$ i.i.d. samples from $\mathcal{N}(0,H)$, where $X\in \mathbb{R}^{n\times s}$ contains the samples as its columns (rescaled by $1/\sqrt{s}$ to ensure the correct expection), we have, with probability at least $0.98$,
\begin{align}\label{eq:covarianceError}
	\|XX^T - H\|_F = O\left( \sqrt{\|H-H_k\|_2 \text{tr}(H)+ \frac{\|H-H_k\|_F \text{tr}(H)}{k}}+ \epsilon \|H\|_2\right).
\end{align}
Setting $\epsilon' = \epsilon/n^{3/2}$ and $k = O(\log n \log(1/\epsilon')) = O(\log n \log(n/\epsilon))$, we have (by the result of Beckermann and Townsend \cite{beckermann2017singular} or \Cref{thm:main_thm}) that $\|H-H_k\|_2 \le \|H-H_k\|_F \le \epsilon/n^{3/2} \cdot \|H\|_F \le \epsilon/n \cdot \norm{H}_2$. Thus, since $\mathrm{tr}(H) \le n \|H\|_2$, the right hand side of \eqref{eq:covarianceError} is bounded by $O(\epsilon \|H\|_2)$.

Now, applying the algorithm of Theorem \ref{thm:const_factor_approx_main} to $XX^T$, where the non-Hankel error is $E =XX^T - H$ and the rank parameter is $k=O(\log n \log (1/\epsilon')) = O(\log n \log(n/\epsilon)) $, we obtain  an approximation $\wh{H}$ to $H$ satisfying (after adjusting $\epsilon$ by a constant factor): $$\|H - \wh H\|_2 \le \|H - \wh H\|_F \le \epsilon \cdot \norm{H}_2.$$
The runtime of the algorithm is $\polylog (n,1/\epsilon)$, and the  algorithm reads just $\polylog (n,1/\epsilon)$ entries from $XX^T$, and thus requires reading just $\polylog (n,1/\epsilon)$ entries from each sample (i.e., from each column of $X$). The number of i.i.d. samples needed from $\mathcal N(0,H)$ is $s = \widetilde{O}(k^4/\epsilon^2) = \poly(\log n,1/\epsilon)$.

\smallskip

\noindent\textbf{Sum-of-Squares (SoS) decompositions of polynomials.} Finally, we discuss a potential application of our results to SoS decompositions of polynomials. We follow the setup of Section 3.3 of \cite{ghadiri2023symmetric}. Consider a univariate polynomial $p(x)$ of degree $2n$ with real coefficients. Then there exists a Hankel matrix $H\in \mathbb{R}^{n\times n}$ such that $p(x) = v_n(x)^T H v_n(x)$ for any $x$, where $v_n(x)$ is the moment vector as in \Cref{def:moment_vector}. Suppose $H$ is rank $k$ PSD, then we can write $H=BB^T$ for $B\in \mathbb{R}^{n\times k}$. Then $p(x)$ admits a sum of squares decomposition of the form $p(x)=\sum_{i=1}^{k}l_i(x)^2$, where $l_i$ is the polynomial whose coefficients are given by the entries in the $i^{th}$ column of $B$.

Now, suppose we have an algorithm that can compute $\wh{B}\in \mathbb{R}^{n\times k}$ such that $\|\wh{B}\wh{B}^T-H\|_F\leq  \|H\|_F/\poly(n)$. Then, for all $x\in [-1,1]$, since $\norm{v_n(x)}_2 \le \sqrt{n}$, we have that,
\begin{align}\label{eq:polynomial}
|v_n(x)^T \wh{B}\wh{B}^T v_n(x) -v_n(x)^T H v_n(x)|\leq  \|H\|_F/\poly(n).
\end{align}
If we let $\wh{l}_i$ be the polynomial with coefficients given by the $i^{th}$ column of $\wh{B}$ and let $\wh p(x) = \sum_{i\in [k]}\wh{l}_i(x)^2$, then $\wh p$ is an approximate SoS decomposition of $p$. In particular, by \eqref{eq:polynomial}, for all $x \in [-1,1]$, we have that $|\wh p(x) - p(x)| \le \|H\|_F/\poly(n)$.

Unfortunately, the algorithm of Theorem \ref{thm:const_factor_approx_main} does not quite output an approximation $\wh H$ of the form $\wh H = \wh B \wh B^T$. In particular, our $\wh H$ may not be exactly PSD, even though it is a highly accurate approximation to $H$, which is PSD. Extending our approach to output $\wh H$ that is exactly PSD is an interesting problem, and would open up further applications, such as the one described above. 
\section{Conclusion.}\label{sec:conclusion}

Our work leaves open several interesting questions, which we summarize below.
\begin{enumerate}
    \item Can the existence of a good structure preserving low-rank approximation be proven for non-PSD
Hankel matrices? Can a sublinear time algorithm be designed to recover good low-rank approximations to non-PSD Hankel matrices? Such a result would also apply to non-PSD Toeplitz matrices, simply via row-reversal.
\item Can the rank lower bound of Theorem \ref{thm:epsilon_rank_lower_bound} be improved to $\Omega(\log n\log(1/\epsilon))$, matching the rank upper bound of Theorem \ref{thm:main_thm} and Beckermann Townsend \cite{beckermann2017singular}?
\item Can a sublinear time algorithm be designed to achieve the error guarantee of Theorem \ref{thm:const_factor_approx_main} in the spectral norm -- i.e., $\norm{H - \wh H}_2 \le \epsilon \|\wh H\|$ for $\wh H$ with rank $O(\log n \log(1/\epsilon))$? Note that \Cref{thm:const_factor_approx_main} implies this bound when $\wh H$ has rank $O(\log n \log(n/\epsilon))$ since we can apply the theorem with $\epsilon' = \epsilon/\sqrt{n}$ and bound:
$$\norm{H-\wh H}_2 \le \norm{H-\wh H}_F \le \epsilon' \norm{H}_F \le \epsilon' \cdot \sqrt{n} \norm{H}_F \le \epsilon \norm{H}_2.$$
However, for fixed $\epsilon$, this approach yields rank $O(\log^2 n)$, rather than the optimal $O(\log n)$.
\end{enumerate}

\section{Acknowledgments.}
Cameron Musco was partially supported by NSF grants 2046235 and 2427362. The authors acknowledge Daniel Kressner and Alex Townsend for helpful discussions.

\bibliographystyle{alpha} 
\bibliography{refs}

\end{document}